\documentclass[lettersize,journal]{IEEEtran}
\usepackage{amsmath,amsfonts}
\usepackage{algorithmicx}
\usepackage[ruled]{algorithm2e}
\usepackage{array}
\usepackage[caption=false,font=normalsize,labelfont=sf,textfont=sf]{subfig}
\usepackage{textcomp}
\usepackage{stfloats}
\usepackage{url}
\usepackage{verbatim}
\usepackage{graphicx}
\usepackage{cite}
\usepackage{cases}
\usepackage{bm}
\usepackage{amssymb}
\usepackage{multirow}
\usepackage{cuted}
\usepackage{lipsum}
\usepackage{epstopdf}
\usepackage{threeparttable}
\usepackage{epsfig}
\usepackage{amsthm}
\usepackage{color}

\newtheorem{prop}{Proposition}

\newtheorem{lem}{Lemma}
\newtheorem{thm}{\bf Theorem}

\newenvironment{theorem}{\begin{thm}} {\end{thm}}
\newtheorem{rmk}{\bf Remark}

\newenvironment{remark}{\begin{rmk}} {\end{rmk}}

\begin{document}
\title{Fast Multi-grid Methods for Minimizing Curvature Energies}

\author{Zhenwei Zhang,
        Ke Chen,
        Ke Tang
        and Yuping Duan$^*$
\thanks{The work was partially supported by the National Natural Science Foundation of China (NSFC 12071345, 11701418). \emph{Asterisk indicates the corresponding author.}}
\thanks{Z. Zhang is with Center for Applied Mathematics, Tianjin University, Tianjin 300072, China. E-mail: zzw7548@163.com.}
\thanks{K. Chen is with Liverpool Centre of Mathematics for Healthcare and Centre for Mathematical Imaging Techniques, Department of Mathematical Sciences, University of Liverpool, Liverpool L69 7ZL, UK. E-mail: k.chen@liv.ac.uk}
\thanks{K. Tang is with the Guangdong Key Laboratory of Brain-Inspired Intelligent Computation, Department of Computer Science and Engineering, Southern University of
Science and Technology, Shenzhen 518055, China, and also with the Research Institute of Trustworthy Autonomous Systems, Southern University of Science and Technology, Shenzhen 518055, China. E-mail: tangk3@sustech.edu.cn.}
\thanks{Y. Duan is with Center for Applied Mathematics, Tianjin University, Tianjin 300072, China. E-mail: yuping.duan@tju.edu.cn.}
}

\markboth{Journal of \LaTeX\ Class Files,~Vol.~14, No.~8, August~2021}%
{Shell \MakeLowercase{\textit{et al.}}: A Sample Article Using IEEEtran.cls for IEEE Journals}


\maketitle

\begin{abstract}
The geometric high-order regularization methods such as mean curvature and Gaussian curvature, have been intensively studied during the last decades due to their abilities in preserving geometric properties including image edges, corners, and contrast. However, the dilemma between restoration quality and computational efficiency is an essential roadblock for high-order methods. In this paper, we propose fast multi-grid algorithms for minimizing both mean curvature and Gaussian curvature energy functionals without sacrificing accuracy for efficiency. Unlike the existing approaches based on operator splitting and the Augmented Lagrangian method (ALM), no artificial parameters are introduced in our formulation, which guarantees the robustness of the proposed algorithm. Meanwhile, we adopt the domain decomposition method to promote parallel computing and use the fine-to-coarse structure to accelerate convergence.
Numerical experiments are presented on image denoising, CT, and MRI reconstruction problems to demonstrate the superiority of our method in preserving geometric structures and fine details. The proposed method is also shown effective in dealing with large-scale image processing problems by recovering an image of size $1024\times 1024$ within $40$s, while the ALM method  requires around $200$s.
\end{abstract}

\begin{IEEEkeywords}
Mean curvature, Gaussian curvature, multi-grid method, domain decomposition method, image denoising, image reconstruction.
\end{IEEEkeywords}

\section{Introduction}

\IEEEPARstart{I}{mage} restoration is a fundamental task in image processing, which aims to recover the latent clean image $u$ from the observed noisy image $f:{\Omega}\rightarrow \mathbb{R}$ defined on an open bounded domain ${\Omega}\subset \mathbb{R}^2$.
The total variation (TV) proposed by Rudin, Osher, and Fatemi is the most successful regularization used for image denoising problem \cite{rudin1992nonlinear}, which minimizes the total lengths of all level sets of the image.
Although the Rudin-Osher-Fatemi model has been proven to effectively remove noise and preserve sharp edges, it also suffers from some unfavorable properties including loss of image contrast and staircase effect \cite{Yves2002}.
High-order regularization methods have been investigated for image restoration problems to overcome the drawback of TV regularization including the fourth-order partial differential equation (PDE) \cite{Lysaker2003Noise}, total generalized variation \cite{Bredies2010total}, Euler's elastica \cite{2011AA}, mean curvature \cite{Zhu2012image} and total roto-translational variation \cite{ChambolleP19} etc.

The curvature regularization methods achieved great success by minimizing curvature-dependent energies, which  are well-known for their good  geometric interpretability and strong priors in the continuity of edges and has been applied to various data processing tasks such as
image decomposition \cite{belyaev2018adaptive}, graph embedding \cite{pei2020curvature}, and missing data recovery \cite{dong2020cure} etc.
Considering the associated image surface characterized by $(x, f(x))$ for $x \in\Omega$, the image restoration problem is to find a piecewise smooth surface $(x, u(x))$ to approximate the noisy surface and simultaneously remove the outliers.
The curvature minimization problem can be formulated as follows
\begin{equation}\label{MCGC}
\min_{u\in V}~ \int_\Omega |\kappa(u)|d x+\frac{\alpha}{2}\int_\Omega(u-f)^2d x,
\end{equation}
where $\kappa(u)$ can be either the mean curvature or Gaussian curvature of the image surface, and $V$ is a function space. The definitions of mean curvature and Gaussian curvature are described in Table \ref{curvature used}. To the best of our knowledge, one has not identified the proper function to formulate problem \eqref{MCGC}, which should be a subset of $L^2(\Omega)$.

\begin{table*}[t]
  \centering
  \caption{The mean curvature and Gaussian curvature regularization terms can be used in the image restoration model \eqref{MCGC}.\label{curvature used}}
  \begin{threeparttable}
    \begin{tabular}{|c|c|c|c|}
 \hline
 Curvature & Definition& Description  & References\\
 \hline
\multirow{2}[0]{*}{Mean curvature}  & $\nabla\cdot\Big(\frac{\nabla}{\sqrt{|\nabla u|^2+1}}\Big)$&the mean curvature of the zero level set $z=u(x,y)$& \cite{Zhu2012image,2017Augmented,2010Multigrid,gong2019weighted}  \\
  & $\frac{1}{2}(\kappa_{\min}(u)+\kappa_{\max}(u))$ &the mean curvature of  image surface $(x,y,u(x,y))$ &  \cite{2017Curvature,Zhong2020image,wang2022efficient} \\
  \hline
\multirow{2}[0]{*}{ Gaussian curvature} &$  \frac{ \mbox{det} (\nabla^2 u) }{(1+|\nabla u|^2)^2} $ & the Gaussian curvature of the zero level set $z=u(x,y)$  & \cite{2009Analogue,2015Image} \\
 &  $\kappa_{\min}(u)\kappa_{\max}(u)$ & the Gaussian curvature of image surface $(x,y,u(x,y))$ &\cite{2017Curvature,Zhong2020image,wang2022efficient}\\
     \hline
    \end{tabular}
    \begin{tablenotes}
        \footnotesize
        \item[1] $\kappa_{\max}(u)$ and $\kappa_{\min}(u)$ denotes the two principal normal curvatures.
    \end{tablenotes}
    \end{threeparttable}
\vspace{-0.5cm}
\end{table*}

The minimization of curvature energies is more challenging, such that efficient algorithms for solving the model \eqref{MCGC} are still limited.
Originally, the gradient descent method \cite{Zhu2012image} was presented to solve the mean curvature model, which has to solve fourth-order nonlinear evolution equations.
Liu \emph{et al.} \cite{lu2016implementation} developed a fast numerical algorithm for solving the high-order variational models based on split Bregman iteration.
Zhu, Tai, and Chan \cite{2017Augmented} developed the augmented Lagrangian method (ALM) for the mean curvature model.
Brito-Loeza and Chen \cite{2010Multigrid} propose a  multi-grid algorithm for solving
the mean curvature model, which is based on an augmented Lagrangian formulation with a special linearized fixed point iteration.
The situation is even worse for Gaussian curvature minimization since no fast algorithms are developed yet.
There are not many studies of effective numerical algorithms for Gaussian curvature minimization. Alboul and van Damme \cite{2000Polyhedral} used the total absolute Gaussian curvature in the different contexts of connectivity optimization for the triangulated surfaces.
Gong and Sbalzarini \cite{gong2013local} proposed a variational model using local weighted Gaussian curvature as a regularization term, which was effectively solved by the closed-form solution.
Elsey and Esedoglu \cite{2009Analogue} introduced the Gaussian curvature regularization for surface processing as the natural analog of the total variation,  which was discretized on a triangulated surface for reducing the difficulty of solution.
Brito-Loeza and Chen \cite{2015Image} presented a two-step method based on vector field smoothing and gray level interpolation to solve two second-order PDEs instead of fourth-order PDEs.

Zhong, Yin and Duan \cite{Zhong2020image} formulated  the following curvature regularization model by minimizing either Gaussian curvature or mean curvature over image surfaces
\begin{equation}\label{wmsurface}
\min_{u}~\int_\Omega g(\kappa)\sqrt{1+|\nabla u|^2}d x +\frac\lambda2\int_\Omega(u-f)^2d x,
\end{equation}
where $g(\cdot)$ denotes a certain function of the curvature. Model \eqref{wmsurface} was regarded as a re-weighted minimal surface model and handled by the alternating direction method of multipliers (ADMM).
Although the curvature function can be explicitly evaluated using the current estimation, one nonlinear sub-minimization problem has to be computed by the Newton method resulting in rising computational costs.

Gong and Sbalzarini \cite{2017Curvature} developed the curvature filters by minimizing either Gaussian curvature or mean curvature to smooth the noisy images. Rather than solving the higher-order PDEs, the pixel-wise solutions were presented to find the locally developable and minimal surfaces, which give zero Gaussian curvature and zero mean curvature, respectively. The idea of curvature filters was further studied in \cite{Gong2019mean,gong2019weighted}.
However, the curvatures lack rigorous definitions, which limits their performance in real applications.
Besides, when combined with the data fidelity term, gradient descent was used to estimate the solution leading to the slow convergence.

The multi-grid method is a fast numerical method for solving large-scale linear and nonlinear optimization problems\cite{Nash2000A,2003An,Frohn2004Nonlinear,2006On} and has been successfully applied to image processing models.
Chen and Tai \cite{Chen2007A} proposed a nonlinear multi-grid method for the total variation minimization based on the coordinate descent method. Savage and Chen \cite{savage2005improved} presented a nonlinear multi-grid method based on the full approximation scheme for solving the total variation model. Chan and Chen \cite{Chan2006An} proposed a fast multilevel method using primal relaxations for the total variation image denoising and analyzed its convergence. Zhang \emph{et al.} \cite{Zhang2021ipi} developed a multi-level domain decomposition method for solving the total variation minimization problems, which used the piecewise constant functions to ensure fast computation. Tai, Deng and Yin \cite{Tai2021} proposed a multiphase image segmentation method by solving the min-cut minimization problem under the multi-grid method framework.
For the high-order model, Brito-Loeza and Chen \cite{2010Multigrid} presented a new multi-grid method based upon a stabilized fixed point method for dealing with the mean curvature model.
The nonlinear multi-grid method was applied to fourth-order models to accelerate the convergence in \cite{2011A}. However, these methods require very high computational costs to solve the high-order PDEs and result in low efficiency.

This work presents the efficient multi-grid method for solving the highly nonlinear curvature regularization models.
We formulate a patch-based correction strategy from the fine grid layer to coarse grid layers and then interpolate the correction to each point nodal belonging to the patch.
We proposed a forward-backward splitting scheme \cite{combettes2005signal,raguet2013generalized} to solve the curvature minimization problem and prove its convergence theoretically.
More specifically, we first obtain analytical solutions to the mean curvature/Gaussian curvature minimization based on the local geometry property.
In what follows, we solve a convex optimization problem to estimate the patch-wise update.
To further improve the efficiency, we use the four-color domain decomposition method on each layer to enable all subproblems in the same color can be solved in parallel.
Numerous numerical experiments on both image restoration and image reconstruction are presented to demonstrate the efficiency and effectiveness of our algorithm in dealing with large-scale image processing problems.
To sum up, our contributions are concluded as follows
\begin{itemize}
\item We propose an efficient multi-grid method based on subspace correction method for solving the curvature minimization problem \eqref{MCGC}, where the whole space is transferred into small-size local patches;
\item We use the forward-backward splitting scheme to solve the non-convex patch-wise minimization problems, where  subproblems can be efficiently handled by the closed-form solutions;
\item The non-overlapping domain decomposition method is applied to circumvent the dependencies between the adjacent patches, which enables the parallel computation for these subproblems;
\item We develop a GPU-based curvature minimization package by utilizing the parallel computation ability of a GPU card, which is desirable for high-speed real applications.  All our codes and data are available at \url{https://github.com/Duanlab123/MGMC}.
\item The application of the mean curvature minimization is extended to CT and MR reconstruction problems to prove our approach can well balance the image quality and computational efficiency.
\end{itemize}

The rest of the paper is organized as follows.
In Sect. \ref{section1}, the coordinate descend method is proposed to solve the mean curvature minimization problem and the coarse grid structure is used to achieve the fast multi-grid algorithm.
We extend the idea to Gaussian curvature minimization in Sect. \ref{section2}.
Numerical experiments are conducted to illustrate the advantages by comparing with curvature filters and other curvature-related models in Sect. \ref{section3}.
In Sect. \ref{section4}, the multi-grid algorithm is extended to solve image reconstruction problems including CT reconstruction and MR reconstruction.
We conclude the paper in Sect. \ref{section5} with some remarks and future works.

\section{The mean curvature minimization problem}\label{section1}

Without loss of generality, a gray-scale discrete image $u$ of size $m\times n$ has its pixel values $u[k_1,k_2]$ defined at the locations $[k_1,k_2]$ in the domain $\Omega=\{1,\ldots,m\}\times\{1,\ldots,n\}$, where $k_1$ and $k_2$ are the row and column indices, respectively. Then the discrete mean curvature energy of \eqref{MCGC} can be given as
\begin{equation}\label{MCM}
\min_{u}F(u)\!:=\! \sum_{k_1=1}^m\sum_{k_2=1}^n|H(u[k_1,k_2])|+ \frac {\alpha}{2}(u[k_1,k_2]-f[k_1,k_2])^2,
\end{equation}
where $H(u[k_1,k_2])=\frac{1}2(\kappa_{\min}(u[k_1,k_2])+\kappa_{\max}(u[k_1,k_2]))$ denotes the mean curvature of the pixel $(k_1,k_2,u[k_1,k_2])$. We introduce a series of basis functions $\phi_{k_1,k_2}(x)$ on each pixel $(k_1,k_2)\in\Omega$ as follows
\begin{equation*}
\phi_{k_1,k_2}(x)=\left\{
\begin{array}{rcl}
1, & & \mbox{if}\ x=[k_1,k_2];\\
0, & &  \mbox{if}\ x\neq[k_1,k_2].\\
\end{array} \right.
\end{equation*}
Relying on the above basis functions, the minimization problem \eqref{MCM} can be considered as finding the best correction to minimize the curvature-related energy. Referred to \cite{chen2007nonlinear}, we choose an initial value $u_{0}$ and set $l=0$. For $k_1=1,\ldots,m$, $k_2=1,\ldots,n$, we update $u$ using the correction $c\in\mathbb R$ over each pixel as follows
\begin{equation}\label{correction}
u_{l+1} = u_{l} + \sum_{k_1=1}^m\sum_{k_2=1}^nc\phi_{k_1,k_2}(x).
\end{equation}
Relying on the coordinate descend method \cite{Chen2007A,chan2010multilevel, Zhang2021ipi}, the correction equation \eqref{correction} is transferred into a sequence of the one-dimensional minimization problems, the details of which is described as Algorithm \ref{coordinate descend}.

\begin{algorithm*}[tbh]
\caption{The coordinate descend method for solving the  minimization problem \eqref{MCM}} \label{coordinate descend}
\KwIn{ $u_0$, $f$, and $\alpha$}

\For {$l=0,1,2,\cdots$}
{
      \For{ $k_1=1,\cdots,m$, $k_2=1,\cdots,n$}
      {
       \begin{itemize}
       \item Compute the correction $c$ from \begin{equation}\label{local_MSM}
          c=\arg\min_{c\in \mathbb R}J(c):= |H(u_l[k_1,k_2]+c)| + \frac {\alpha}{2}\big(c-f_l^*[k_1,k_2]\big)^2,
       \end{equation}
where $f_l^*[k_1,k_2] =f[k_1,k_2]-u_l[k_1,k_2]$.
        \item Update $u_{l+1}[k_1,k_2]$ by
        \[u_{l+1}[k_1,k_2]=u_l[k_1,k_2]+c;\]
        \end{itemize}
      }
End till some stopping criterion meets;}
\KwOut{$u_{l+1}$}
\end{algorithm*}

The core issue becomes how to solve \eqref{local_MSM} effectively. Here, we use the forward-backward splitting (FBS) method to reformulate the local problem \eqref{local_MSM} into a couple of sub-minimization problems. The forward-backward splitting algorithm is a popular choice for the minimization problem with a smooth data fidelity; see for instance \cite{pustelnik2011parallel,tang2020compressive}. The detailed algorithm is provided as Algorithm \ref{sub_alg}.

\begin{algorithm}[tbh]
\caption{The forward-backward splitting method for solving the local  minimization problem \eqref{local_MSM}} \label{sub_alg}
\KwIn{ $u_l$, $f$, $c_0=0$, and $\alpha$, $\eta_0=1$;}
\For {$t=0,1,2,\cdots$}
{\begin{itemize}
\item The forward step
\begin{equation}
c_{t+\frac12}=\arg\min_{c}|H(u_l[k_1,k_2]+c)|+\frac 1{2\eta_{t}}(c-c_{t})^2;\label{FOBOS step1}
\end{equation}
\item The backward step
\begin{equation}\label{FOBOS step2}
c_{t+1}=\arg\min_{c}\frac {\alpha}{2}\big(c-f_l^*[k_1,k_2]\big)^2+\frac 1{2\eta_{t}}(c-c_{t+\frac12})^2;
\end{equation}
\item Update
\[\eta_{t+1}=1/\sqrt{(1+t)};\]
\item End till some stopping criterion meets;
\end{itemize}
}
\KwOut{$c_{t+1}$}
\end{algorithm}

\noindent\textbf{\emph{Solution to the sub-minimization problem  \eqref{FOBOS step1}.}}
We can use the geometric interpretation to estimate the correction in a local window according to well-known Bernstein's theorem.

\begin{prop}\label{Bernstein}
Let $d_\ell$ be the distances of $(k_1,k_2,u[k_1,k_2])$ on the image surface leaving from the tangent planes. Supposing that the correction $c$ is defined as $c = \frac{1}{\iota}\sum_{\ell=1}^{\iota}d_\ell$ with $\iota$ being the total number of tangent planes, the mean curvature energy $|H(u[k_1,k_2]+c)|$ decreases.
\end{prop}

\begin{proof}
According to Bernstein's theorem, a graph of a real function on $\mathbb R^2$ is a minimal surface, which should be a plane in $\mathbb R^3$. Thus, the flatter the image surface, the smaller the mean curvature regularization term. Suppose there are $\iota$ tangent planes and the corresponding distances of $(k_1,k_2, u[k_1,k_2])$ to its tangent planes are denoted by $d_\ell,~\ell=1,\cdots,\iota$.
To make the image surface as flat as possible, we consider the following quadratic minimization problem
\[\min_{c\in \mathbb R} ~\sum_{\ell=1}^\iota (c-d_\ell)^2,\]
where the optimal correction is  $c=\frac{1}{\iota}\sum_{\ell=1}^{\iota}d_\ell$.
\hfill
\end{proof}

As shown in Fig. \ref{fig2}, we enumerate total 8 local tangent planes in a $3\times3$ window, which are denoted as $T_1$ to $T_8$ located pairwise centrosymmetric and passed through the center point to avoid the grid bias.
Note that we can use more tangent planes to obtain accurate principal curvatures, but this also increases the calculation cost.
Therefore, we later introduce more tangent planes on the coarse layers in the multi-grid framework to balance the effectiveness and efficiency.

Similar to our previous work \cite{Zhong2020image}, we compute the distances $d_\ell$, $\ell=1,\cdots,8$, as illustrated in Fig. \ref{tangent}. More specifically, let the plane $XYZ$ be a tangent plane of $O$ and $\bm n$ be the normal vector.
The directed distance from $O$ to the tangent plane can be calculated by $d=\overrightarrow{XO}\cdot\bm n$ as \eqref{distance to tri}
\begin{figure*}
\begin{equation}\label{distance to tri}
d=\overrightarrow{XO}\cdot \bm n=
\frac{(u[k_1,k_2-1]+u[k_1,k_2+1]-2u[k_1,k_2])}{\sqrt{(u[k_1,k_2-1]+u[k_1,k_2+1]-2u[k_1-1,k_2])^2+(u[k_1,k_2+1]-u[k_1,k_2-1])^2+4}},
\end{equation}
\end{figure*}
where $\bm n$ is defined by the cross product of the vector $\overrightarrow{XZ}$ and $\overrightarrow{XY}$, i.e., $\bm n=\overrightarrow{XZ}\times\overrightarrow{XY}$. The computation of $d_\ell$, $\ell=1,\ldots, 8$, can be implemented in the same way.
And the update of $c_{t+\frac12}$ can be estimated as
\begin{equation}\label{c-sol1}
c_{t+\frac12}= \frac18\sum_{\ell=1}^8d_\ell.
\end{equation}
\noindent\textbf{\emph{Solution to the sub-minimization  problem \eqref{FOBOS step2}.}}
We are facing a quadratic minimization problem,
which is formulated as
\begin{equation}\label{quadratic}
\min_{c}\frac {\alpha}{2}\big(c-f_l^*[k_1,k_2]\big)^2+\frac 1{2\eta_{t}}(c-c_{t})^2+\frac 1{2\eta_{t}}(c-c_{t+\frac12})^2.
\end{equation}
The minimization problem \eqref{quadratic} can be solved by the closed-form solution as follows
\begin{equation}\label{c-sol2}
c_{t+1} \!=\! \frac{1}{2+\alpha\eta_t}\Big(c_t+c_{t+\frac12}+\alpha\eta_t f_l^*\Big).
\end{equation}

\subsection{Convergence analysis of Algorithm \ref{sub_alg}}

\begin{figure}[t]
\centering
\includegraphics[width=0.5\textwidth]{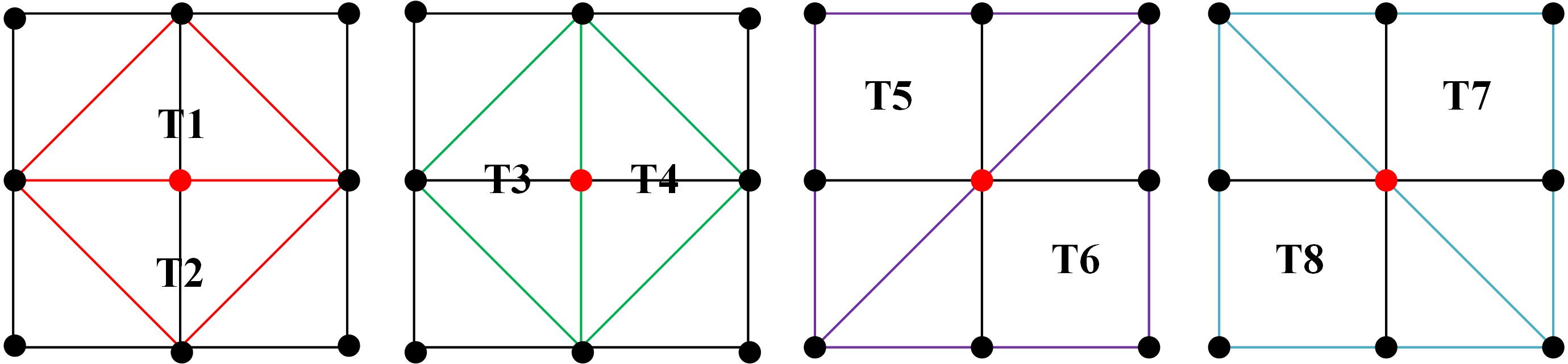}
\caption{Illustration of the eight tangent planes located in a $3\times 3$ local patch on the finest layer, which locates pairwise centrosymmetric  concerning  the center point nodal (marked by red color). }\label{fig2}
\end{figure}

In the subsection, we present a brief discussion to show the energy diminishing of Algorithm \ref{sub_alg}.
Let $c^*\in \mathbb R$ denote a minimizer of model \eqref{local_MSM}, the mean curvature term and date fidelity term are denoted by  $f(c)=|H(u_l[k_1,k_2]+c)|$ and $r(c)=\frac{\alpha}{2}(c-f^*[k_1,k_2])^2$, respectively.
The next lemma provides a key tool for deriving convergence.  For more details, please refer to the forward-backward splitting scheme in \cite{duchi2009efficient}.
\begin{lem}\label{Bounding Step Differences}
(\textbf{Bounding Step Differences})
Assume that the norm of the gradient of $r(c)$ and $f(c)$  are bounded as
\[ \|\nabla r(c)\|^2\leq G^2,\quad \|\partial f(c)\|^2\leq D^2,\]
where $G,D$ are the Lipschitz constant of $\nabla r(c)$ and $\partial f(c)$, respectively. We have
\begin{equation}\label{telescoping}
\begin{split}
 2\eta_t \Big(f(c_{t})+r(c_{t+1})&-J(\widetilde{c})\Big)\leq \|c_{t}- \widetilde{c}\|^2 \\
&- \|c_{t+1}- \widetilde{c}\|^2+ \eta_t^2(5G^2+3D^2).
\end{split}
\end{equation}
\end{lem}
\begin{proof}
The proof is sketched in Appendix.
\end{proof}
Based on Lemma \ref{Bounding Step Differences}, we can prove the following result, which is the basis for deriving convergence results.

\begin{lem}\label{Telescoping sum}
Assuming that the norm of $\|\widetilde{c}\|^2\leq E^2$ with $E$ being a positive constant, we sum the residuals \eqref{telescoping} over $t$ from 1 through $T$ and get a telescoping sum
\[\sum_{t=1}^T\eta_t\Big[ f(c_{t})+r(c_{t})-J(\widetilde{c})\Big]\leq\widetilde{G},\]
where $\widetilde{G}= E^2+3\sum_{t=1}^T\eta_t^2(G^2+D^2)$.
\end{lem}
\begin{proof}
The proof is similar to Theorem 2 in \cite{duchi2009efficient}.
\end{proof}
Therefore, a direct consequence of Lemma \ref{Telescoping sum} can be obtained when running  Algorithm \ref{sub_alg}
with $\eta_t \varpropto 1/\sqrt{t}$  or with non-summable step sizes decreasing to zero.

\begin{theorem}\label{MC convergence}
Assume that the conditions of Lemma \ref{Telescoping sum} hold and the step size $\eta_t$ satisfy $\eta_t\rightarrow 0$ and   $\sum_{t=1}^\infty\eta_t=\infty$.
Then we have
\[\liminf_{t\rightarrow\infty} J(c_t)  -J(c^*)=0. \]
\end{theorem}
\begin{proof}
For the first $T$ iterations,  Lemma \ref{Telescoping sum} gives
\begin{small}
\begin{equation*}
\min_{t\in\{0,\cdots,T\}} \big(f(c_{t})+r(c_{t})-J(\widetilde{c})\big)\sum_{t=1}^T\eta_t
\leq\sum_{t=1}^T\eta_t\Big[f(c_{t})+r(c_{t})-J(\widetilde{c})\Big].
\end{equation*}
\end{small}
Let $\widetilde{c}=c^*$, we have
\begin{small}
\begin{equation*}
\liminf_{t\rightarrow\infty}  J(c_t)-J(c^*) \leq \frac{E^2}{\sum_{t=1}^\infty\eta_t}+\frac{3\sum_{t=1}^\infty \eta_t^2(G^2+D^2)}{\sum_{t=1}^\infty\eta_t}\rightarrow 0.
\end{equation*}
\end{small}
\end{proof}

\begin{remark}
The step size $\eta$ can be a constant or diminishing with iterations, e.g., $\eta_t= 1/\sqrt{t}$.
\end{remark}

\begin{remark}
Since the mean curvature minimization subproblem can be explicitly solved by \eqref{c-sol1}, we simply set $t=0$  in Algorithm \ref{sub_alg} for all numerical experiments.
\end{remark}

\begin{figure}[t]
\centering
\includegraphics[width=0.25\textwidth]{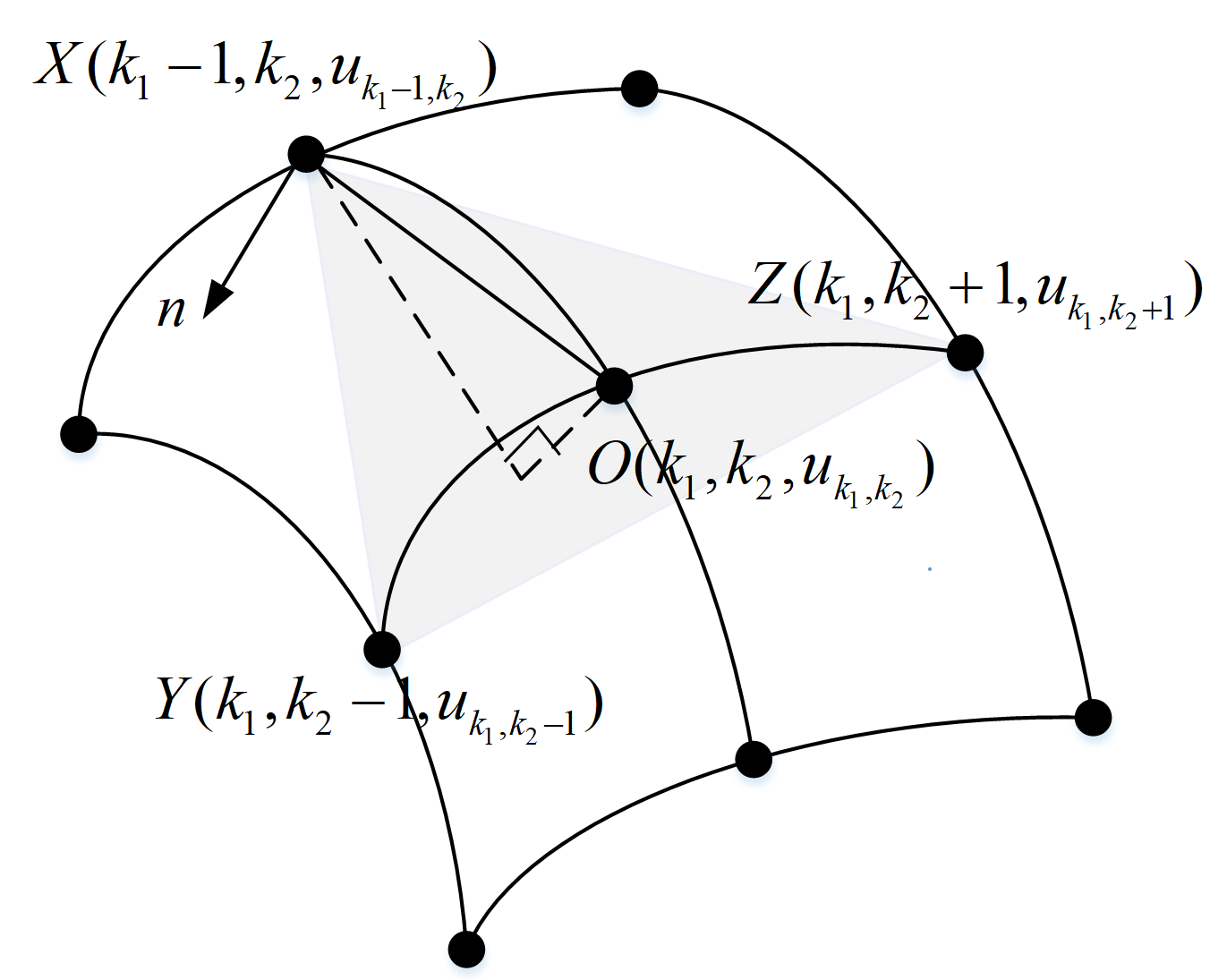}
\caption{Illustration of the distance to the tangent plane, where $x=(k_1,k_2)$ indicates the coordinates and $u$ denotes the image intensity function. }\label{tangent}
\vspace{-0.3cm}
\end{figure}

\subsection{Our multi-grid algorithm}\label{mg}
Regarding problem \eqref{local_MSM} as the finest grid, we can use a larger local window and generate the multi-grid algorithm for solving the mean curvature minimization problem.
Without loss of generality, we assume the initial grid $\mathcal{T}$ consists of $m\times n$ grid points.
Starting from the finest grid $\mathcal{T}_1=\mathcal{T}$, we consider a sequence of coarse structure,
$\mathcal{T}_1,\mathcal{T}_2,\cdots,\mathcal{T}_J$, with $J$ being the total number of layers.
Our multi-grid structure is constructed by gathering the grids point into non-overlapping patches of different sizes.
Specifically, the size of the patch set $\tau_j$ on the $j$th coarse layer is $(2^{j}-1)\times (2^j-1)$, and  $\mathcal{T}_j$ contains $m_j\times n_j$ patches with
\[m_j=\lceil m/(2^{j}-1)\rceil,~n_j=\lceil n/(2^{j}-1)\rceil,\]
where $\lceil \cdot \rceil$ is a rounding up function.
Then there are $N_j=m_j\times n_j$ patches on the $j$th coarse layers and the
partition can be expressed as $\mathcal{T}_j=\{\tau_j^i\}_{i=1}^{N_j}$.
In order to ensure each patch is complete, we prolongate the image domain before the partition using boundary conditions. It is straightforward to define $V_j$ as a finite element space
\[V_j=\{v:\ v|_{\tau_j}\in\mathcal{P}_c(\tau_j),\ \forall\tau_j\in\mathcal{T}_j  \},\]
where $\mathcal{P}_c$ denotes the space of all piecewise constant functions.
We equip the piecewise constant function space $V_j$ with a set of basis functions $\{\phi_j^i\}_{i=1}^{N_j}$, which is defined as
\begin{equation*}
\phi_j^i(x)=\left\{
\begin{array}{rcl}
1, & & \mbox{if}\ x\in\tau_j^i;\\
0, & &  \mbox{if}\ x\notin\tau_j^i;\\
\end{array}\quad i=1,\cdots,N_j. \right.
\end{equation*}
Associated with each basis function, we define the one dimensional subspace  $V_j^i=\mathrm{span}\{\phi_j^i\}$.
Then, the whole space $V$ can be expressed as
$V=\sum_{j=1}^J\sum_{i=1}^{N_j}V_j^i$.

On the coarse grids, we consider a larger local patch  including more local tangent planes to be enumerated. We can come up with a recurrence formula for the number of tangent planes on the $j$th layer as $\iota = 2^{j+2}$. For example, we enumerate the total of $16$ triangular planes for patches on the first coarse layer, half of which are the same as the finest layer and the left ones are displayed in Fig. \ref{second level}.
Correspondingly, we define one-dimensional subspace minimization problem \eqref{FOBOS step2} over the subspace $V_j^i$, $i=1,\cdots,N_j$, $j=1,\cdots,J$ as follows
\begin{equation*}
\min_{c\in \mathbb R}~\frac {1}{2\eta_t}\Big(c-\frac{1}{\iota}\sum_{\ell=1}^{\iota}d_\ell\Big)^2 +\frac {1}{2\eta_t}(c-c_t)^2+ \frac{\alpha s}{2}\big(c-f_l^*\big)^2,
\end{equation*}
where
$f_l^*=\sum\limits_{(k_1,k_2)\in\tau_j^i}(f[k_1,k_2]-u_l[k_1,k_2])/s$, $s=\sum\limits_{x\in\tau_j^i}\phi_j^i(x)$.
The closed-form solution is defined as follows
\[c_{t+1} = \frac{1}{2+\eta_t\alpha s}\Big(c_t+\frac{1}{\iota}\sum_{\ell=1}^\iota d_\ell+\eta_t\alpha sf_l^*\Big).\]
Then, the correction $\bm c_j =(c_j^1,\cdots,c_j^i,\cdots,c_j^{N_j})$ on the $j$th layer is reshaped into a matrix of size $m_j\times n_j$ as follows
\[\bm c_j\!=\!\begin{bmatrix}
 \begin{array}{c|c|c|c|c}
   \cdots &\vdots & \vdots &\vdots & \cdots  \\
   \hline
    \vdots & c_j^{i_1-1,i_2-1}& c_j^{i_1-1,i_2} &c_j^{i_1-1,i_2+1}& \vdots\\
    \hline
     \vdots & c_j^{i_1,i_2-1}& c_j^{i_1,i_2} &c_j^{i_1,i_2+1}& \vdots\\
     \hline
    \vdots & c_j^{i_1+1,i_2-1}& c_j^{i_1+1,i_2} &c_j^{i_1+1,i_2+1}& \vdots\\
    \hline
   \cdots &\vdots & \vdots &\vdots & \cdots  \\
    \end{array}
\end{bmatrix}_{m_j\times n_j},\]
where $i_1=1,\cdots,m_j$, $i_2=1,\cdots,n_j$, and $m_j$, $n_j$ are the number of patches in the row direction and column direction, respectively.
Because we use the piecewise constant basis function over the support set, we can define an interpolation matrix $
L_j:\mathbb R^{m_jn_j}\rightarrow \mathbb R^{mn}$ to update the solution on the finest layer, i.e.
\[L_j\bm c_j\!=\!\begin{bmatrix}
\begin{footnotesize}
  \begin{array}{c|ccc|c}
   \vdots & \vdots & \vdots   &  \vdots & \vdots \\
   c_j^{i_1-1,i_2-1}&c_j^{i_1-1,i_2}& \cdots &c_j^{i_1-1,i_2}&c_j^{i_1-1,i_2+1} \\
    \hline
     c_j^{i_1,i_2-1}   &  c_j^{i_1,i_2}   & \cdots & c_j^{i_1,i_2} &c_j^{i_1,i_2+1}  \\
    \cdots & \vdots & \cdots &\vdots&\cdots\\
     c_j^{i_1,i_2-1}   & c_j^{i_1,i_2}   & \cdots &c_j^{i_1,i_2} & c_j^{i_1,i_2+1} \\
     \hline
       c_j^{i_1+1,i_2-1}& c_j^{i_1+1,i_2} &\cdots & c_j^{i_1+1,i_2}&c_j^{i_1+1,i_2+1} \\
       \vdots   & \vdots   & \cdots & \vdots & \vdots \\
     \end{array}
\end{footnotesize}
\end{bmatrix}_{m\times n}.\]
Then the solution can be defined as $u_{l+1} = u_{l}+L_j\bm c_j$.

\begin{figure}[t]
\centering
\includegraphics[width=0.5\textwidth]{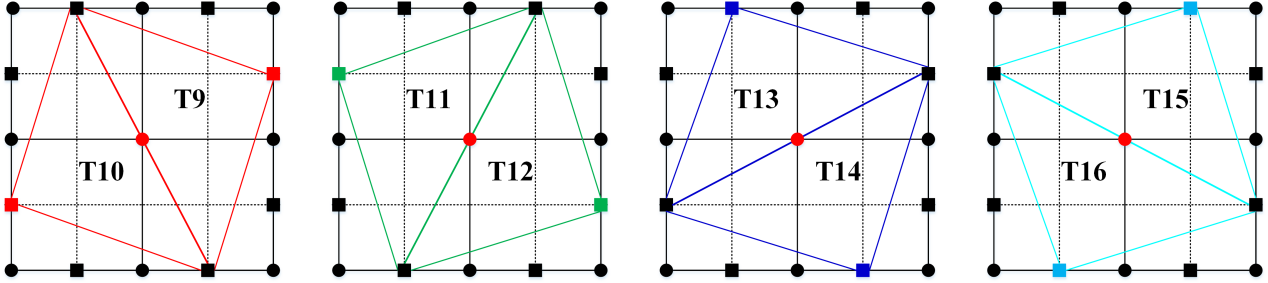}
\caption{Illustration of the eight tangent planes $T_{9}-T_{16}$ together with $T_1 -T_8$ in Fig. \ref{fig2} form a coherent whole for a $5\times 5$ patch on the first coarse layer. }\label{second level}
\vspace{-0.3cm}
\end{figure}

We use the V-cycle to solve the minimization model from the finest layer $V_1$ to the coarsest layer $V_J$, and then from the coarsest layer to the finest layer. In practice, we find that half of the V-cycle is sufficient for the decrease of the energy functional, while the other half of the V-cycle does little improvement. Thus the coarse to fine subspace correction can be omitted.

\subsection{The domain decomposition strategy}
The domain decomposition method (DDM) is another promising technique to deal with large-scale problems, which divides the large-scale problem into smaller problems for parallel computation.
In the following, we will apply the non-overlapping domain decomposition method to enable parallel computation.

Fig. \ref{4-color} displays the four-color decomposition on the finest layer and the second layer, respectively.
More specifically, we divide the basis function $\{\phi_j^i\}_{i=1}^{N_j}$ into four groups $\cup_{k=1}^4\{\phi_j^i:i\in I_k\}$
to reduce the dependency on the order of basis functions and improve the parallelism for subproblems on each layer, where $I_k$ contains the indexes with the same color.
This decomposition guarantees that neighboring patches in a 4-connected neighborhood are in different subsets. We can see that the support of the basic functions $\{\phi_j^i:i\in I_k\}$ are non-overlapping for each $k=1,2,3,4$, and the minimization of $F(u+c_j^i\phi_j^i)$ for $i\in I_k$ can be solved in parallel.
In particular, four subproblems are solved in consecutive order
\[\min_{\delta u\in V_j^{(k)}}F(u+\delta u),\quad \mbox{for} ~~k=1,2,3,4,\]
where $V_j^{(k)}=\mathrm{span}\{\phi_j^i:i\in I_k\}$ and $V_j=\sum_{k=1}^4V_j^{(k)}$. It is readily checked that
\begin{equation*}
\min_{\delta u\in V_j^{(k)}} F(u+\delta u)=\min_{\bm c_j\in \mathbb R^{|I_k|}}F(u+\sum_{i\in I_k}c_j^i\phi_j^i),
\end{equation*}
where $\bm c_j=(c_j^1,c_j^2,\cdots,c_j^{|I_k|})$ with $|I_k|$ being the total number of elements in $I_k$ and $N_j=\sum_{k=1}^4|I_k|$.

\begin{figure}[t]
  \centering
  \includegraphics[width=0.39\textwidth]{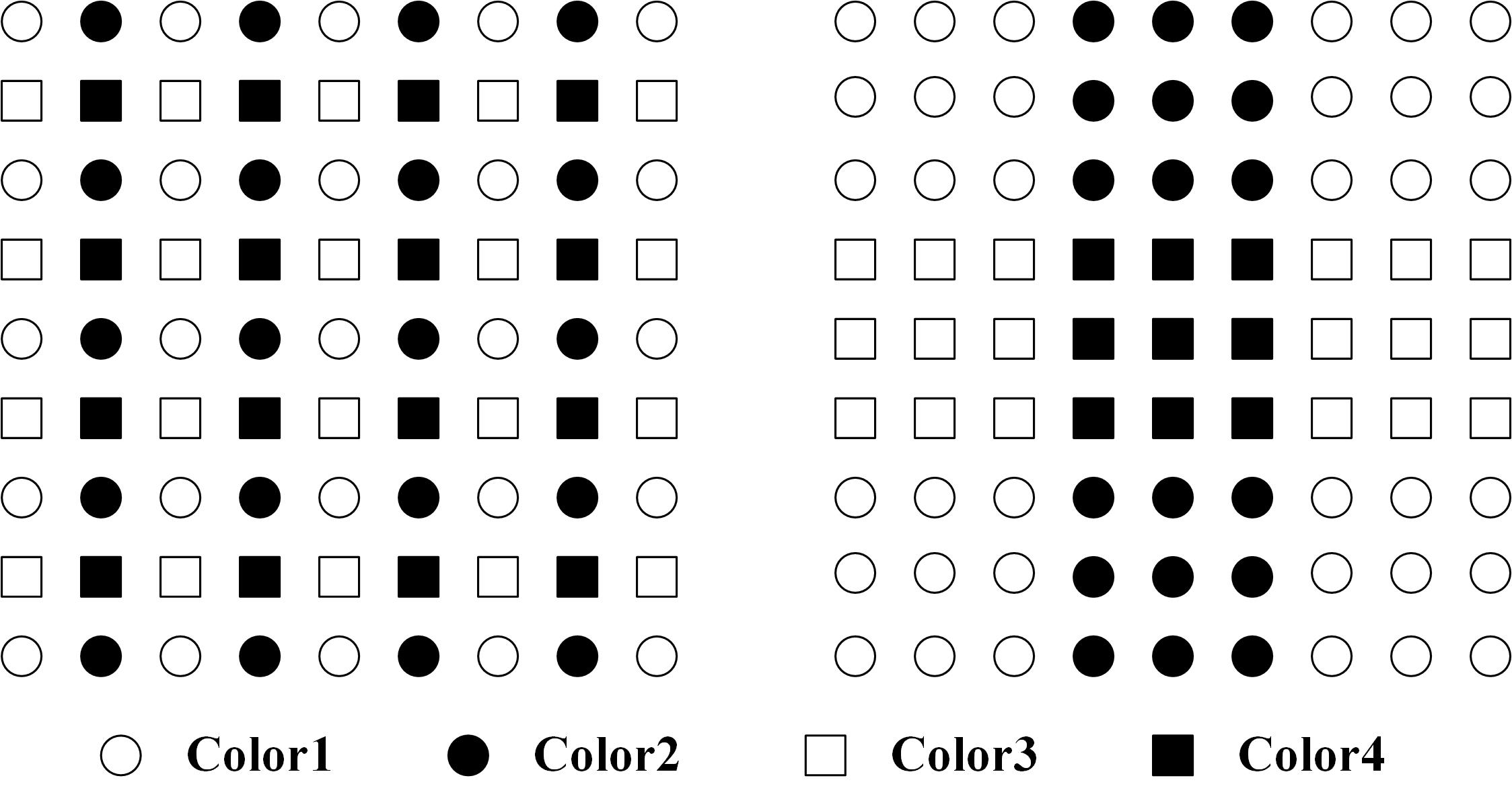}
    \vspace{-0.3cm}
  \caption{Illustration of 4-color domain decomposition for a domain of size $9\times 9$, where the subproblems of the same color can be computed in parallel.}\label{4-color}
  \vspace{-0.4cm}
\end{figure}
Then, the implementation of the algorithm to solve the mean curvature minimization problem \eqref{MCM} is summarized in Algorithm \ref{MC_algor}.
\begin{algorithm}[tbh]
\caption{The multi-grid method for solving the mean curvature minimization model \eqref{MCM}} \label{MC_algor}
\KwIn{ $u_0$, $f$, $\alpha$;}
\For {$l=1,2,\cdots$\tcc{Outer iterations}}
{
      \For{$j=1$ to $J$ \tcc{ From the fine layer to coarse layer}}
      {

       \For {$k=1$ to $4$, \tcc{the four-color DDM iterations}}
      {
$\bm c_j\!=\!\arg\min\limits_{\bm c\in \mathbb R^{|I_k|}} F(u+\sum_{i\in I_k}c_j^i\phi_j^i)$  \;

$u_{l+j+\frac{k}{4} }\leftarrow u_{l+j+\frac{k-1}{4}}+\sum\limits_{i\in I_k}c_j^i\phi_j^i$;

}
      }

      End till some stopping criterion meets;

}
\KwOut{$u_{l+1}$}
\end{algorithm}

\subsection{Complexity analysis}
In this section, the floating point operations (FLO) are used to evaluate the complexity of the algorithm.
The cost of our method is  mainly related to the following two parts.
The first is the distances to the tangent planes $d_\ell$, which is about  $4\times 2^{j+2}$ floating point operations (FLO) on each patch set $\tau_j^i$.
Therefore, the total cost on the coarse level $j$ is about $\mathcal{O}(4m_j n_j2^{j+2})\approx\mathcal{O}(4m n/(2^{j}-1))$.
The second cost is the interpolation operator $L_j$, which is about  $mn$ FLO.
The computation of $f^*$ is about $2(2^{j}-1)(2^{j}-1)m_jn_j\thickapprox2mn$.
Then the number of FLOs over all $J$ levels should be $\mathcal{O}(\sum_{j=1}^J(3+4/(2^{j}-1))mn)$.
 In particular, the computational complexity of our multi-grid algorithm with 3 layers is about $\mathcal{O}(12mn)$.

As a point of reference, the mean curvature method \cite{2017Augmented} contains three sub-problems, which can be solved by either the shrinkage operation or the FFT with the total computational complexity of $\mathcal{O}(6mn\log_2 mn+8mn)$.
The computational complexity of the total absolute mean curvature model \cite{Zhong2020image} and Euler's elastica model \cite{2011AA} can be obtained similarly, which are $\mathcal{O}(2mn\log_2 mn+3mn)$ and  $\mathcal{O}(6mn\log_2 mn+4mn)$, respectively. It is shown that our multi-grid  algorithm is with much low computational complexity compared to the existing high order methods.

\section{The Gaussian curvature minimization problem}\label{section2}

We can directly extend the proposed multi-grid method to solve the following Gaussian curvature minimization problem
\begin{equation}\label{GCM}
\min_{u}F(u)\!:=\! \sum_{k_1=1}^m\sum_{k_2=1}^n|K(u[k_1,k_2])|+ \frac {\alpha}{2}(u[k_1,k_2]-f[k_1,k_2])^2,
\end{equation}
where $K(u[k_1,k_2]) = \kappa_{\min}(u[k_1,k_2])\kappa_{\max}(u[k_1,k_2])$ is the Gaussian curvature over pixel $(k_1,k_2,u[k_1,k_2])$.
The one-dimensional problem for the Gaussian curvature minimization problem over the finest grid is given as follows
\begin{equation}\label{local_GC}
\min_{c\in \mathbb R}~~ |K(u[k_1,k_2]+c)| + \frac {\alpha}{2}\big(c-f^*[k_1,k_2]\big)^2.
\end{equation}
Similarly, we use the FBS scheme to solve the local problem \eqref{local_GC}. The only difference is how to estimate the minimizer of curvature regularization term.
Supposing that $\iota$ tangent planes are enumerated, we can estimate $\iota$ normal curvatures $\kappa_\ell,~\ell=1,2,\cdots,\iota$ in the local patch. According to differential geometry theory, the normal curvature can be calculated as the quotient of \emph{the second fundamental form} and \emph{the first fundamental form} as follows
\[\kappa_\ell=\frac{\mathrm{II}}{\mathrm{I}}\approx \frac{d_\ell}{ds^2},\]
where $d_\ell$ denotes the distance of a neighboring point to the tangent plane and $ds$ denotes the arc-length between the neighboring point and the central point.
Then Gaussian curvature can be defined by the two principal curvatures,
where the principal curvatures are obtained by
$\kappa_{\min}=\min\{\kappa_\ell(u[k_1,k_2])\},~\kappa_{\max}=\max\{\kappa_\ell(u[k_1,k_2])\}$
for $\ell=1,\cdots,\iota$. We further denote $\kappa^*(u[k_1,k_2])=\min\{|\kappa_{\min}|,|\kappa_{\max}|\}$ be the principal curvature with the smaller absolute value, and $T^{*}$ be the corresponding tangent plane.
Thereupon, we have the following proposition to estimate the analytical solution for Gaussian curvature minimization.
\begin{prop}\label{min GC}
The correction $c$ on each point $(k_1,k_2,u[k_1,k_2])\in\Omega$ to minimize the Gaussian curvature $|K(u[k_1,k_2]+c)|$ is given as $c=d^*$, where $d^*$ is the distance of  $(k_1,k_2,u[k_1,k_2])$ to the tangent plane $T^{*}$.
\end{prop}
\begin{proof}
Since the point $(k_1,k_2,u[k_1,k_2]+d^*)$ is on the tangent plane w.r.t. the principle curvature, we have $0=\big|K(u[k_1,k_2]+d^*)\big|\leq \big|K(u[k_1,k_2])\big|$.
\hfill
\end{proof}
Then, we can use Algorithm \ref{sub_alg} to solve the patch problem \eqref{local_GC}, and both the multi-grid method and domain decomposition method can be applied to solve the Gaussian curvature minimization problem \eqref{GCM} without much effort. Therefore, we omit the details here.

\section{Numerical experiments}\label{section3}
In this section, we evaluate the performance of the proposed multi-grid algorithm on the image denoising problem.
The qualities of the denoised images are measured by both the Peak Signal to Noise Ratio (PSNR)
and the Structural Similarity Index Measure (SSIM). All of the experiments are implemented in a MATLAB R2016a environment on a desktop with an Intel Core i9 CPU at 3.3 GHz and 8 GB memory.

\begin{figure}[t]
  \centering
   \includegraphics[width=0.35\textwidth]{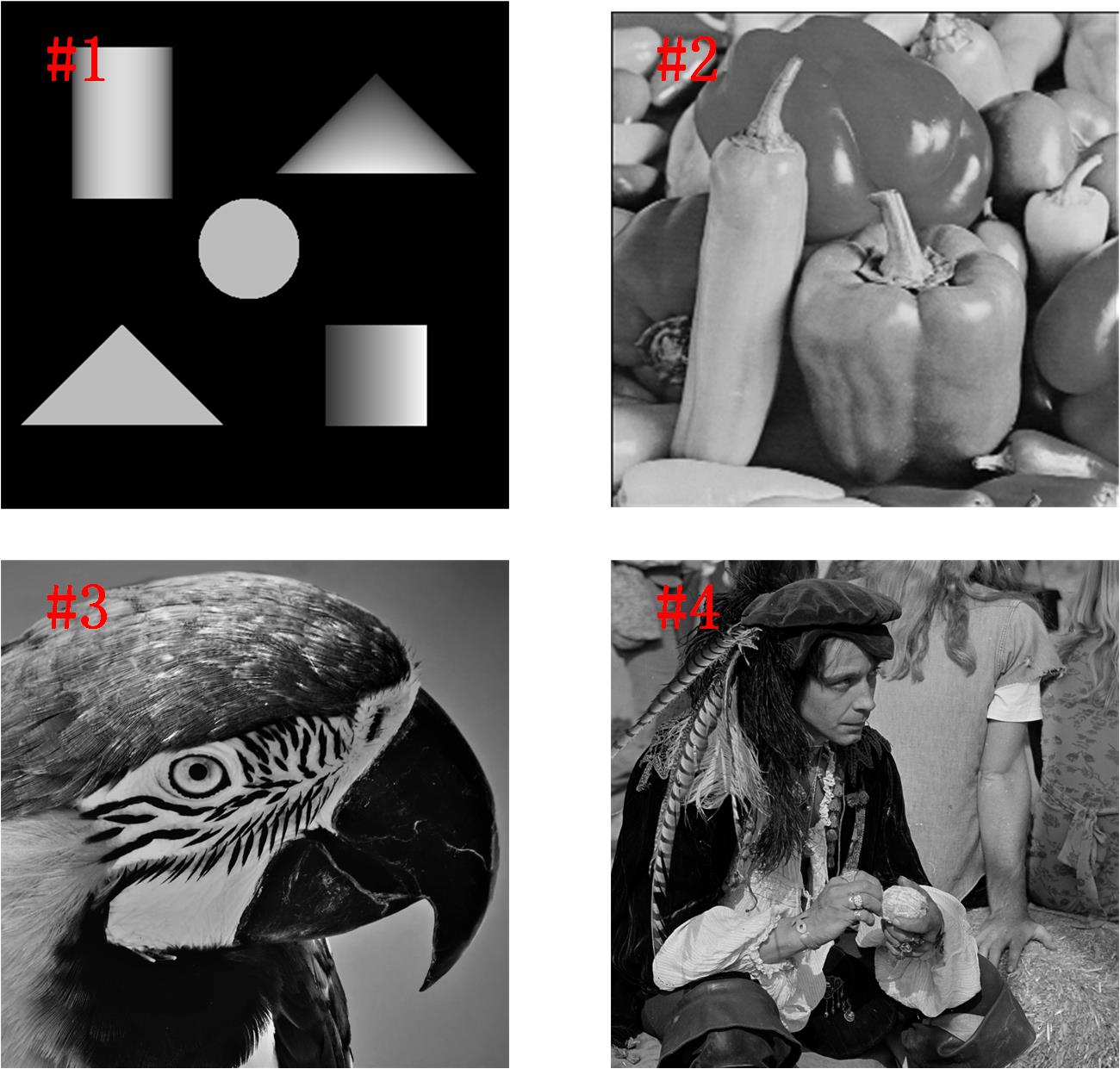}
  \caption{The test images used in the numerical experiments, where the image `Triangle' and `Peppers' are of size $512\times 512$, image `Parrot' and `Man' are of size $1024\times1024$.}\label{example}
\vspace{-0.3cm}
\end{figure}

\begin{table*}[!htbp]
 \centering
 \small
  \caption{The denoising results on test images to different numbers of layers for the noise level $\sigma=10$.}
\begin{tabular}{p{8pt}|p{3pt}|p{5pt}p{11pt}p{21pt}|p{5pt}p{11pt}p{21pt}|p{5pt}p{11pt}p{21pt}|p{5pt}p{11pt}p{21pt}|p{5pt}p{11pt}p{21pt}|p{5pt}p{11pt}p{21pt}}
  \hline
  \hline
  & &\multicolumn{9}{c|}{Mean curvature}&\multicolumn{9}{c}{Gaussian curvature}\\
  \hline
  &$\alpha$ &\multicolumn{3}{c|}{0.1}& \multicolumn{3}{c|}{0.06}&\multicolumn{3}{c|}{0.03}  &\multicolumn{3}{c|}{0.1}& \multicolumn{3}{c|}{0.06}&\multicolumn{3}{c}{0.03} \\
   \cline{2-20}
&$J$& {\scriptsize Iter}  & {\scriptsize CPU(s)} &{\scriptsize  Energy}  &  {\scriptsize Iter}  & {\scriptsize CPU(s)} &{\scriptsize  Energy}& {\scriptsize Iter}  & {\scriptsize CPU(s)} &{\scriptsize  Energy}  & {\scriptsize Iter}  & {\scriptsize CPU(s)} &{\scriptsize  Energy} & {\scriptsize Iter}  & {\scriptsize CPU(s)} &{\scriptsize  Energy} & {\scriptsize Iter}  & {\scriptsize CPU(s)} &{\scriptsize  Energy} \\
   \hline
    \multirow{6}[0]{*}{$\#1$}
   & 1 & 257  & 24.41 & 2.3158  & 378  & 38.01 & 2.4117 & 417 & 34.41 & 2.5074& 118   & 10.05  & 1.8541  & 197   & 15.47  & 2.1098  & 275   & 21.10  & 2.3178 \\
   & 2 & 77   & 8.96  & 2.3004  & 80   & 9.43  & 2.4161 & 114 & 13.21 & 2.5077& 32    & 4.84  & 1.8496  & 41    & 4.81  & 2.1064  & 40    & 6.72  & 2.2934 \\
   & 3 & 55   & \textbf{5.39}  & 2.2835  & 62   & \textbf{5.47}  & \textbf{2.4003} & 67  & \textbf{8.67}  & \textbf{2.5071}& 27    & \textbf{3.66}  & \textbf{1.8465}  & 35    & \textbf{4.07}  & \textbf{2.1002}  & 43    & \textbf{5.82}  & \textbf{2.2911}  \\
   & 4 & 57   & 8.74  & \textbf{2.2821}  & 65   & 8.98  & 2.4005 & 77  & 10.94 & 2.5078& 31    & 4.84  & 1.8471  & 30    & 4.80  & 2.1005  & 43    & 5.92  & 2.2972 \\
   & 5 & 61   & 10.11 & 2.2891  & 61   & 6.65  & 2.4008 & 68  & 9.42  & 2.5077& 31    & 4.68  & 1.8496  & 37    & 5.70  & 2.1005  & 41    & 6.26  & 2.2936 \\
   & 6  & 56  & 9.84  & 2.2875  & 60   & 7.91  & 2.4009 & 69  & 9.67  & 2.5072 & 30    & 4.80  & 1.8496  & 32    & 4.92  & 2.1006  & 37    & 5.68  & 2.2914 \\
         \hline
    \multirow{6}[0]{*}{$\#2$}
   & 1 &195   & 21.28 & 2.2833 & 318   & 32.16 & 2.4764 & 431   & 43.97 & 2.4703 & 120   & 10.43 & 1.8855 & 188   & 15.46 & 2.1717 & 254   & 20.82 & 2.3629 \\
   & 2 &80    & 9.23  & 2.2675 & 93    & 11.65 & 2.4754 & 107   & \textbf{12.91} & 2.4772 & 53    & 6.69  & 1.8776 & 73    & 9.33  & 2.1529 & 86    & 10.81 & 2.3661\\
   & 3 &56    & \textbf{7.56}  & \textbf{2.2637} & 64    & \textbf{9.37}  &\textbf{2.4557} & 95    & 13.64 & \textbf{2.4721}  & 47    & \textbf{6.65}  & \textbf{1.8641} & 59    & \textbf{8.14}  & \textbf{2.1512} & 76    & 10.46 & 2.3635\\
   & 4 &50    & 7.98  & 2.2639 & 66    & 10.54 & 2.4558 & 106   & 15.81 & 2.4722 & 46    & 6.88  & 1.8691 & 65    & 9.81  & 2.1513 & 71    & \textbf{10.45} & \textbf{2.3625} \\
   & 5 &52    & 8.03  & 2.2641 & 60    & 9.91  & 2.4595 & 98    & 14.45 & 2.4761 & 46    & 6.95  & 1.8655 & 56    & 9.21  & 2.1517 & 69    & 10.63 & 2.3646\\
   & 6 &60    & 10.45 & 2.2641 & 64    & 10.75 & 2.4578 & 106   & 16.61 & 2.4737 & 50    & 7.56  & 1.8684 & 56    & 9.86  & 2.1514 & 69    & 10.55 & 2.3642 \\
              \hline
    \multirow{6}[0]{*}{$\#3$}
    & 1  &159    & 58.76 & 1.1672 & 227   & 85.18 & 1.3693 & 293  & 100.2 & 1.5325& 124   & 46.99  & 8.7365  & 184   & 68.37  & 1.1333  & 246   & 90.40  & 1.3789  \\
    & 2  &42     & 23.66 & 1.1512 & 54    & 32.49 & 1.3647 & 54   & 29.81 & 1.5339 & 46    & 25.92  & 8.7362  & 57    & \textbf{30.69}  & 1.1205  & 63    & 35.93  & 1.3769  \\
    & 3  &38     & \textbf{23.24} & \textbf{1.1512}& 51    & \textbf{26.86} & \textbf{1.3646} & 59   & \textbf{28.71} & \textbf{1.5155}& 44    & 27.30  & \textbf{8.7352}  & 50    & 31.58  & \textbf{1.1024}  & 54    & \textbf{33.96}  & 1.3536 \\
    & 4  &41     & 26.67 & 1.1529 & 43    & 29.01 & 1.3661 & 47   & 31.78 & 1.5174& 45    & \textbf{23.91}  & 8.7358  & 58    & 38.63  & 1.1094  & 84    & 57.77  & \textbf{1.3529}  \\
    & 5  &40     & 27.75 & 1.1522 & 42    & 29.35 & 1.3671 & 49   & 31.77 & 1.5220  & 44    & 29.23  & 8.7353  & 56    & 38.10  & 1.1028  & 73    & 51.80  & 1.3560  \\
    & 6  &42     & 29.72 & 1.1548 & 42    & 27.87 & 1.3667 & 51   & 33.67 & 1.5202 & 44    & 29.56  & 8.7358  & 52    & 35.79  & 1.1038  & 77    & 54.39  & 1.3592  \\
           \hline
    \multirow{6}[0]{*}{$\#4$}
   & 1  & 177    & 65.96 & 1.2591 & 273   & 98.02 & 1.5321 & 350    & 124.4 & 1.7344 & 127   & 47.95  & 9.4909  & 186   & 69.20  & 1.2658  & 245   & 89.48  & 1.3643\\
   & 2  & 54     & 25.59 & 1.2585 & 57    & 30.33 & 1.5366 & 59     & 30.44 & 1.7259  & 53    & \textbf{31.34}  & 9.4814  & 60    & 36.25  & 1.2664  & 73    & 40.80  & 1.3463   \\
   & 3  & 43     & 26.36 & 1.2588 & 50    & \textbf{26.68} & 1.5371 & 53     & \textbf{29.32} & \textbf{1.7233}& 49    & 31.75  & 9.4709  & 56    & \textbf{34.67}  & \textbf{1.2611}  & 60    & \textbf{37.00}  & \textbf{1.3406}  \\
   & 4  & 38     & \textbf{24.23} & \textbf{1.2523} & 41    & 28.46 & \textbf{1.5320} & 57     & 30.74 & 1.7234& 47    & 33.21  & \textbf{9.4702}  & 55    & 35.73  & 1.2618  & 73    & 49.36  & 1.3412  \\
   & 5  & 45     & 28.58 & 1.2599 & 45    & 28.69 & 1.5364 & 51     & 32.33 & 1.7237 & 41    & 32.50  & 9.4716  & 52    & 35.68  & 1.2616  & 52    & 45.34  & 1.3415 \\
   & 6  & 47     & 30.81 & 1.2602 & 47    & 29.99 & 1.5323 & 43     & 31.45 & 1.7239  & 49    & 35.07  & 9.4724  & 54    & 38.92  & 1.2615  & 50    & 44.06  & 1.3411\\
\hline
\hline
\end{tabular}
\label{table1}
\end{table*}

\subsection{The effect of the multi-grid method}\label{test_level}
The choice of the maximal number of layers is important in the proposed multi-grid method, which affects the numerical performance of the curvature minimizations. We implement both multi-grid mean curvature (MGMC) and multi-grid Gaussian curvature (MGGC) methods on test images shown in Fig. \ref{example}, which are corrupted by white Gaussian noise with zero mean and standard deviation $\sigma = 10$. In the experiment, the regularization parameter varies as $\alpha\in\{0.1,0.06,0.03\}$ and the number of layers changes as $J\in\{1,2,3,4,5,6\}$. Both MGMC and MGGC are stopped when the following relative error of the numerical energy is smaller than the predefined tolerance
\begin{equation}\label{reerr}
RelErr\big(F(u_{l+1})\big)=|F(u_{l+1})-F(u_{l})|/|F(u_{l+1})|\leq\epsilon,
\end{equation}
which is set as $\epsilon=10^{-6}$.

\begin{table*}[htbp]
  \centering
  \caption{The compared result of the number of iterations, PSNR, CPU(s), and CPU ratio for different size images with V-cycle and half of the V-cycle (denoted by H-V-cycle).}
\small
    \begin{tabular}{p{44pt}|p{20pt}|c|c|c|c|c|c|c|c|c}
    \hline
    \hline
        \multirow{2}[0]{*}{Method}  &  \multirow{2}[0]{*}{Sizes}     &    \multirow{2}[0]{*}{N}   & \multicolumn{4}{c|}{$\#1$}    & \multicolumn{4}{c}{$\#3$} \\
          \cline{4-11}
          & & & \multicolumn{1}{c|}{$\#$} & \multicolumn{1}{c|}{PSNR} & \multicolumn{1}{c|}{CPU} & \multicolumn{1}{c|}{CPU Ratio} & \multicolumn{1}{c|}{$\#$} & \multicolumn{1}{c|}{PSNR} & \multicolumn{1}{c|}{CPU} & \multicolumn{1}{c}{CPU Ratio} \\
          \hline
    \multirow{5}[0]{*}{H-V-cycle} & 128   & 16384  & 85    & 32.63 & 0.72  & $-$ & 78    & 22.93 & 0.84  &  $-$    \\
          & 256   & 65536  & 65    & 35.19 & 2.27  & 3.2  & 61    & 25.43 & 2.01  & 2.4 \\
          & 512   & 262144 & 65    & 38.91 & 9.14  & 4.0& 57    & 28.23 & 8.21  & 4.0 \\
          & 1024  & 1048576 & 63    & 40.65 & 36.94 & 4.0 & 48    & 32.85 & 30.34 & 3.8  \\
          & 2048  & 4194304 & 62    & 41.19 & 148.50& 4.0& 47    & 34.15 & 122.45 & 4.0 \\
          \hline
    \multirow{5}[0]{*}{V-cycle} & 128   & 16384      & 63    & 32.66 & 1.01  &$-$ & 64    & 22.94 & 0.98  & $-$  \\
          & 256   & 65536   & 63    & 35.12 & 2.83  & 2.8 & 58    & 25.42 & 2.68  & 2.7 \\
          & 512   & 262144     & 62    & 38.91 & 11.82 & 4.0& 55    & 28.18 & 10.25 & 4.0  \\
          & 1024  & 1048576     & 60    & 40.68 & 44.01 & 4.0 & 47    & 32.85 & 38.31 & 3.8 \\
          & 2048  & 4194304     & 59    & 41.19 & 179.75 & 4.0& 45    & 34.15 & 152.68 & 4.0  \\
     \hline
     \hline
    \end{tabular}%
   \label{linear}%
\end{table*}%

Table \ref{table1} displays the number of iterations, CPU time, and numerical energies for different combinations of the number of grid layers $J$ and regularization parameter $\alpha$.
As can be seen, both MGMC and MGGC converge to similar numerical energies for a fixed value of $\alpha$. Besides, we also conclude the following two observations
\begin{itemize}
\item[$\bullet$]Introducing the coarse layers can greatly reduce the outer iterations. Much CPU time is saved by increasing the maximum layers from $J=1$ to $J=3$. However, the CPU time increases as $J$ keeps increasing to $J=6$ for all examples.
\item[$\bullet$]The advantage of the multi-grid method is dominant when the regularization parameter $\alpha$ becomes smaller and smaller. The computational time of the single layer method is almost doubled as $\alpha$ decreases from $\alpha=0.1$ to $\alpha=0.03$, while the growth of the multi-grid method is much smaller.
\end{itemize}
Thus, the number of layers is fixed as $J=3$ for both MGMC and MGGC  in the following experiments.

\subsection{Complexity discussion}
We verify the linear convergence of our multi-grid method on both images `Triangle' and `Parrot', the size of which varies as $\{128\times 128,~256\times 256,~512\times 512,~1024\times 1024,~2048\times 2048\}$.
All images are corrupted by Gaussian noises with zero mean and standard deviation $\sigma = 10$.
We set the regularization parameter as $\alpha=0.06$ and the error tolerance as $\epsilon=10^{-5}$.
We implement both V-cycle (fine-to-coarse-to-fine) and half V-cycle (fine-to-coarse).
The comparison results of the number of iterations, PSNR, CPU(s), and CPU ratio are recorded in Table \ref{linear}.
By CPU ratio, it can be checked that the computational time of both the V-cycle and half of the V-cycle is proportional to the size of image $N$, and of complexity $\mathcal{O}(N)$.
For different sizes of images,  half of the V-cycle algorithm always consumes fewer costs than the V-cycle one, especially for images of size $2048\times2048$ down by a sixth, without sacrificing any accuracy.
Therefore, the fine-to-coarse structure is used in our experiments.

\begin{figure*}[t]
      \centering
      \subfloat{
			\includegraphics[width=0.13\linewidth]{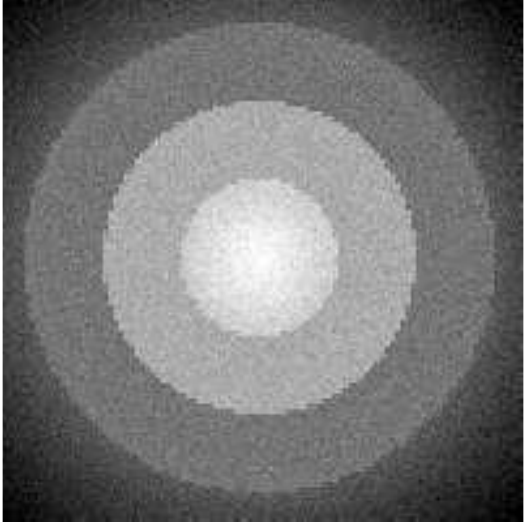}}
      \subfloat{
			\includegraphics[width=0.133\linewidth]{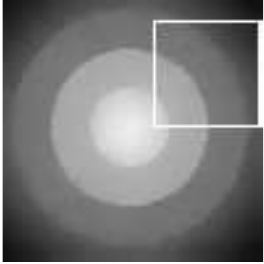}}
      \subfloat{
            \includegraphics[width=0.13\linewidth]{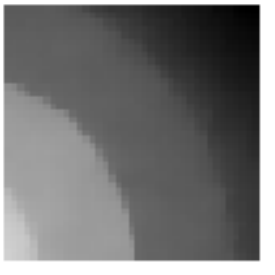}}
     \subfloat{
			\includegraphics[width=0.13\linewidth]{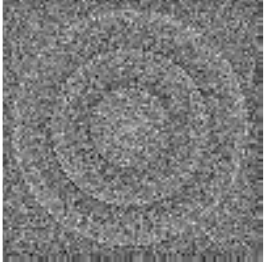}}
      \subfloat{
            \includegraphics[width=0.13\linewidth]{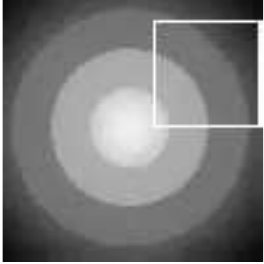}}
      \subfloat{
            \includegraphics[width=0.133\linewidth]{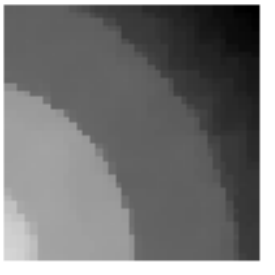}}
      \subfloat{
			\includegraphics[width=0.13\linewidth]{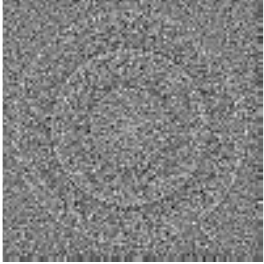}}\\
      \vspace{-0.2cm}
      \setcounter{subfigure}{0}
      \subfloat[\footnotesize{Noisy image}]{
			\includegraphics[width=0.13\linewidth]{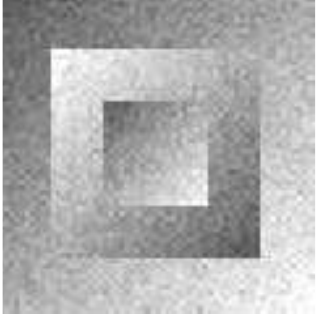}}
      \subfloat[\footnotesize{Euler's elastica}]{
			\includegraphics[width=0.13\linewidth]{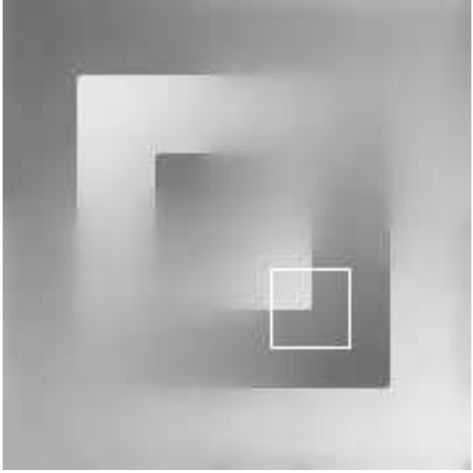}}
      \subfloat[\footnotesize {Zoom region}]{
              \includegraphics[width=0.13\linewidth]{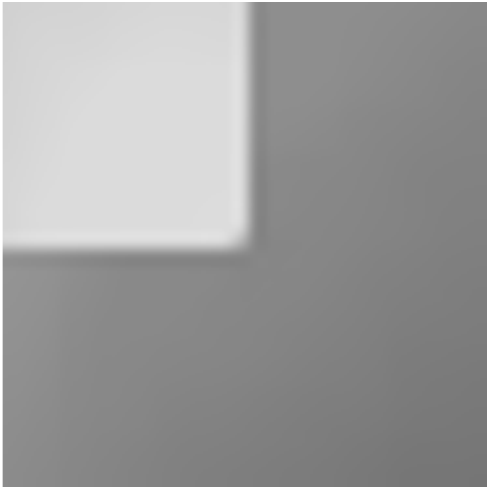}}
      \subfloat[\footnotesize{Residual image}]{
			\includegraphics[width=0.13\linewidth]{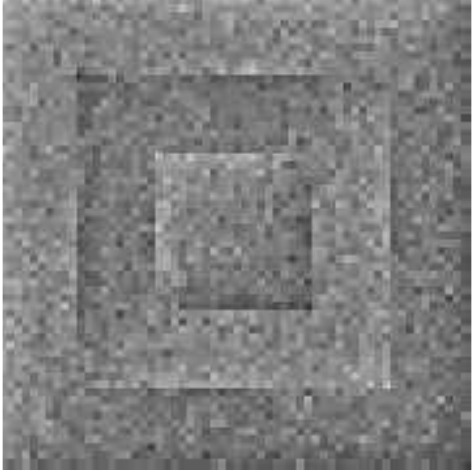}}
      \subfloat[\footnotesize{MGGC}]{
                 \includegraphics[width=0.13\linewidth]{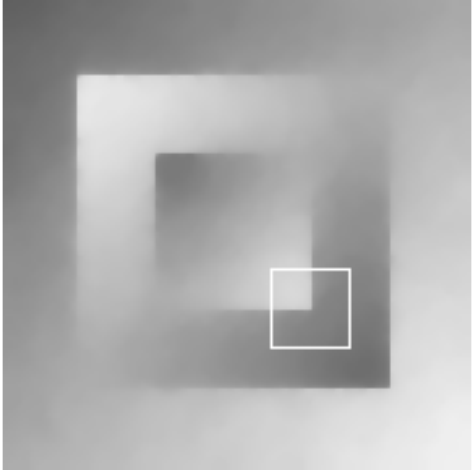}}
      \subfloat[\footnotesize{Zoom region}]{
			  \includegraphics[width=0.13\linewidth]{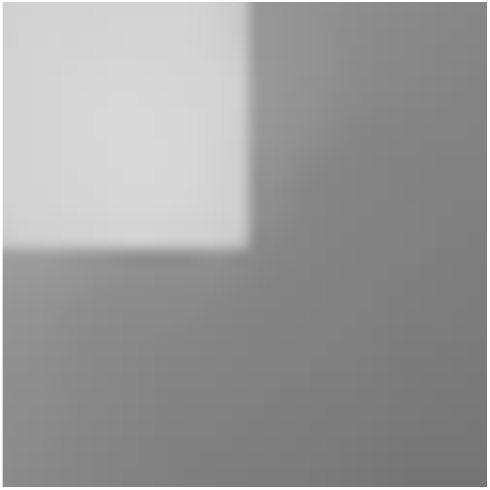}}
      \subfloat[\footnotesize{Residual image}]{
			\includegraphics[width=0.13\linewidth]{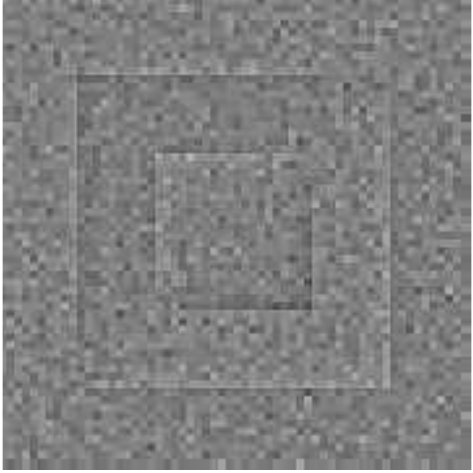}}\\
	\caption{The denoising results of the smooth images A1  and A2 (from top to bottom) obtained by the Euler`s elastica \cite{2011AA} and our Gaussian curvature model, where we set the regularization parameter as $\alpha=0.06$ and the error tolerance as $\epsilon=10^{-4}$.}
	\label{smoothimages}
\end{figure*}

\begin{figure*}[t]
      \centering
      \subfloat[Clean image]{
			\includegraphics[width=0.15\linewidth]{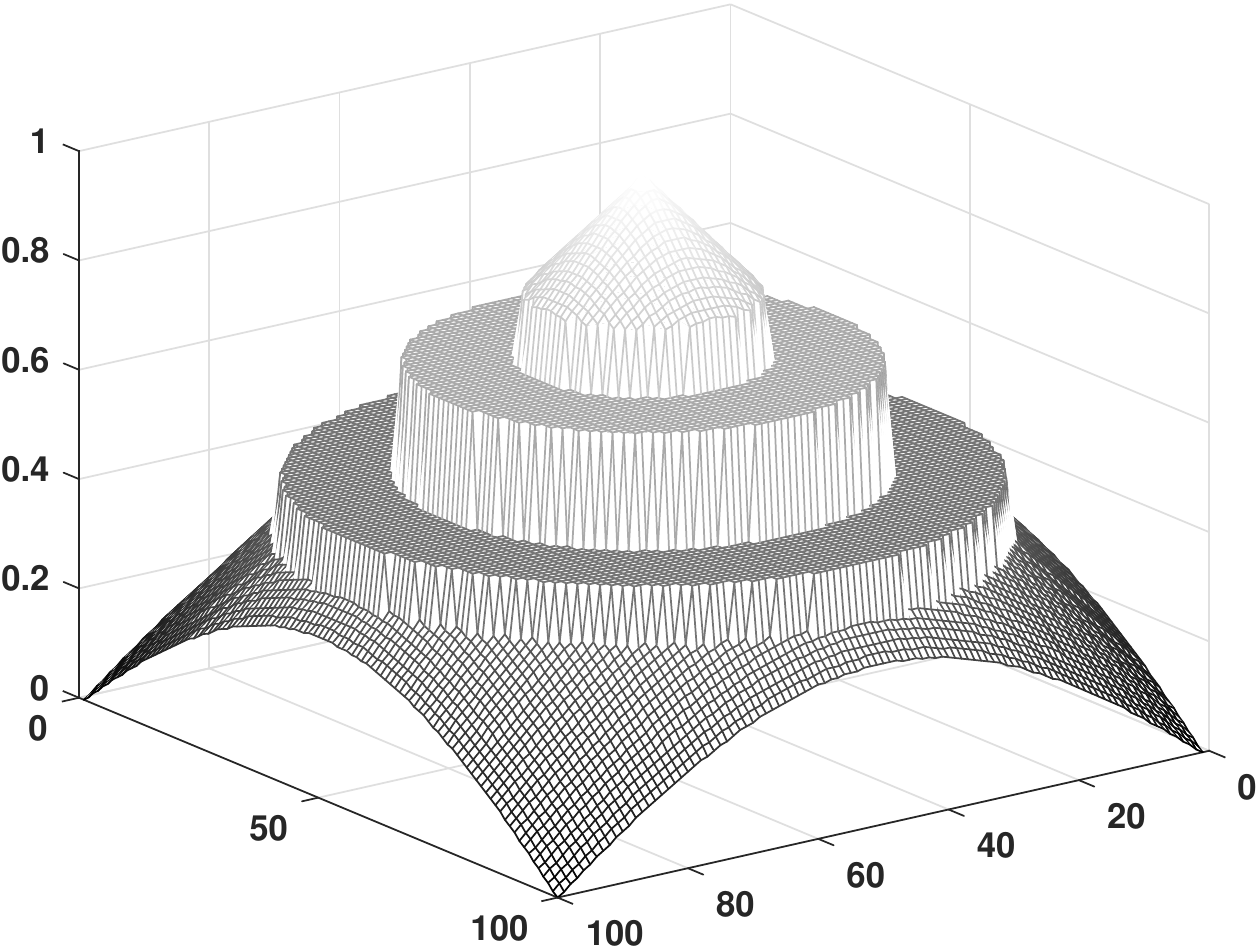}}
      \subfloat[Euler's elastica]{
			\includegraphics[width=0.15\linewidth]{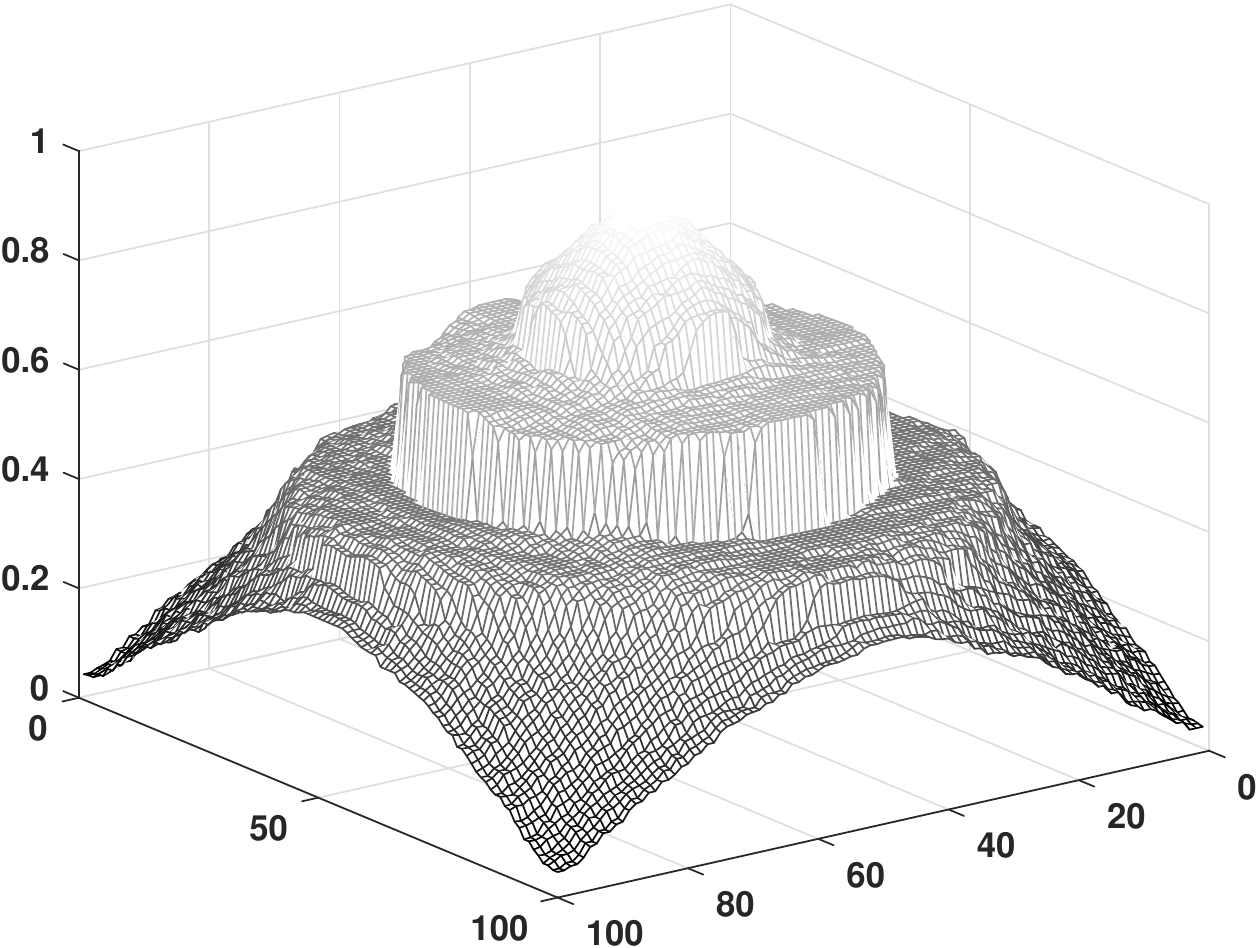}}
      \subfloat[MGGC]{
            \includegraphics[width=0.15\linewidth]{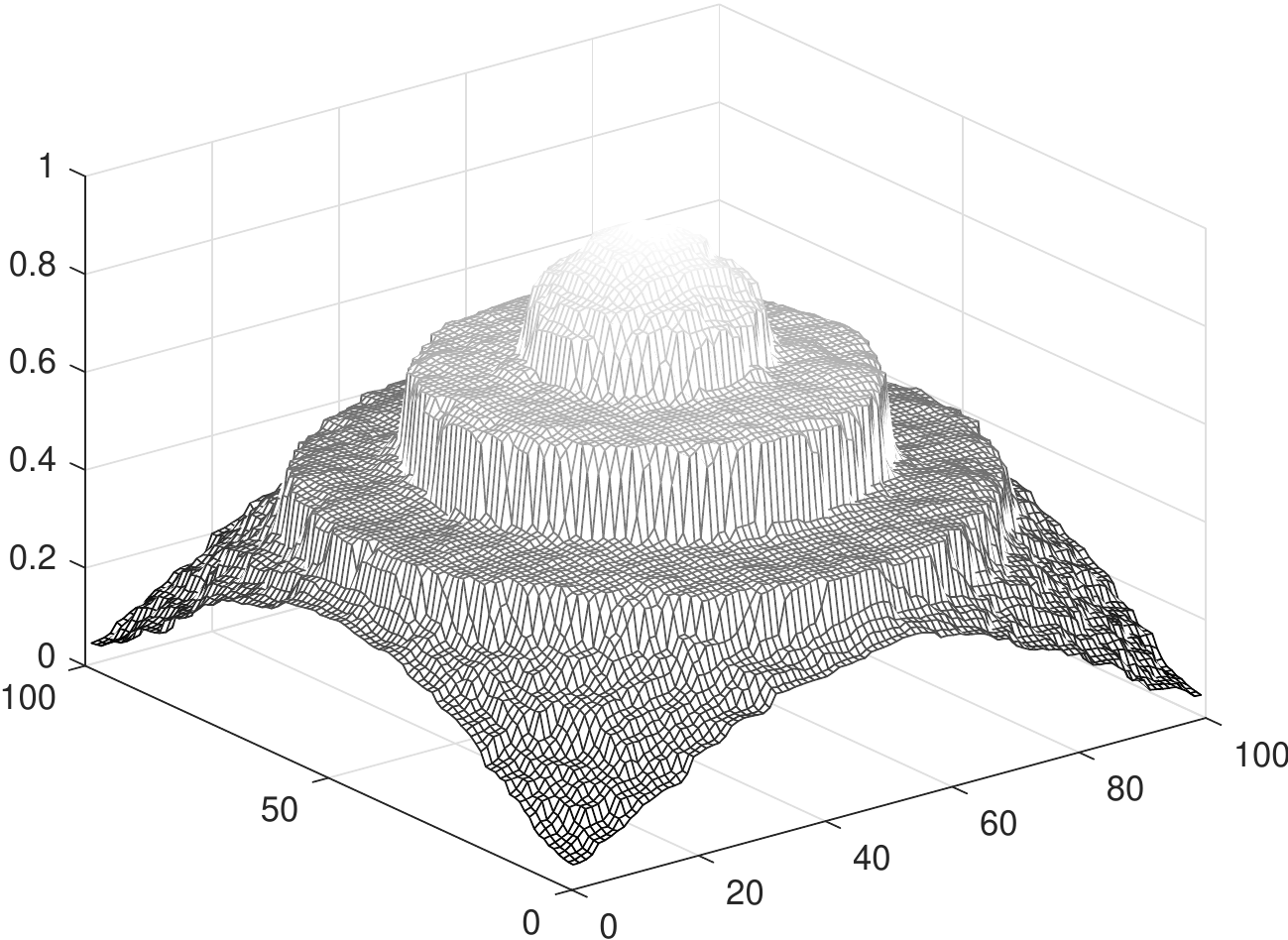}}
      \subfloat[Clean image]{
			\includegraphics[width=0.15\linewidth]{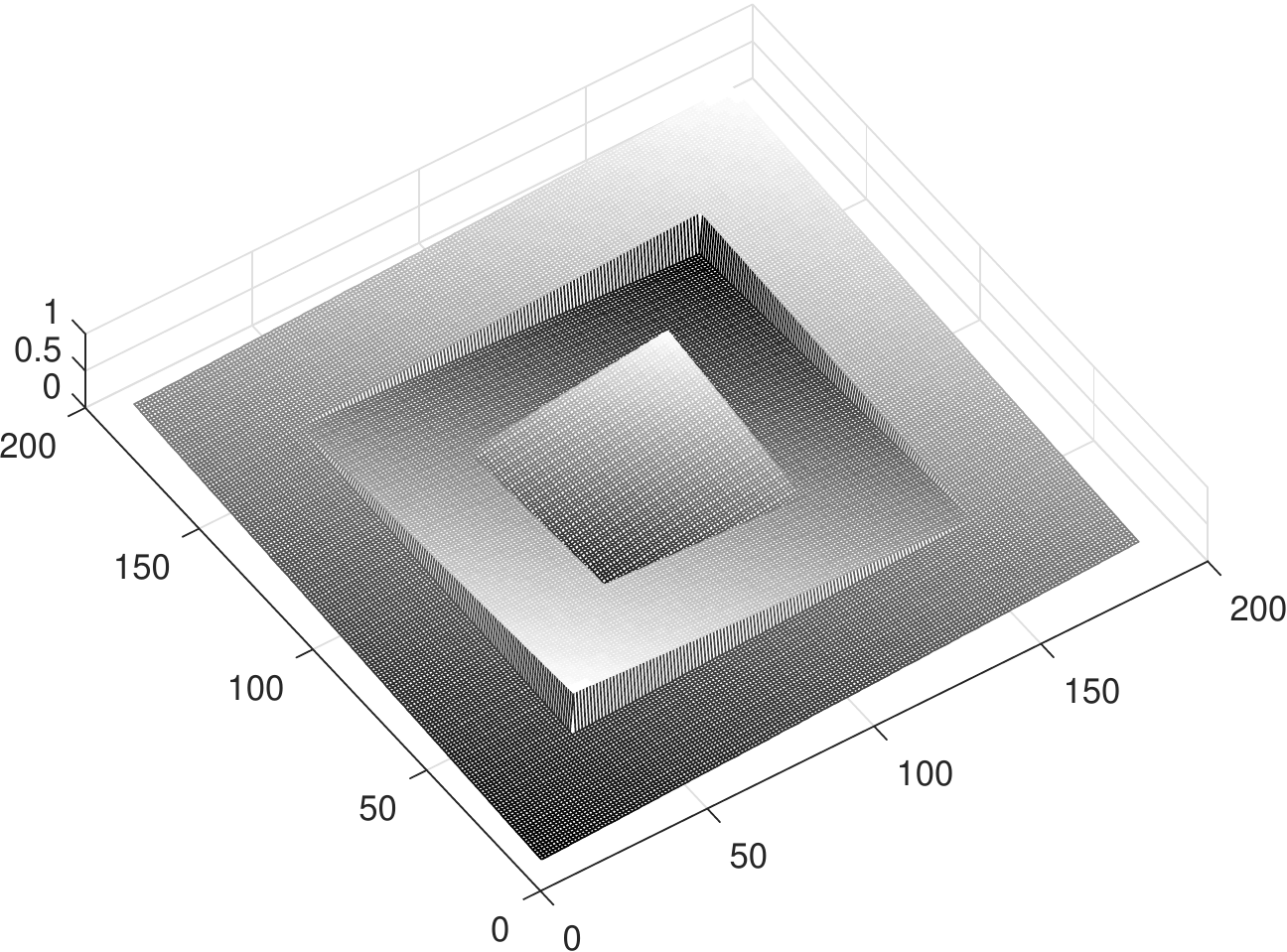}}
      \subfloat[Euler's elastica]{
			\includegraphics[width=0.15\linewidth]{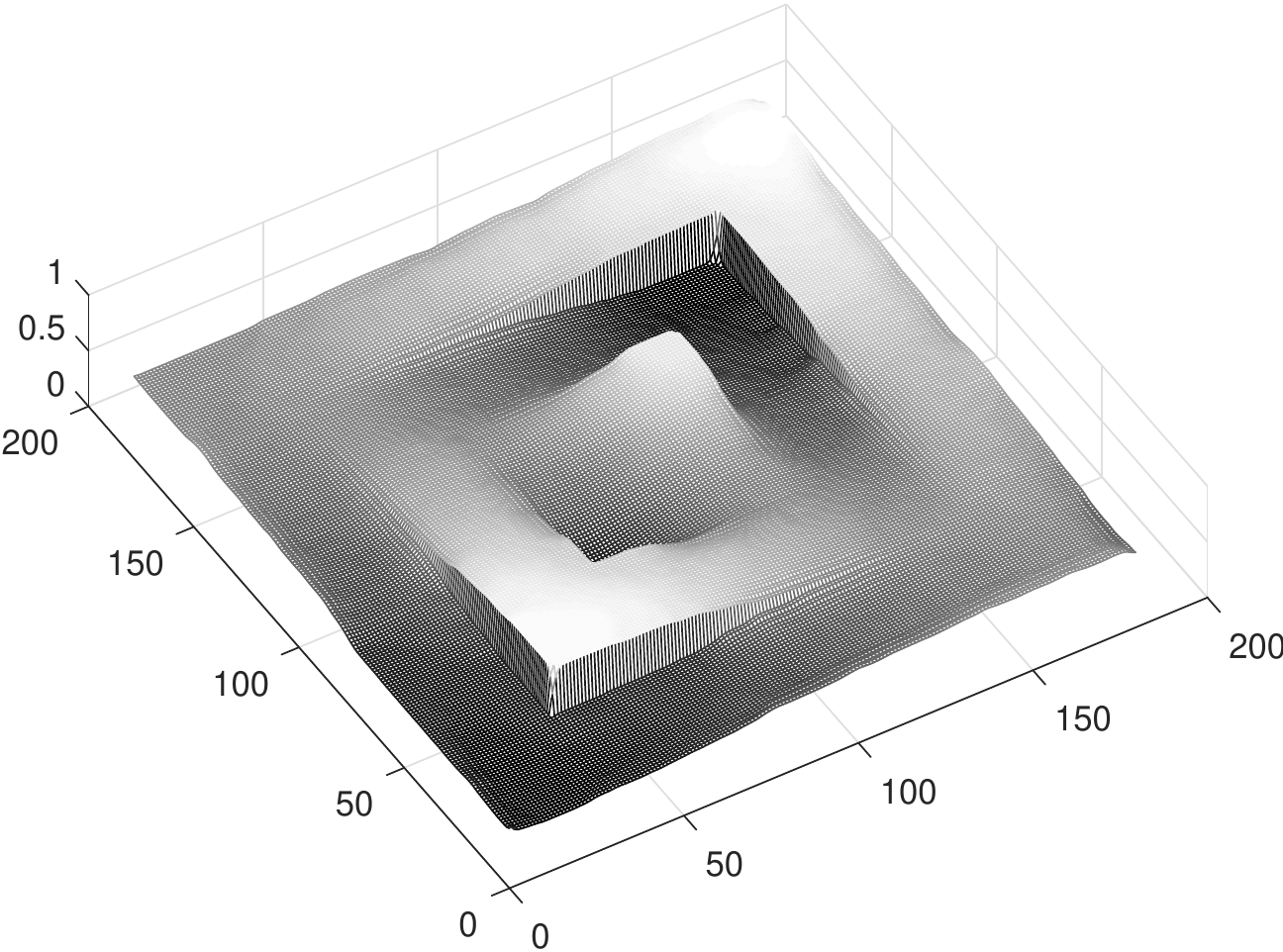}}
      \subfloat[MGGC]{
              \includegraphics[width=0.15\linewidth]{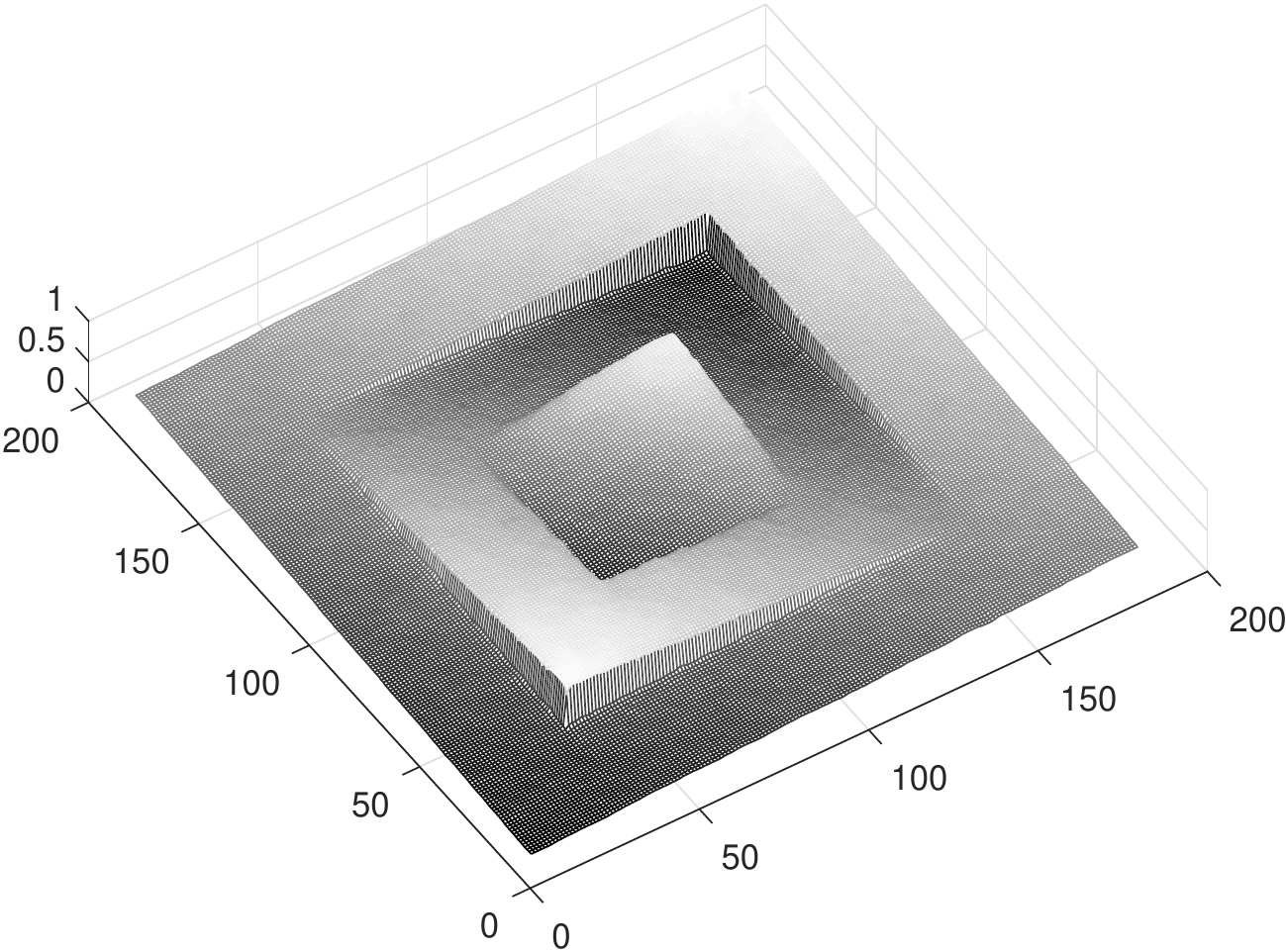}}
	\caption{The image surfaces of the clean images and restoration images were obtained by Euler's elastic regularization model \cite{2011AA} and  our Gaussian curvature regularization model.}
	\label{smoothsurface}
\end{figure*}

\begin{table*}[!htbp]
  \centering
\footnotesize
  \caption{The comparison of image denoising between the multi-grid method and mean curvature filter for the noise level $\sigma=10$.}
    \begin{tabular}{p{20pt}|p{8pt}|p{8pt}p{15pt}p{16pt}|p{8pt}p{15pt}p{16pt}|p{8pt}p{15pt}p{16pt}|p{5pt}p{5pt}|p{5pt}p{5pt}|p{5pt}p{5pt}}
    \hline
    \hline
  \multicolumn{2}{c|}{Curvature}&\multicolumn{9}{c|}{Mean curvature}&\multicolumn{6}{c}{Gaussian curvature}\\
     \hline
           & \multirow{2}[0]{*}{Im} & \multicolumn{3}{c|}{$\alpha$=0.1} & \multicolumn{3}{c|}{$\alpha$=0.06} & \multicolumn{3}{c|}{$\alpha$=0.03}& \multicolumn{2}{c|}{$\alpha$=0.1} & \multicolumn{2}{c|}{$\alpha$=0.06} & \multicolumn{2}{c}{$\alpha$=0.03} \\
          \cline{3-17}
          &       & \multicolumn{1}{l}{\scriptsize{CFMC}} & \multicolumn{1}{c}{\scriptsize{MG}}&\multicolumn{1}{c|}{\scriptsize{MGMC}} & \multicolumn{1}{l}{\scriptsize{CFMC}} & \multicolumn{1}{c}{\scriptsize{MG}}&\multicolumn{1}{c|}{\scriptsize{MGMC}} & \multicolumn{1}{l}{\scriptsize{CFMC}} & \multicolumn{1}{c}{\scriptsize{MG}}&\multicolumn{1}{c|}{\scriptsize{MGMC}} & \multicolumn{1}{l}{\scriptsize{CFGC}}&\multicolumn{1}{c|}{\scriptsize{MGGC}}& \multicolumn{1}{l}{\scriptsize{CFGC}}&\multicolumn{1}{c|}{\scriptsize{MGGC}}& \multicolumn{1}{l}{\scriptsize{CFGC}}&\multicolumn{1}{c}{\scriptsize{MGGC}}\\
          \hline
 \multirow{4}[0]{*}{\footnotesize{PSNR}}
          & $\#1$  & 40.43 &41.81 &\textbf{42.24} & 40.88 &42.83 &\textbf{43.01} & 40.82&\textbf{42.61} &42.58 & 38.15& \textbf{39.84}    & 39.01 & \textbf{41.61}  & 38.96& \textbf{41.09}\\
          & $\#2$  & 33.54& 35.42&\textbf{35.49}& 33.67 &35.57 &\textbf{35.61}  & 33.24&35.39 &\textbf{35.52} & 33.67  & \textbf{35.64}  & 33.87& \textbf{35.85} & 33.19  & \textbf{35.41}\\
          & $\#3$  & 30.19&32.17 &\textbf{32.23}& 30.21  & 32.84&\textbf{32.85}  & 29.67&32.12 &\textbf{32.24}& 30.52& \textbf{32.91}   & 31.89 & \textbf{33.92}  & 30.44& \textbf{31.83} \\
          & $\#4$    & 30.17&\textbf{32.51} &32.42  & 30.26& 32.62&\textbf{32.68}  & 29.46& 31.37&\textbf{32.25}& 29.49 & \textbf{31.66}  & 29.84 & \textbf{32.23}  & 29.41 & \textbf{30.44} \\
                 \hline
    \multirow{4}[0]{*}{\footnotesize{Iter}}
          & $\#1$    & 118&6 &55  & 214& 8 &62 & 346&8 &67 & 154 & 45  & 253 & 55  & 363& 73  \\
          & $\#2$  &125 &7 &56  & 220&8  & 64 & 325&8 &95 &114 & 47   & 268& 59  & 335& 76 \\
          & $\#3$  & 134 &5 &38 & 199&6  & 51& 362 &6 &59& 124& 44  & 191 & 50  & 351 & 54\\
          & $\#4$    & 135&6 &43  & 208& 8& 50& 382&8 &53 & 124 & 49 & 186 & 56  & 350 &60 \\
       \hline
    \multirow{4}[0]{*}{\footnotesize{CPU(s)}}
          & $\#1$    & 9.21&110.25&\textbf{5.39}   & 10.15& 124.31 &\textbf{5.47}  & 17.21& 148.01&\textbf{8.67} & 9.31  & \textbf{3.66}  & 13.25 & \textbf{4.07}  & 20.13& \textbf{5.82}  \\
          & $\#2$  & 10.24 &116.07 &\textbf{7.56}  & 13.21&125.48  & \textbf{9.37} & 19.62&130.24 &\textbf{13.64}& 9.12 & \textbf{6.65}   & 14.13& \textbf{8.14}  & 16.12 & \textbf{10.46} \\
          & $\#3$  & 27.15 &475.21 &\textbf{23.24} & 45.15&482.68  & \textbf{26.86}& 51.64 &499.21 &\textbf{28.71}& 34.41 & \textbf{27.30}  & 47.41 & \textbf{30.58}  & 80.46 & \textbf{33.96}\\
          & $\#4$    & 28.32& 463.93&\textbf{26.36}  & 50.87&550.68 & \textbf{28.68}& 60.85& 562.41&\textbf{29.32} & 41.12 & \textbf{31.75}  & 55.14 & \textbf{34.67}  & 67.12 & \textbf{37.00}\\
        \hline
   \multirow{4}[0]{*}{\footnotesize {Energy}}
          & $\#1$    & 2.4556 &2.2946 &\textbf{2.2835}  & 2.4871&2.4087 & \textbf{2.4003} & 2.5577&\textbf{2.4982} &2.5071 & 1.9186  & \textbf{1.8465}   & 2.1904 & \textbf{2.1009}  & 2.3539 & \textbf{2.2811} \\
          & $\#2$  & 2.3478&2.2746 &\textbf{2.2639}  & 2.5609&2.4615  & \textbf{2.4557}& 2.5502 &2.4821 &\textbf{2.4721}& 1.9261   & \textbf{1.8641} & 2.2485& \textbf{2.1510}  & 2.5352& \textbf{2.3635} \\
          & $\#3$  & 1.2409 &1.1588 &\textbf{1.1512} & 1.4527&1.3698 & \textbf{1.3646} & 1.6183  &1.5213 &\textbf{1.5155}& 8.8509 & \textbf{8.7352}   & 1.2284& \textbf{1.1064}   & 1.4294 & \textbf{1.3536}\\
          & $\#4$    & 1.3415 &1.2598 &\textbf{1.2588}   & 1.6692&1.5435 & \textbf{1.5371} & 1.8582&1.7998 &\textbf{1.7259}  & 9.5501 & \textbf{9.4709}    & 1.3216 & \textbf{1.2615}  & 1.4164& \textbf{1.3406} \\
        \hline
        \hline
  \end{tabular}%
  \label{table2}%
\end{table*}%

\subsection{Properties of curvature regularization}
Both mean curvature and Gaussian curvature are well-known for their abilities in preserving image contrast and structural features \cite{Zhu2012image}\cite{2015Image}. Now we use our Gaussian curvature regularization model as an example and compare it with the Euler`s elastica regularization model on two synthetic images. As shown in figure \ref{smoothimages},
both images are corrupted by Gaussian noise with zero mean and standard deviation $\sigma = 20$.
The restoration images demonstrate that our model outperforms Euler`s elastica \cite{2011AA} in preserving edges and corners. Moreover, the residual images obtained by Euler's elastica also contain more image information than ours, which confirms Gaussian curvature regularization is better at maintaining image contrast.
In addition, we display the image surface plots of clean images and restored images of Euler's elastica and MGGC in Figure \ref{smoothsurface},
where our Gaussian curvature regularization model effectively keeps the sharp corners and jumps.

\subsection{Comparison with curvature filters}\label{test_speed}
In what follows, we verify the advantages of the proposed multi-grid method by comparing it with the multi-grid method in \cite{2010Multigrid} and curvature filter \cite{Gong2019mean}.
Note that the domain decomposition method has been applied to the curvature filter for a fair comparison. We degrade the test images in Fig. \ref{example} by the white Gaussian noises with zero mean and standard deviation $\sigma=10$.
The regularization parameter $\alpha$ are set as $\alpha\in\{0.1~0.06~0.03\}$ and error tolerance is fixed as $\epsilon=10^{-6}$ for all methods.
There are no other parameters for the mean curvature filter (CFMC) and Gaussian curvature filter (CFGC), where both methods are solved by gradient flow as presented in \cite{Gong2019mean}. The parameters of multi-grid method \cite{2010Multigrid} (denoted by MG) are set as: the total number of iterations is $10$, the maximal level is $3$, the stopping condition is set as \eqref{reerr}, and all other parameters are set as suggested by the paper.

Table \ref{table2} records the PSNR, the number of iterations, CPU time, and the numerical energies obtained by different approaches.
As can be seen, our method always achieves higher PSNR and smaller energies than the curvature filter, while providing better or similar PSNR as the MG method.
More importantly, much CPU time can be saved by our fine-to-coarse strategy, especially for images of large scales and regularization parameters.
We notice that multi-grid method \cite{2010Multigrid} is very time consuming for solving the high-order PDEs.
Obviously, our multi-grid method can well balance efficiency and effectiveness.

More than that, we evaluate and compare the performance of the multi-grid method and curvature filter on images corrupted by different noise levels, i.e., $\sigma\in \{10,20,30\}$, where $\alpha$ is chosen to achieve the best restoration results. As provided in Table \ref{table3}, our multi-grid method always outperforms the curvature filter in both image quality and computational efficiency. The main reason behind this is that both mean curvature and Gaussian curvature in our model are estimated by the definitions in differential geometry. To make it more clear, we present one representative restoration result in Fig. \ref{fig6}, where the results of the one-layer multi-grid methods are also illustrated for comparison. It can be observed the one-layer multi-grid methods produce much better results than curvature filters with much smoother details. And the multi-grid strategy can further improve the restoration quality.

 \begin{table*}[!htbp]
  \centering
  \footnotesize
  \caption{The comparison between the multi-grid method and curvature filter with different noise levels of $\sigma=10,\ 20,\ 30$, respectively.}
       \begin{tabular}{c|c|cc|cc|cc|cc|cc|cc}
    \hline
    \hline
          & \multirow{2}[0]{*}{$\alpha$} & \multicolumn{2}{c|}{$\sigma=10$, $\alpha=0.06$} & \multicolumn{2}{c|}{$\sigma=20$, $\alpha=0.05$} & \multicolumn{2}{c|}{$\sigma=30$, $\alpha=0.04$} &\multicolumn{2}{c|}{$\sigma=10$, $\alpha=0.06$} & \multicolumn{2}{c|}{$\sigma=20$, $\alpha=0.05$} & \multicolumn{2}{c}{$\sigma=30$, $\alpha=0.03$}\\
          \cline{3-14}
          &       & \multicolumn{1}{l}{CFMC} & \multicolumn{1}{l|}{MGMC} & \multicolumn{1}{l}{CFMC} & \multicolumn{1}{l|}{MGMC} & \multicolumn{1}{l}{CFMC} & \multicolumn{1}{l|}{MGMC}&\multicolumn{1}{l}{CFGC} & \multicolumn{1}{l|}{MGGC} & \multicolumn{1}{l}{CFGC} & \multicolumn{1}{l|}{MGGC} & \multicolumn{1}{l}{CFGC} & \multicolumn{1}{l}{MGGC} \\
          \hline
    \multirow{4}[0]{*}{PSNR}
           & $\#1$  & 40.88 & \textbf{43.01} & 35.17 & \textbf{38.11}  & 33.87& \textbf{35.53}& 39.01 & \textbf{41.61}   & 34.51 & \textbf{37.33} & 32.31& \textbf{35.19} \\
          &  $\#2$ & 33.67 & \textbf{35.61} & 30.95 & \textbf{33.37} & 29.75& \textbf{31.51}  & 33.87& \textbf{35.85}  & 31.09 & \textbf{33.52} & 29.89& \textbf{31.33}\\
          & $\#3$  & 30.21 & \textbf{32.85} & 28.22 & \textbf{30.48} & 26.94 & \textbf{29.11} & 31.89 & \textbf{33.92} & 27.97& \textbf{29.61}  & 26.21 & \textbf{28.14}\\
          & $\#4$    & 30.26 & \textbf{32.68} & 27.39  & \textbf{29.51} & 26.28& \textbf{28.35}& 29.84& \textbf{32.23}  & 27.18 & \textbf{28.55} & 25.62 & \textbf{27.22} \\
        \hline
    \multirow{4}[0]{*}{CPU(s)}
          &  $\#1$  & 10.15 & \textbf{5.47}  & 13.11& \textbf{5.88}  & 21.41 & \textbf{6.42} & 13.25 & \textbf{4.07}  & 18.95& \textbf{6.88}  & 34.45  & \textbf{7.29} \\
          &  $\#2$   & 13.21& \textbf{9.37}  & 15.41 & \textbf{10.34} & 24.47& \textbf{14.78} & 14.13 & \textbf{8.14} & 19.25 & \textbf{11.17} & 37.76& \textbf{10.47} \\
          & $\#3$ & 45.15 & \textbf{26.86}  & 67.97& \textbf{26.42}  & 72.13 & \textbf{33.85}& 47.41& \textbf{30.58} & 79.94  & \textbf{30.23} & 82.72 & \textbf{39.37}\\
          &$\#4$    & 50.87 & \textbf{28.68} & 70.98 & \textbf{29.42} & 78.67 & \textbf{35.11} & 55.14 & \textbf{34.67} & 97.42 & \textbf{36.41} & 98.42 & \textbf{40.45}\\
       \hline
    \multirow{4}[0]{*}{Energy}
           & $\#1$ & 2.4871& \textbf{2.4003}  & 8.2315 & \textbf{8.1815} & 2.0181& \textbf{2.0104}  & 2.1904& \textbf{2.1009}  & 8.3539 & \textbf{8.2546} & 1.9741 & \textbf{1.9387} \\
          &  $\#2$  & 2.5609& \textbf{2.4557}  & 8.5179 & \textbf{8.3194} & 2.0155& \textbf{1.9791} & 2.2485& \textbf{2.1575} & 8.4884  & \textbf{8.3789} & 1.9833 & \textbf{1.9257} \\
          &  $\#3$   & 1.4527 & \textbf{1.3646} & 3.8236& \textbf{3.6539}  & 7.1025 & \textbf{6.9583} & 1.2284 & \textbf{1.1064} & 3.8313& \textbf{3.7934}& 8.5485  & \textbf{8.3979}\\
          &  $\#4$     & 1.6692 & \textbf{1.5371} & 3.9651 & \textbf{3.7614} & 6.8714 & \textbf{6.6511} & 1.3216 & \textbf{1.2615} & 3.9682& \textbf{3.9344}  & 8.7079& \textbf{8.6038} \\
     \hline
     \hline
    \end{tabular}
  \label{table3}
\end{table*}%

\begin{figure}[!htbp]
\centering
\subfloat[\small{CFMC ($29.75$ dB)}]{\includegraphics[width=0.23\textwidth]{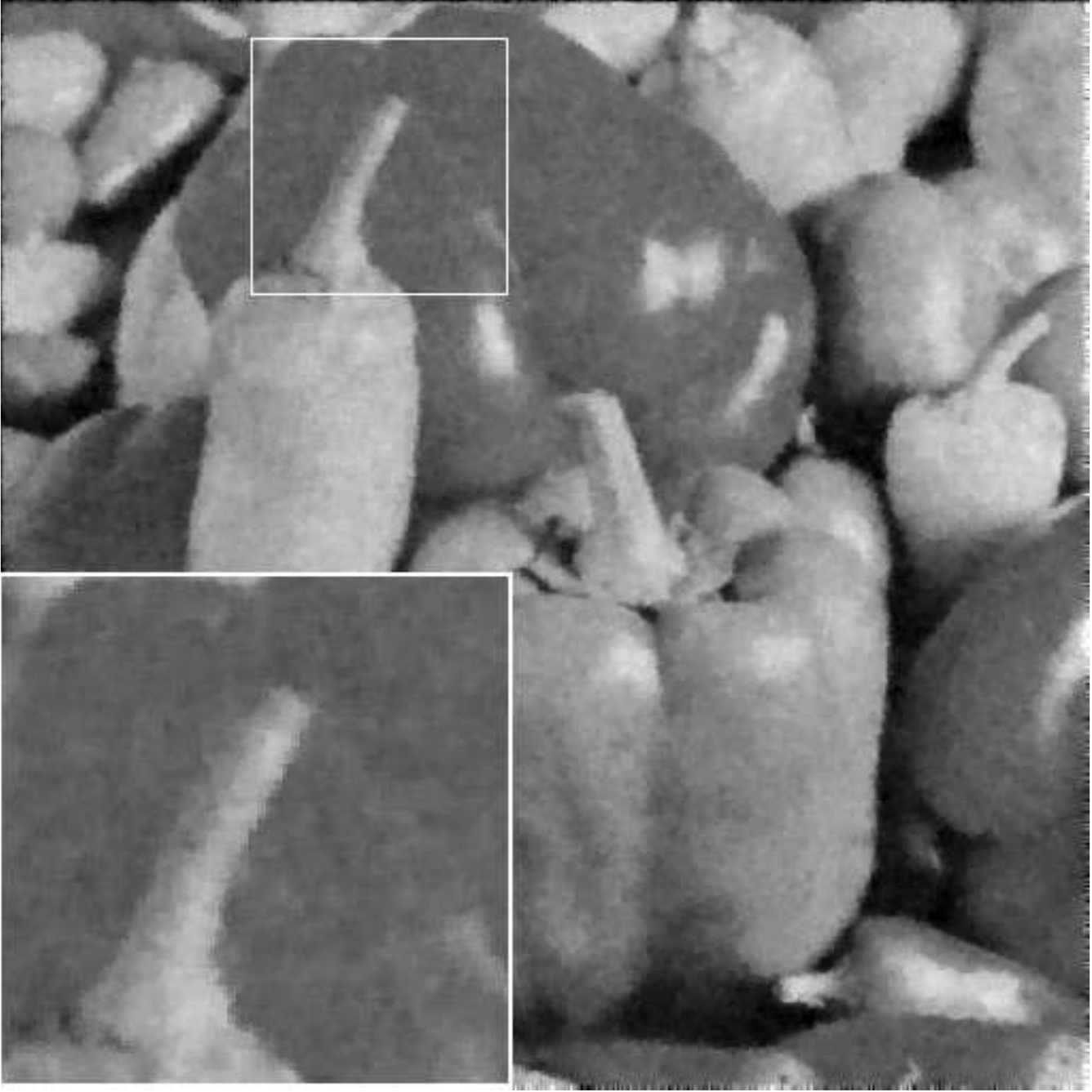}}
\hspace{-1mm}
\subfloat[\small{CFGC ($29.89$ dB)}]{\includegraphics[width=0.23\textwidth]{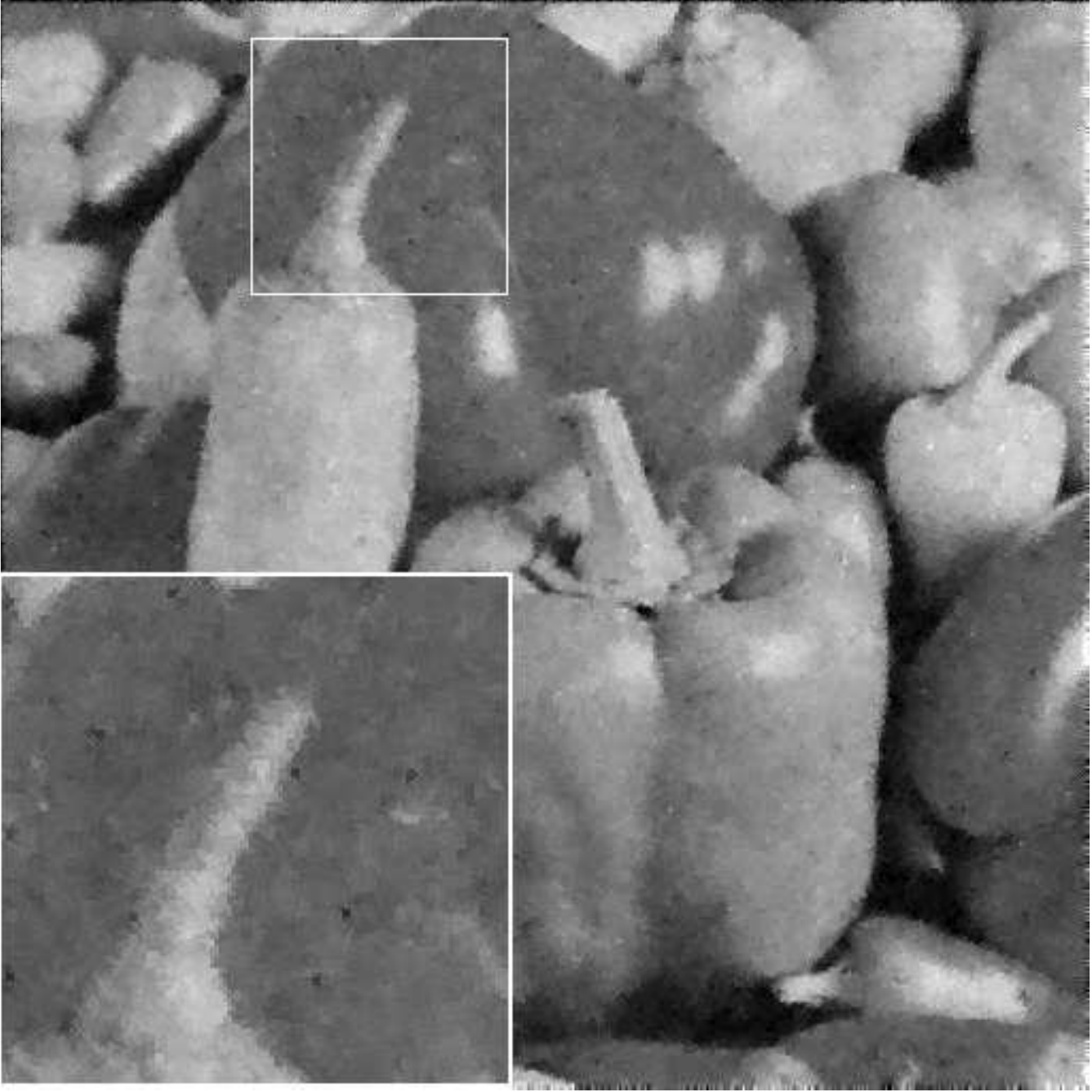}}\\
\vspace{-2mm}
\subfloat[\small{one-layer MGMC ($31.25$ dB)}]{\includegraphics[width=0.23\textwidth]{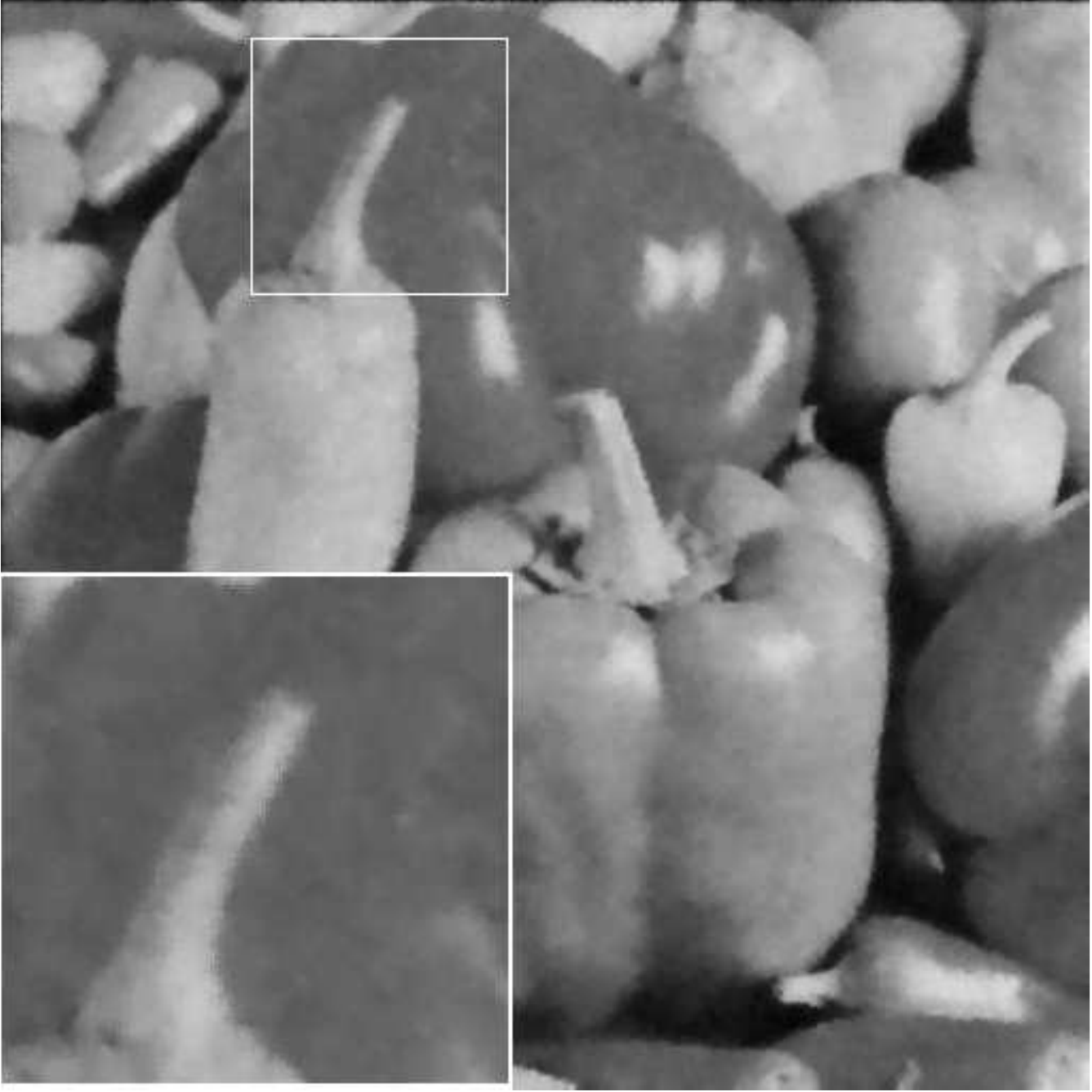}}
\hspace{-1mm}
\subfloat[\small{one-layer MGGC ($31.04$ dB)}]{\includegraphics[width=0.23\textwidth]{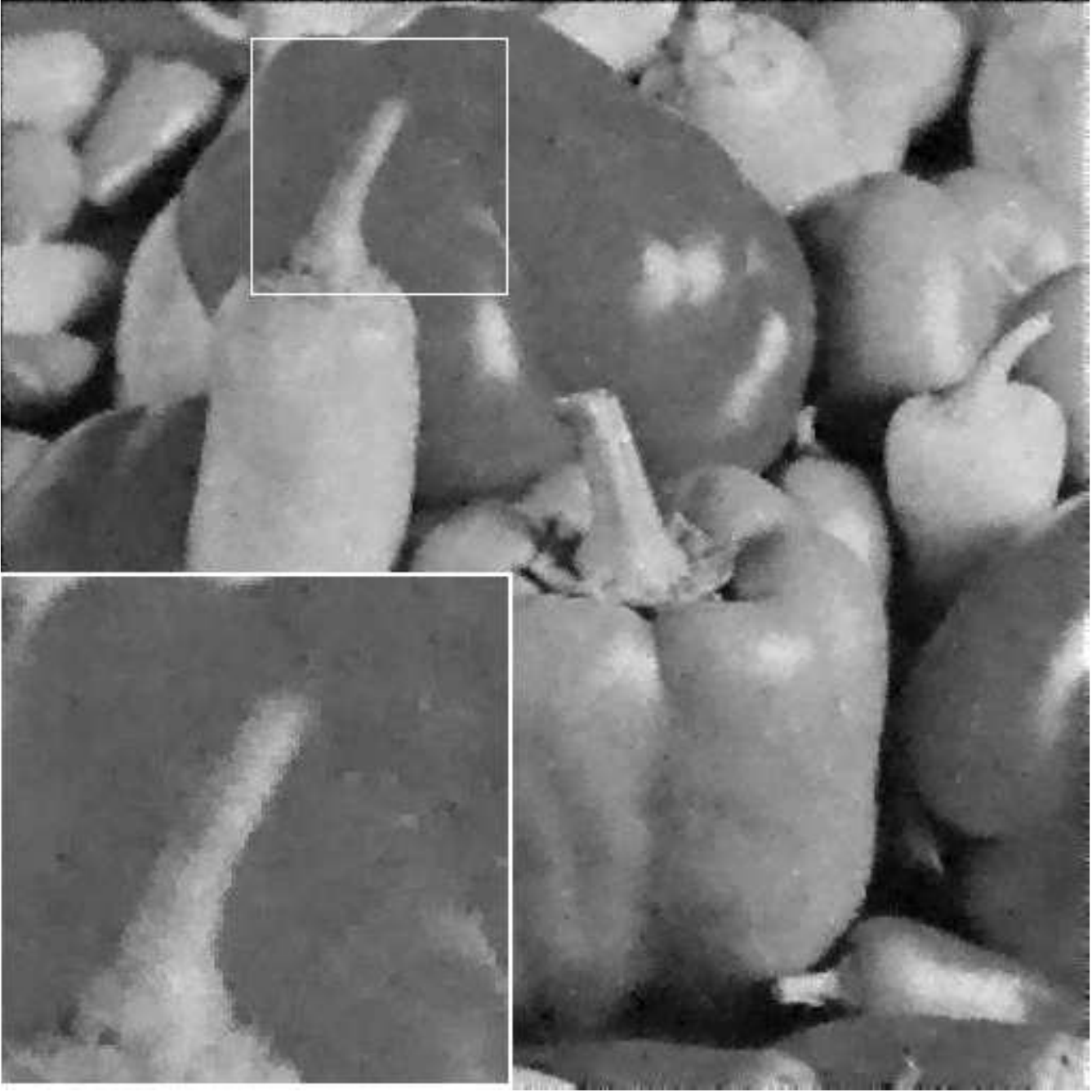}}\\
\vspace{-2mm}
\subfloat[\small{MGMC ($31.51$ dB)}]{\includegraphics[width=0.23\textwidth]{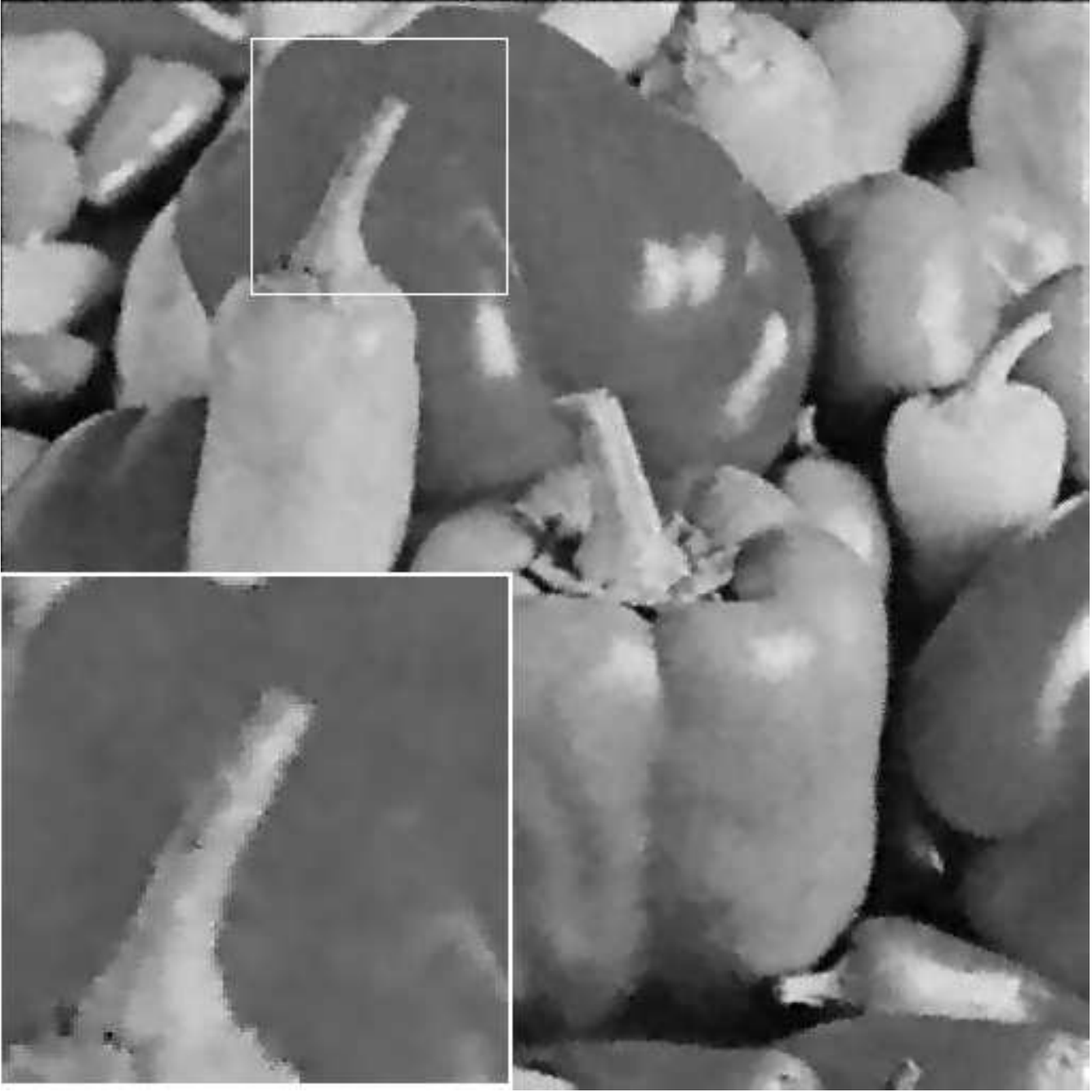}}
\hspace{-1mm}
\subfloat[\small{MGGC ($31.33$ dB)}]{\includegraphics[width=0.23\textwidth]{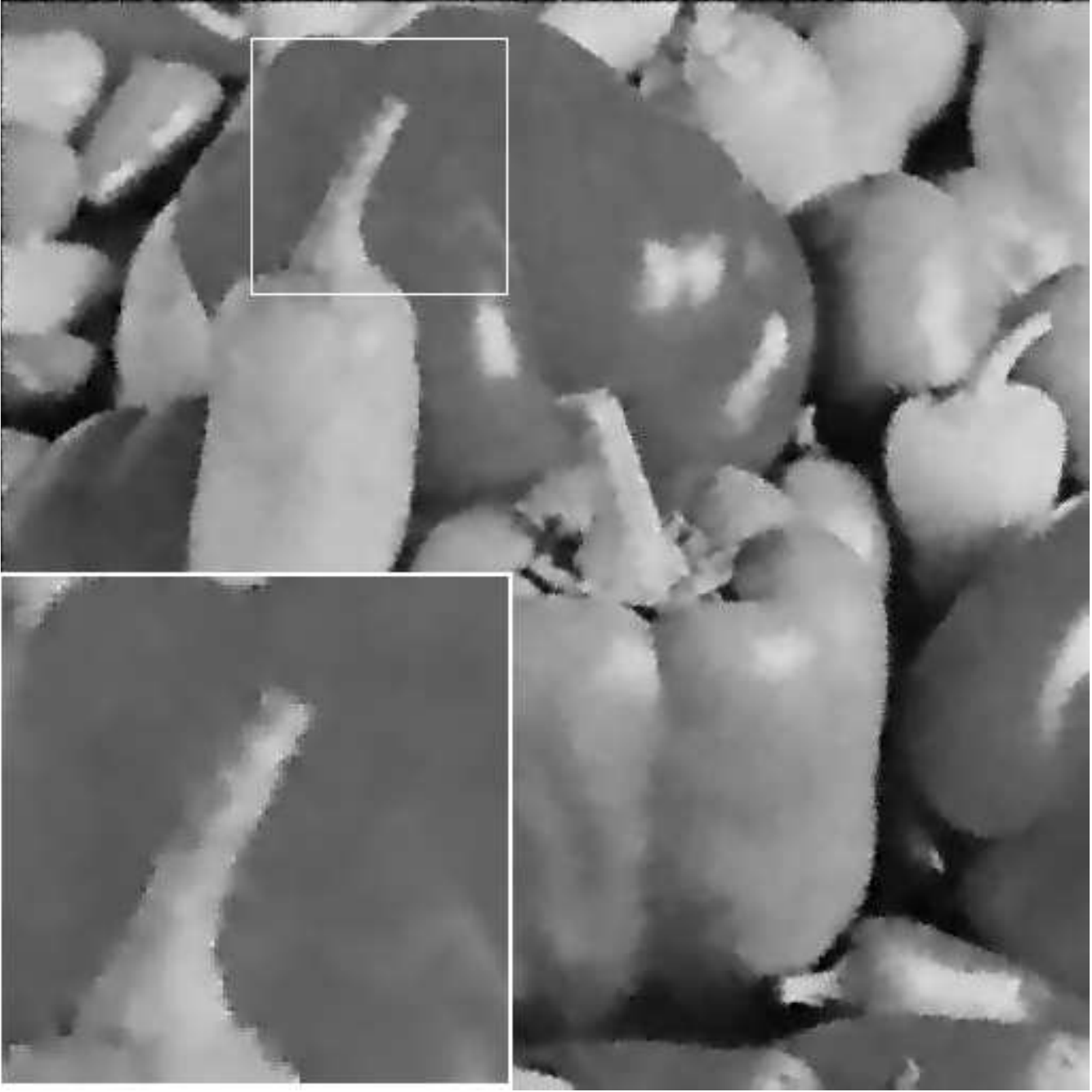}}
\caption{The denoised results obtained by curvature filter and our method on image `Pepper' with noise level $\sigma=30$.}\label{fig6}
\vspace{-0.5cm}
\end{figure}

\subsection{Comparison study and GPU implementation}

In this subsection, we compare our mean curvature method with several state-of-the-art image denoising methods on a dataset containing 30 gray images as shown in Fig. \ref{30images}, where all images are degraded by white Gaussian noise of zero mean and standard deviation $\sigma=20$. All methods are terminated by the relative error in $u$ as given below
\begin{equation}
RelErr(u_{l+1})=\|u_{l+1}-u_{l}\|_1/\|u_{l+1}\|_1\leq\epsilon,
\end{equation}
where $\epsilon$ is fixed as $\epsilon=10^{-4}$. The parameters of the comparison methods are provided as follows
\begin{enumerate}
\item[1)]CFMC \cite{Gong2019mean}: The regularization parameter is set as $\alpha=0.05$.
\item[2)]Hybrid first and second order model (TVTV$^2$) \cite{papafitsoros2014combined}: The parameters are set as $\alpha=0.2$, $\beta=0.2$, $p_1=1$, $p_2=2$ and $p_3=20$.
\item[3)] Total generalized variation (TGV) method \cite{Bredies2010total}: The parameters are selected as suggested in the corresponding papers, which are set as $\alpha_0=1.5$, $r_4=1.5\times 10^5$ and $\lambda=10$.
\item[4)]Euler's elastica model (denoted as Euler) \cite{2011AA}: The parameters of ALM are fixed as $a=1$, $b=10$, $\mu=10$, $\eta=200$, $r_1=2$, $r_2=200$, $r_4=250$.
\item[5)]Mean curvature model (denoted as MC) \cite{2017Augmented}: The parameters of ALM are set as $\varepsilon=0.4$, $\lambda=1600$, $r_1=200$, $r_2=200$, $r_3=1\times 10^4$ and $r_4=1\times 10^4$.
\item[6)]Total absolute mean curvature model (denoted as TAC) \cite{Zhong2020image}: The model is solved by the ADMM-based algorithm with the parameters given by $a=1$, $b=0.4$, $\lambda=0.09$, and $r=2$.
\item[7)]MGMC: The regulation parameter is set as $\alpha=0.05$.
The iteration number for solving the sub-problem $J(c)$ is fixed as $15$.
 \end{enumerate}

\begin{figure*}[!htbp]
  \centering
    \includegraphics[width=1\textwidth]{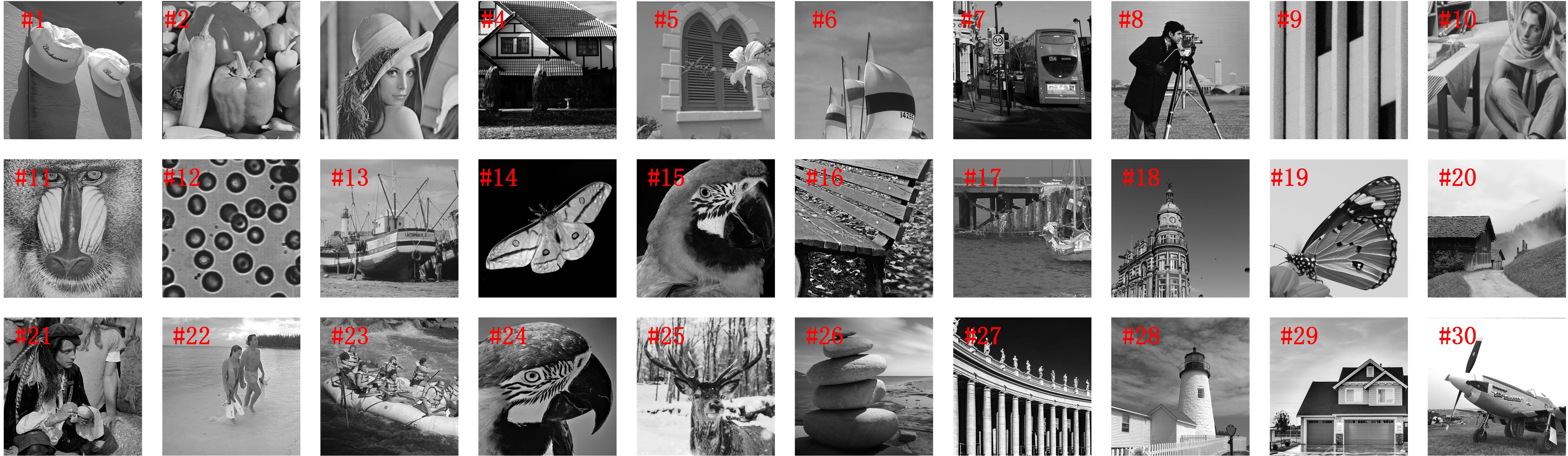}
  \caption{The test dataset includes 30 gray images, where the size of images $\#1-\#15$ is $512\times512$ and the size of images $\#16-\#30$ is $1024\times1024$.}  \label{30images}
  \vspace{-0.3cm}
\end{figure*}

\begin{figure*}[t]
\centering
\subfloat[\scriptsize{CFMC(27.45dB)}]{\includegraphics[width=0.135\textwidth]{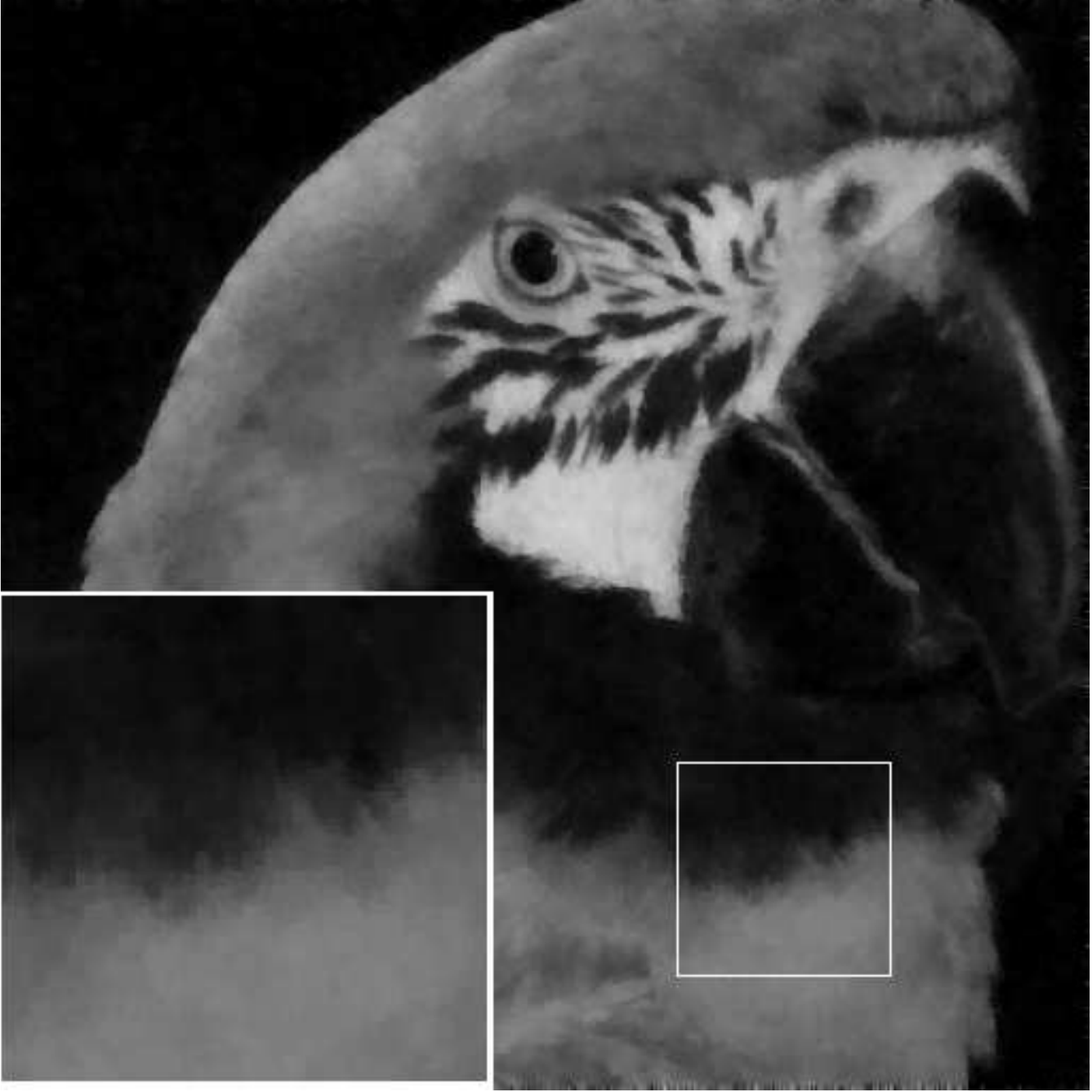}}
\hspace{-0.2mm}
\subfloat[\scriptsize{TVTV$^2$(28.49dB)}]{\includegraphics[width=0.135\textwidth]{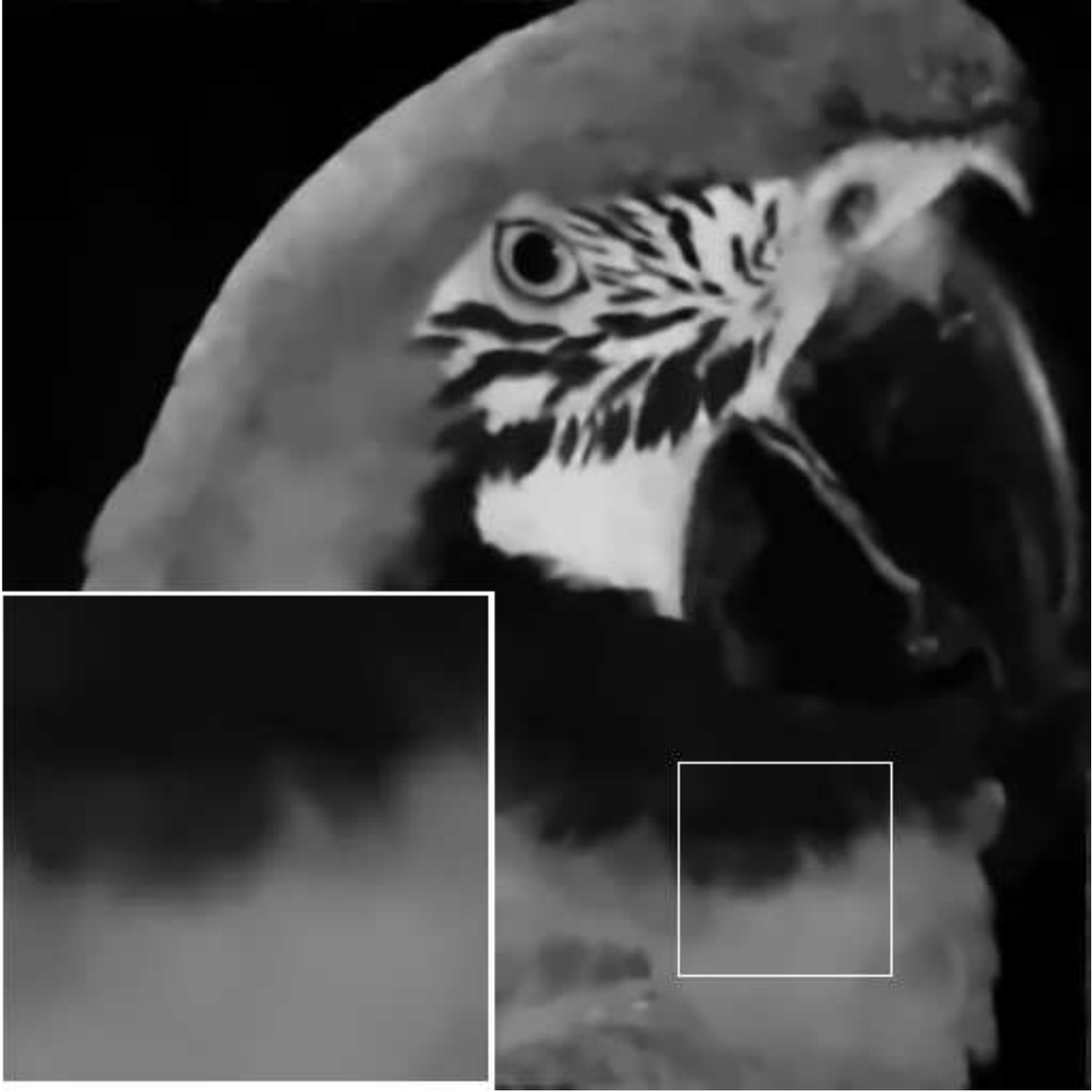}}
\hspace{-0.2mm}
\subfloat[\scriptsize{TGV(28.52dB)}]{\includegraphics[width=0.135\textwidth]{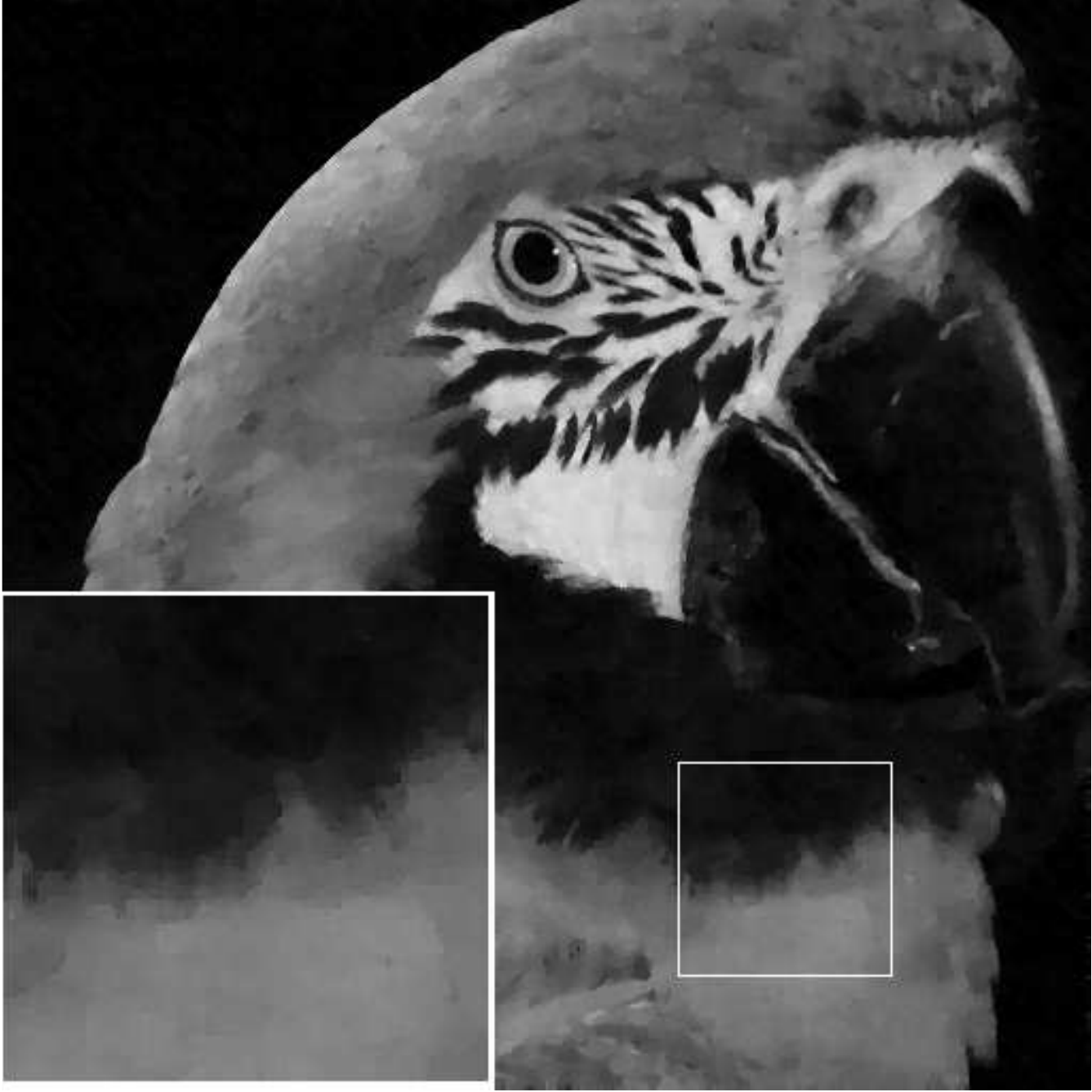}}
\hspace{-0.2mm}
\subfloat[\scriptsize{Euler(28.67dB)}]{\includegraphics[width=0.135\textwidth]{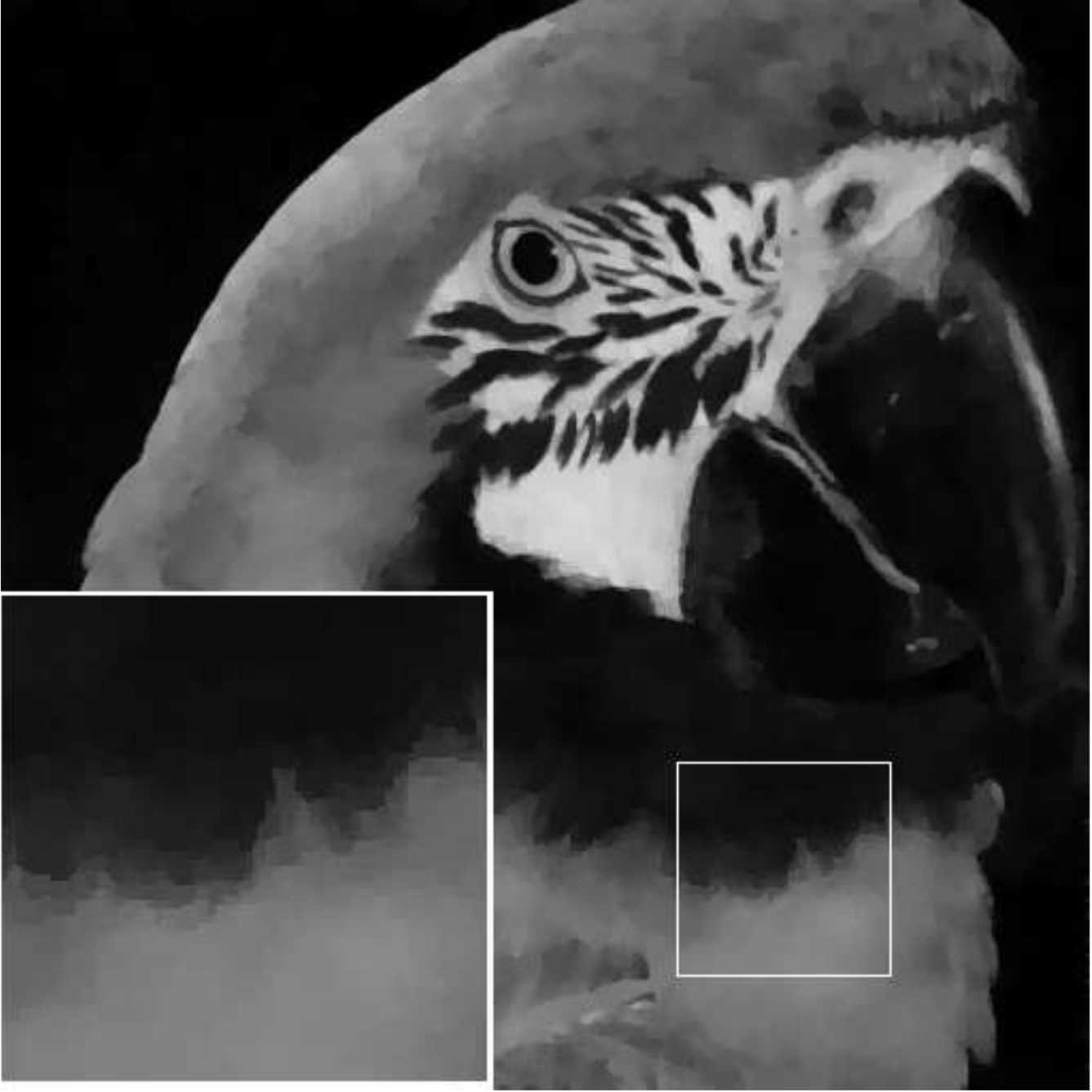}}
\hspace{-0.2mm}
\subfloat[\scriptsize{MC(28.81dB)}]{\includegraphics[width=0.135\textwidth]{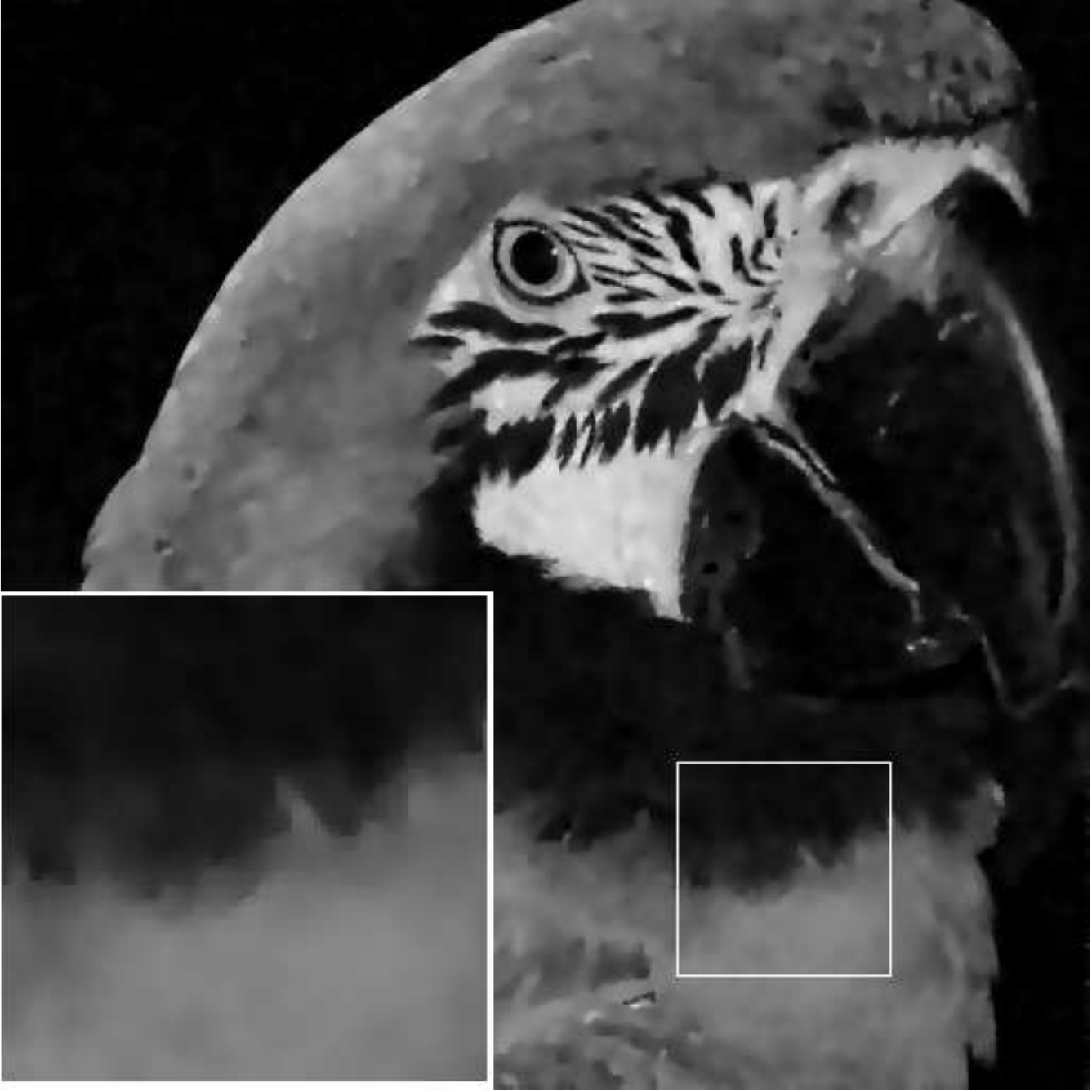}}
\hspace{-0.2mm}
\subfloat[\scriptsize{TAC(28.88dB)}]{\includegraphics[width=0.135\textwidth]{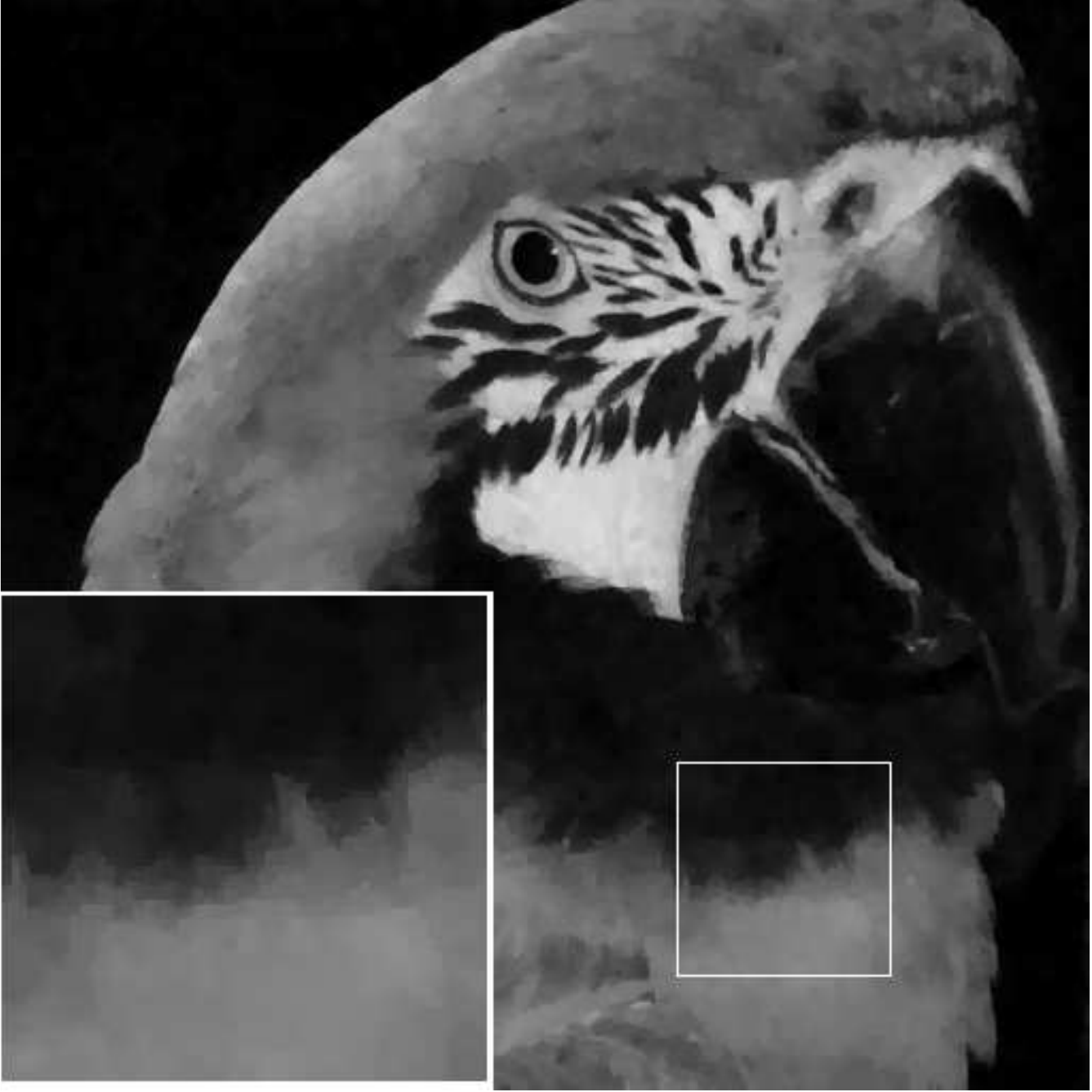}}
\hspace{-0.2mm}
\subfloat[\scriptsize{MGMC(29.08dB)}]{\includegraphics[width=0.135\textwidth]{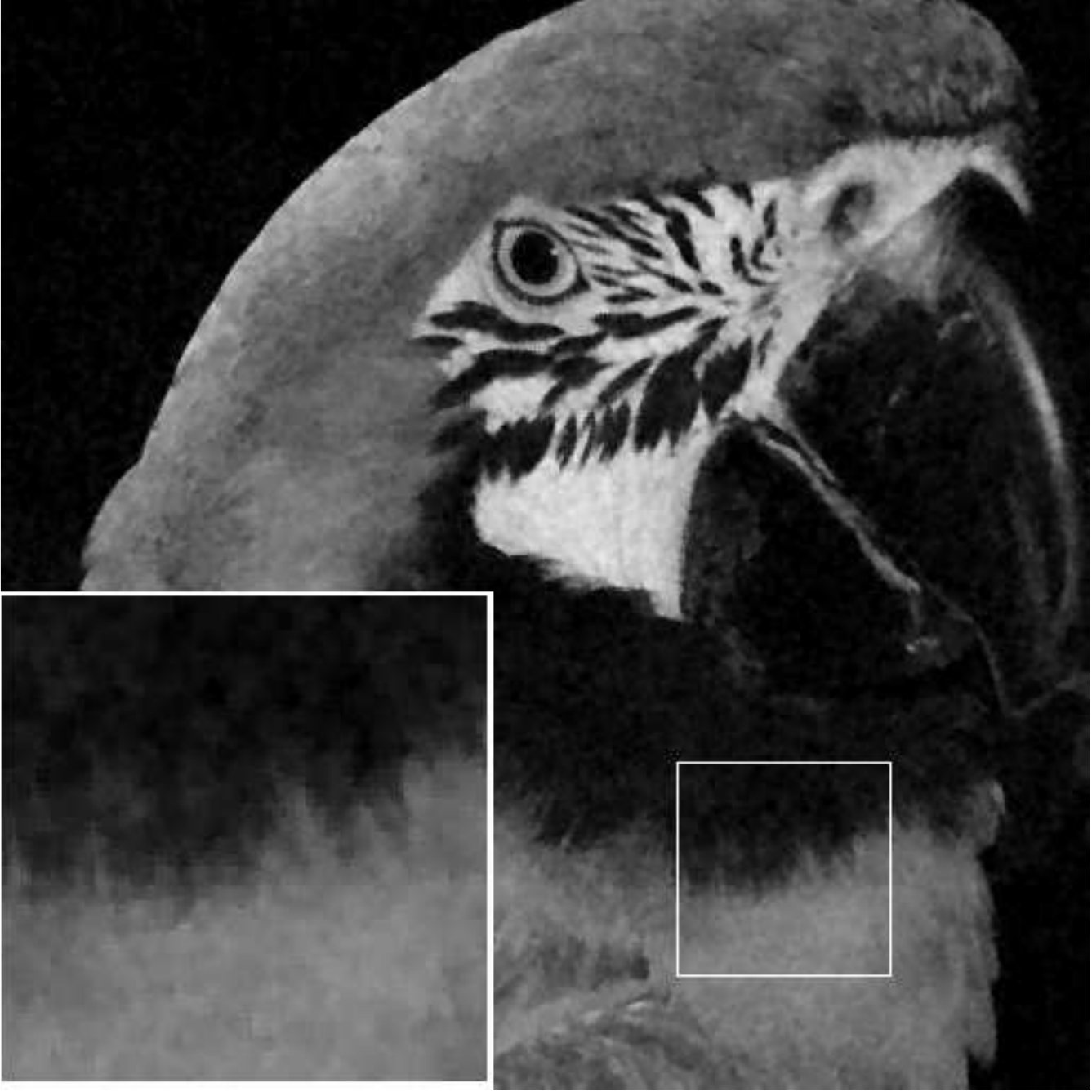}}\\
\vspace{-1.0mm}
\subfloat[\scriptsize{CFMC(27.88dB)}]{\includegraphics[width=0.135\textwidth]{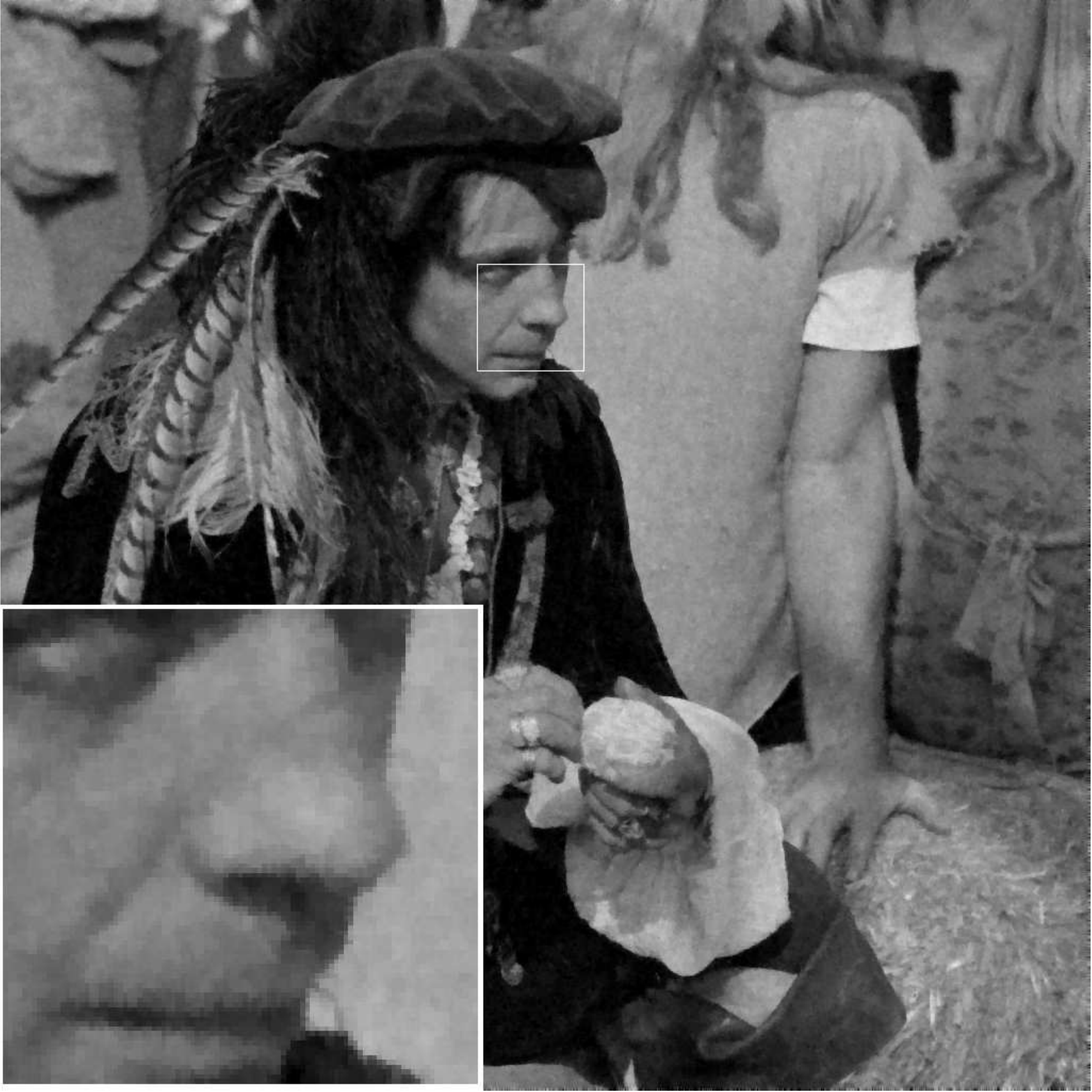}}
\hspace{-0.2mm}
\subfloat[\scriptsize{TVTV$^2$(29.31dB)}]{\includegraphics[width=0.135\textwidth]{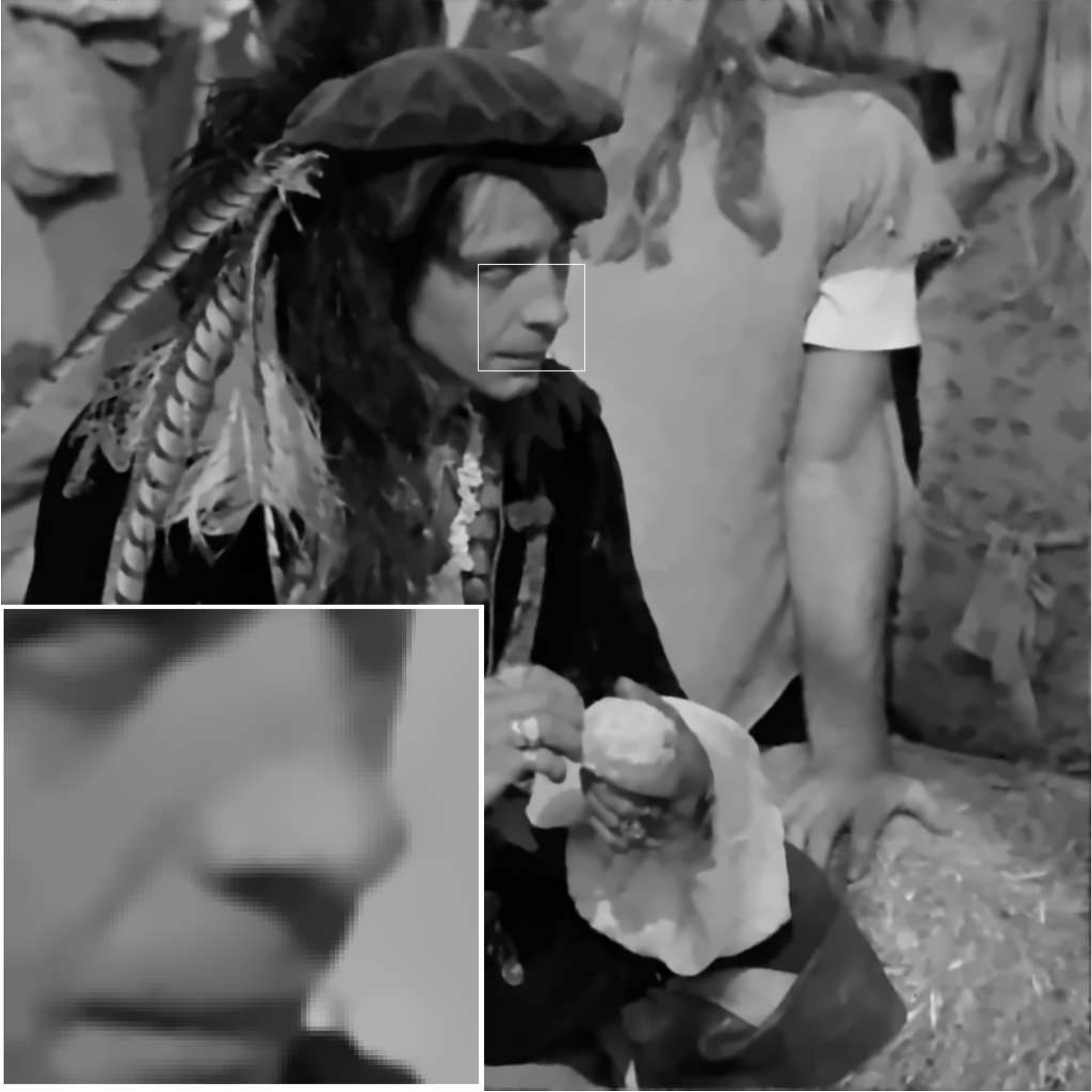}}
\hspace{-0.2mm}
\subfloat[\scriptsize{TGV(29.50dB)}]{\includegraphics[width=0.135\textwidth]{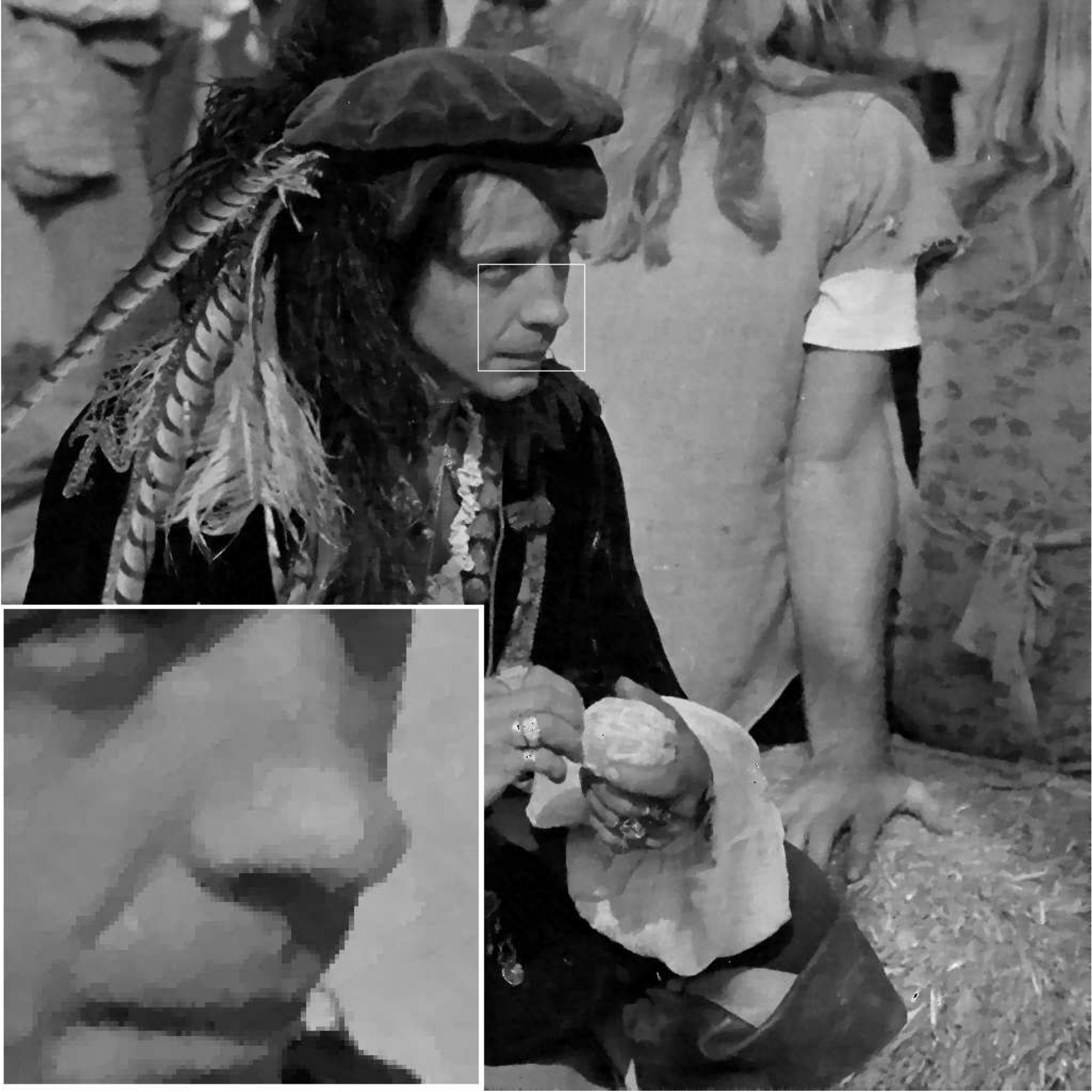}}
\hspace{-0.2mm}
\subfloat[\scriptsize{Euler(28.81dB)}]{\includegraphics[width=0.135\textwidth]{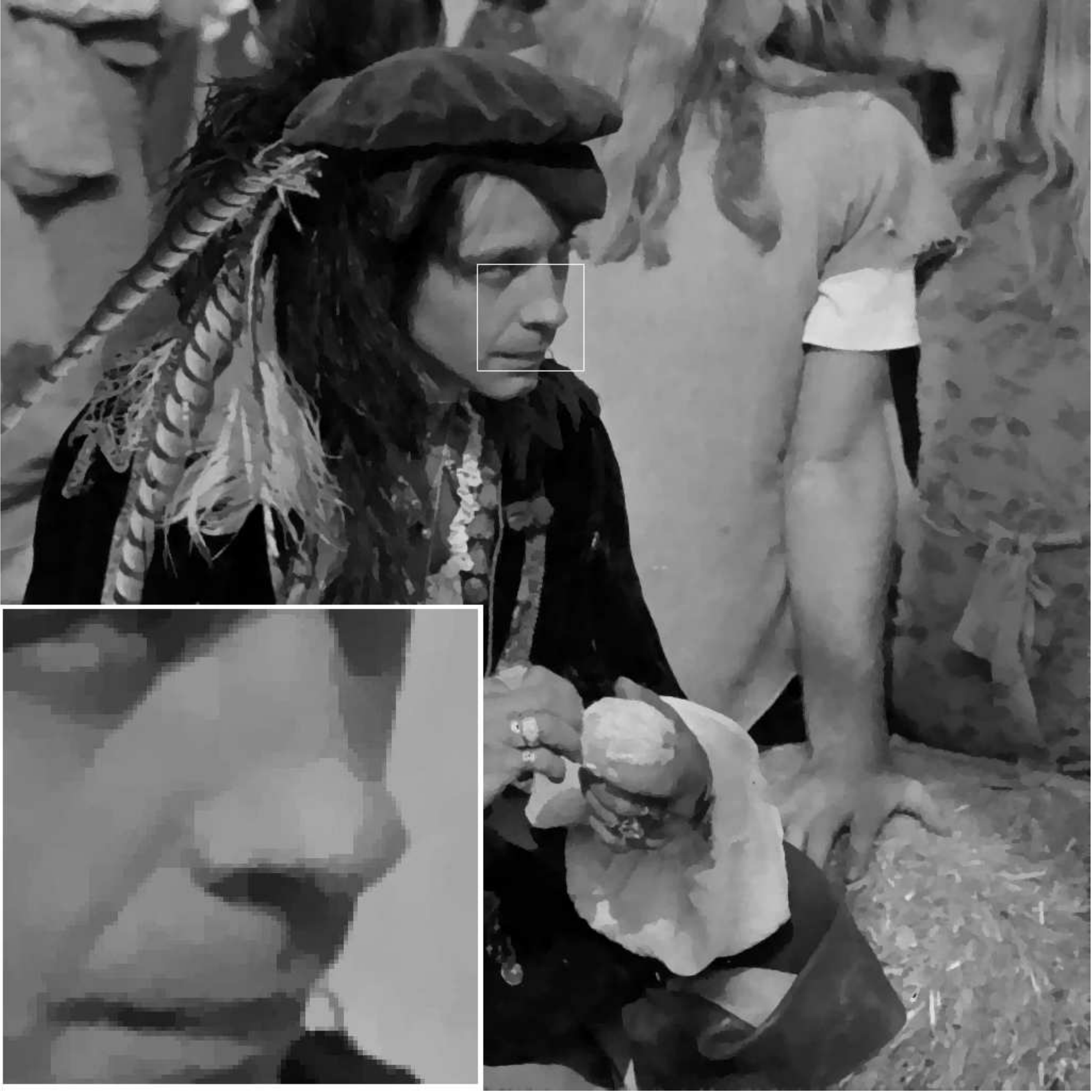}}
\hspace{-0.2mm}
\subfloat[\scriptsize{MC(29.54dB)}]{\includegraphics[width=0.135\textwidth]{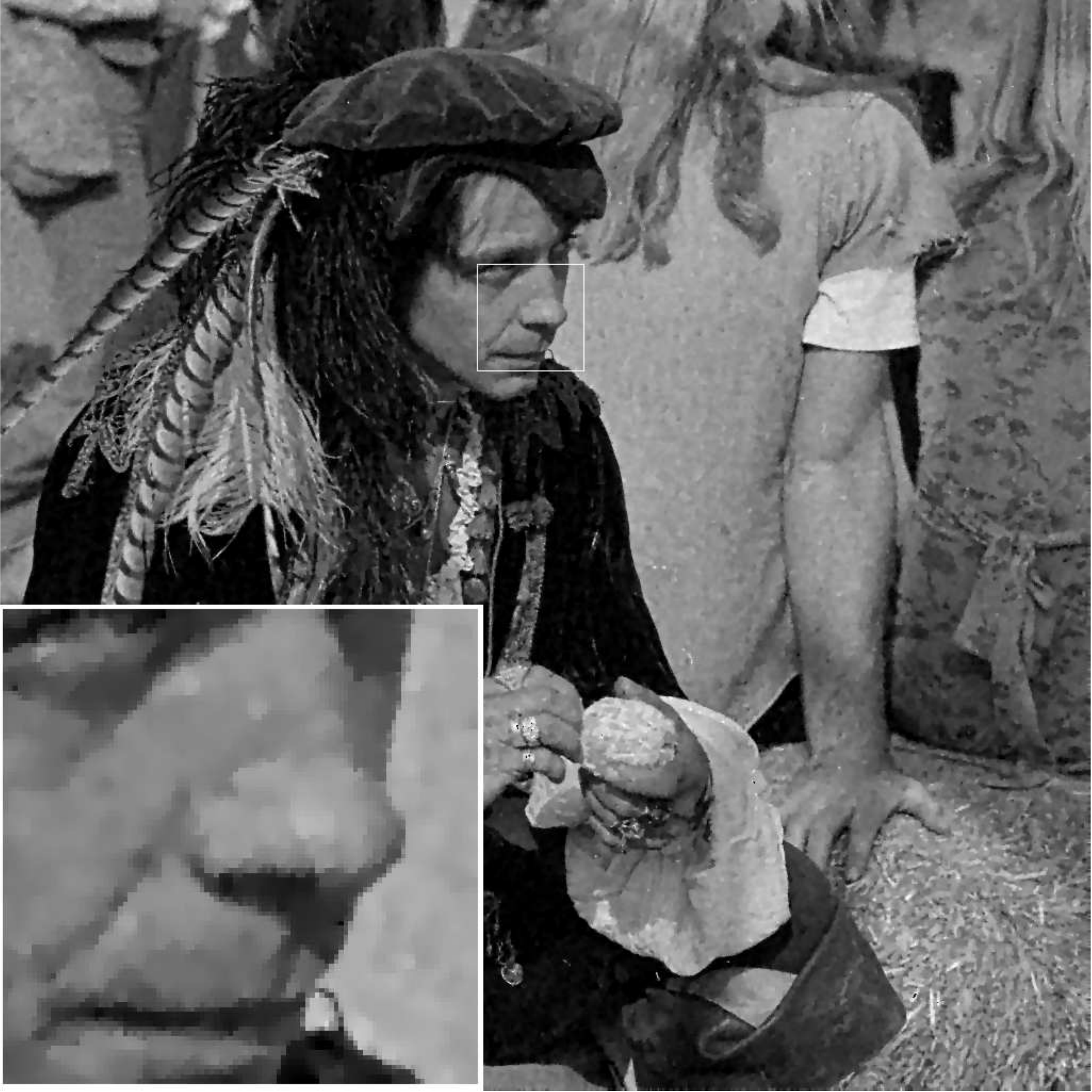}}
\hspace{-0.2mm}
\subfloat[\scriptsize{TAC(29.53dB)}]{\includegraphics[width=0.135\textwidth]{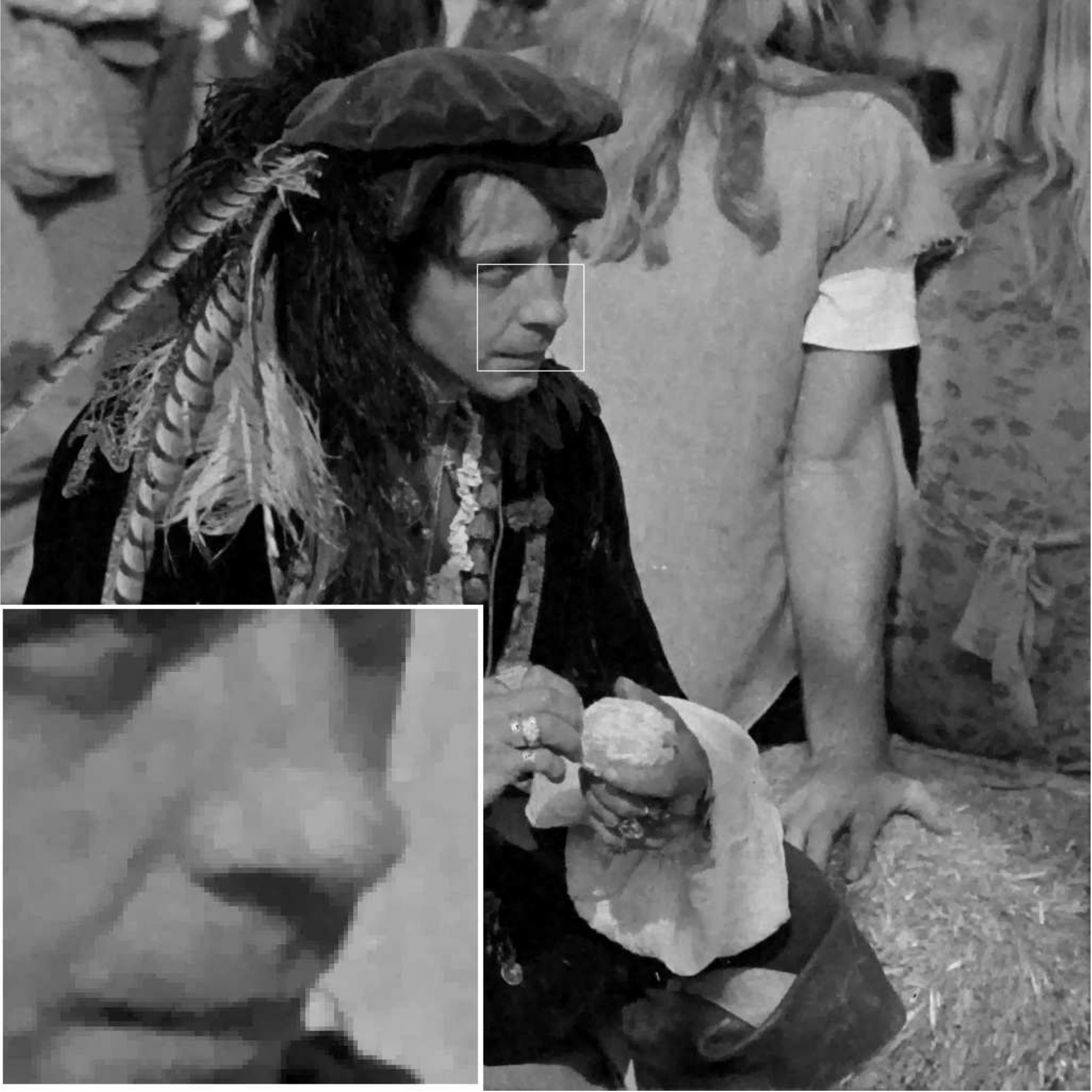}}
\hspace{-0.2mm}
\subfloat[\scriptsize{MGMC(29.62dB)}]{\includegraphics[width=0.135\textwidth]{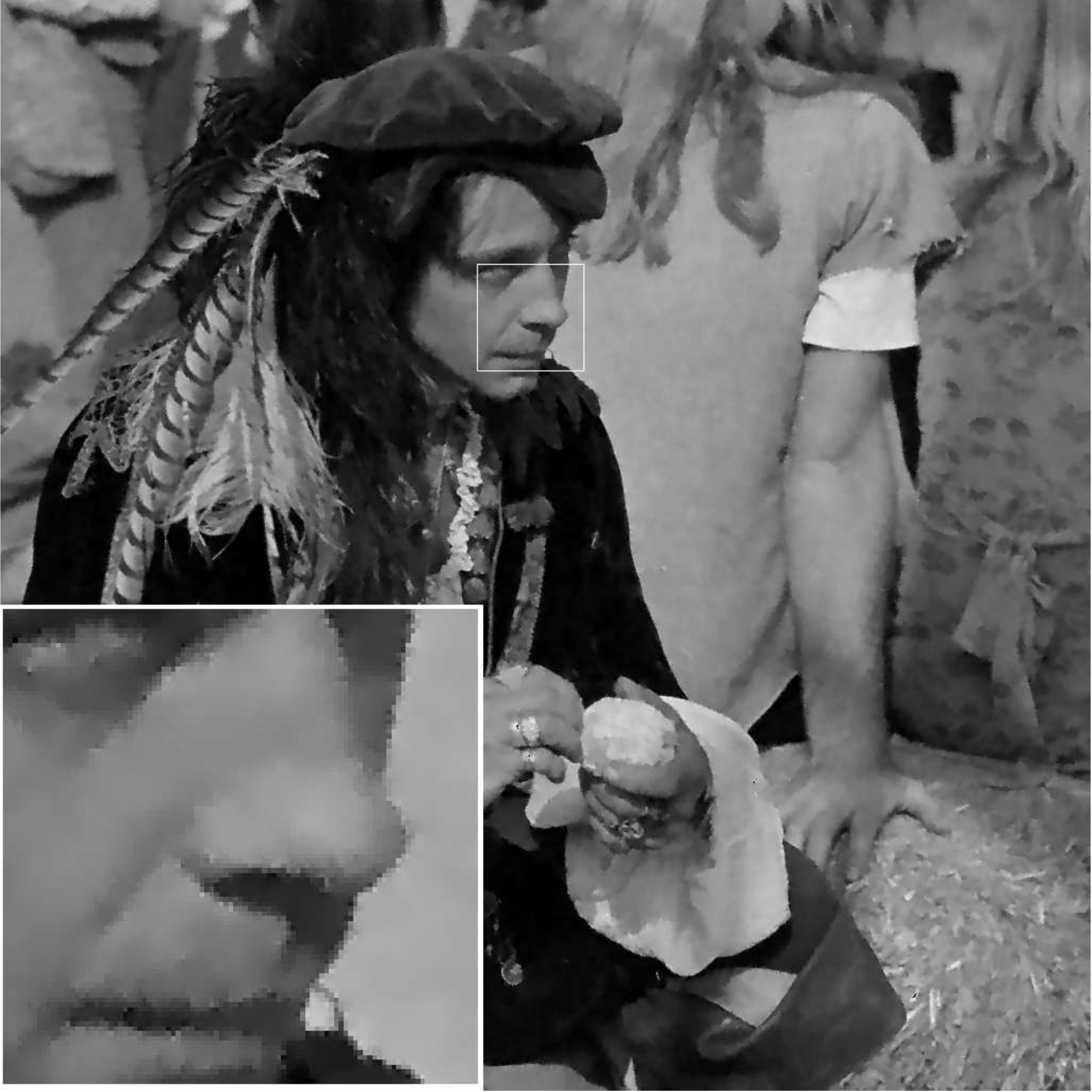}}
\hspace{-0.2mm}
\caption{The denoised results were obtained by different methods on image $\#15$ and $\#21$ with the noise level of $\sigma=20$.}\label{30 denoise result}
\end{figure*}

\begin{table*}[htbp]
  \centering
  \caption{The comparison of SSIM and PSNR (dB) among the  CFMC, Euler's elastica, mean curvature (MC), total absolute mean curvature (TAC) and our MGMC on 30 test images. }   \label{30 image restlt}
   \begin{tabular}{p{16pt}|p{20pt}p{20pt}p{20pt}p{20pt}p{20pt}p{20pt}p{22pt}|p{20pt}p{20pt}p{20pt}p{20pt}p{20pt}p{20pt}p{22pt}}
   \hline
   \hline
          & \multicolumn{7}{c|}{SSIM}                      & \multicolumn{7}{c}{PSNR} \\
          \hline
        & CFMC &TVTV$^2$&TGV    & Euler & MC    & TAC   & MGMC    & CFMC  &TVTV$^2$&TGV  & Euler & MC    & TAC   & MGMC \\
          \hline
    \#1 & 0.8362 & 0.8527  & 0.8543     & 0.8548  & 0.8618  & 0.8652  & \textbf{0.8678}  & 30.94 & 32.81  & 32.94   & 32.85  & 33.05  & 33.12  & \textbf{33.23} \\
    \#2  & 0.8348  & 0.8918  & 0.8976    & 0.8928  & 0.9115  & 0.9066  & \textbf{0.9124}   & 30.00& 31.73  & 31.95     & 31.77  & 31.85  & 31.97  & \textbf{32.32} \\
    \#3  & 0.8758 & 0.9289  & 0.9297     & 0.9355  & 0.9391  & 0.9311  & \textbf{0.9393}  & 28.81   & 30.55  & 30.76   & 30.33  & 30.70  & 30.55  & \textbf{30.80} \\
    \#4 & 0.8225 & 0.8622  & 0.8590   & 0.8592  & 0.8683  & 0.8656  & \textbf{0.8692}   & 24.01& 26.55  & 26.16   & 25.98  & 26.44  & 26.16  & \textbf{26.75} \\
    \#5& 0.7652& 0.9049  & 0.9015      & 0.9011  & 0.9142  & 0.9118  & \textbf{0.9169}  & 31.35& 32.55  & 32.46     & 32.32  & 32.61  & 32.54  & \textbf{32.81} \\
    \#6  & 0.8142 & 0.8279  & 0.8384    & 0.8315  & 0.8323  & 0.8311  & \textbf{0.8410}   & 30.27& 31.53  & 31.33     & 31.32  & 31.42  & 31.51  & \textbf{31.66} \\
    \#7 & 0.7634 & 0.8482  & 0.8467     & 0.8483  & 0.8491  & 0.8489  & \textbf{0.8496}& 25.71  & 27.99  & 28.11     & 27.93  & 28.22  & 28.09  & \textbf{28.24} \\
    \#8  & 0.8319& 0.8649  & 0.8671    & 0.8665  & 0.8787  & 0.8793  & \textbf{0.8798} & 28.54& 31.41  & 31.11      & 31.13  & 31.37  & 31.29  & \textbf{31.42} \\
    \#9  & 0.7365& 0.7765  & 0.7715     & 0.7747  & 0.7945  & \textbf{0.7981}  & \textbf{0.7981} & 27.78& 28.33  & 28.13     & 28.28  & 28.48  & \textbf{28.63}  & \textbf{28.63} \\
    \#10  & 0.8008 & 0.9021  & 0.9038    & 0.9042  & 0.9095  & 0.9007  & \textbf{0.9098} & 29.37& 30.98  & 31.13      & 30.98  & 31.38  & 31.37  & \textbf{31.71} \\
    \#11 & 0.7369 & 0.7505  & 0.7511    & 0.7501  & 0.7521  & 0.7514  & \textbf{0.7572}   & 23.30& 24.59  & 24.52    & 24.73  & 24.68  & 24.72  & \textbf{24.75} \\
    \#12 & 0.8381& 0.8511   & 0.8603    & 0.8655  & 0.8764  & 0.8741  & \textbf{0.8791}  & 30.04 & 32.51 & 32.38    & 32.32  & 32.74  & 32.73  & \textbf{32.97} \\
    \#13   & 0.7536 & 0.8189  & 0.8119    & 0.8151  & 0.8194  & 0.8233  & \textbf{0.8235}   & 27.18 & 28.55  & 29.07     & 28.73  & 29.08  & 28.84  & \textbf{29.16} \\
    \#14  & 0.7269 & 0.8761  & 0.8791    & 0.8824  & 0.8855  & 0.8838  & \textbf{0.8875} & 29.54& 32.19  & 32.22      & 32.12  & 32.37  & 32.24  & \textbf{32.36} \\
    \#15 & 0.8206  & 0.8221  & 0.8214   & 0.8207  & 0.8311  & 0.8321  & \textbf{0.8324}   & 27.45  & 28.49  & 28.52  & 28.67  & 28.81  & 28.88  & \textbf{29.08} \\
    \hline
    AVG    & 0.7599  &0.8588  & 0.8506    & 0.8502  & 0.8582  & 0.8569  & \textbf{0.8609}   & 28.29& 30.05 & 30.12      & 29.96  & 30.21  & 30.18  & \textbf{30.39} \\
    \hline
    \#16  & 0.7695& 0.8524  & 0.8556    & 0.8539  & 0.8624  & 0.8622  & \textbf{0.8642}  & 26.17 & 28.78  & 29.07     & 28.79  & 28.90  & 28.95  & \textbf{29.22} \\
    \#17 & 0.8138 & 0.8255  & 0.8282     & 0.8307  & 0.8315  & 0.8308  & \textbf{0.8401} & 29.82 & 31.11  & 30.98      & 31.05  & 31.22  & 31.14  & \textbf{31.28} \\
    \#18& 0.8107 & 0.8399  & 0.8403      & 0.8408  & 0.8413  & 0.8352  & \textbf{0.8421}   & 29.48& 33.04  & 32.96    & 33.01  & 33.35  & 33.28  & \textbf{33.38} \\
    \#19  & 0.7969 & 0.8984  & 0.8895    & 0.8942  & 0.9087  & 0.8931  & \textbf{0.9091}& 30.04& 32.53  & 32.88       & 32.88  & 33.28  & 33.21  & \textbf{33.30} \\
    \#20 & 0.7513  & 0.8266  & 0.8236    & 0.8329  & 0.8362  & 0.8389  & \textbf{0.8392} & 28.16 & 29.31  & 29.50    & 29.39  & 29.55  & 29.59  & \textbf{29.63} \\
    \#21 & 0.7070& 0.8228  & 0.8269      & 0.8268  & 0.8312  & 0.8280  & \textbf{0.8325} & 27.88 & 28.53  & 29.52     & 28.81  & 29.54  & 29.53  & \textbf{29.62} \\
     \#22 & 0.8073& 0.8794  & 0.8801     & 0.8794  & 0.8807  & 0.8803  & \textbf{0.8834} & 26.92& 27.55  & 27.42      & 27.45  & 27.66  & 27.73  & \textbf{27.88} \\
     \#23   & 0.8025 & 0.8799  & 0.8792    & 0.8835  & 0.8841  & 0.8852  & \textbf{0.8859} & 27.55& 29.68  & 29.55        & 29.60  & 29.67  & 29.67  & \textbf{29.74} \\
     \#24  & 0.8577  & 0.9136  & 0.9246    & 0.9248  & 0.9402  & 0.9439  & \textbf{0.9476}  & 28.61& 30.09  & 30.17    & 30.21  & 30.59  & 30.52  & \textbf{30.68} \\
     \#25  & 0.8943 & 0.8859  & 0.8873    & 0.8861  & 0.8976  & 0.8965  & \textbf{0.8985} & 28.75  & 30.24  & 30.51    & 30.04  & 30.42  & 30.65  & \textbf{30.81} \\
     \#26& 0.8523  & 0.8846  & 0.8873    & 0.8959  & 0.8949  & 0.8966  & \textbf{0.8968}   & 32.88 & 33.49  & 33.91     & 34.01  & 34.28  & 34.21  & \textbf{34.33} \\
     \#27  & 0.7926& 0.8349  & 0.8392      & 0.8451  & 0.8468  & 0.8489  & \textbf{0.8491}  & 25.62  & 27.29  & 27.31  & 27.23  & 27.43  & 27.46  & \textbf{27.49} \\
     \#28 & 0.8195& 0.8311  & 0.8305     & 0.8397  & 0.8401  & 0.8387  & \textbf{0.8402}& 25.51 & 26.69  & 26.71      & 26.71  & 26.83  & 26.86  & \textbf{26.88} \\
     \#29 & 0.8264  & 0.8828  & 0.8847    & 0.8796  & 0.8871  & 0.8869  & \textbf{0.8876} & 30.94 & 32.13  & 32.11    & 32.88  & \textbf{32.99}  & 32.86  & \textbf{32.99} \\
    \#30   & 0.8196 & 0.8691  & 0.8677   & 0.8659  & 0.8741  & 0.8738  & \textbf{0.8754}  & 31.20 & 32.68  & 33.45    & 33.69  & 33.82  & 33.71  & \textbf{33.89} \\
    \hline
    AVG   & 0.7588  & 0.8385  & 0.8396   & 0.8420  & 0.8471  & 0.8459  & \textbf{0.8494}  & 28.50 & 30.08  & 30.27   & 30.38  & 30.64  & 30.62  & \textbf{30.74} \\
    \hline
    \hline
    \end{tabular}
\end{table*}%

Table \ref{30 image restlt} records both PSNR and SSIM obtained by different methods, where
our method gives the best restoration qualities.
We also display two representative restoration results, i.e., image $\#15$ and $\#21$, in Fig. \ref{30 denoise result}.
By observing the results of image $\#15$, our proposed method can preserve very good textural structures, while other methods have smoothed out the fine details. The reason behind this is that we use more neighboring points to estimate the update for the central point; see the examples in Fig. \ref{ROF MC point}. Next, we compare the numerical energy and relative error of CFMC\cite{{Gong2019mean}}, MC \cite{Zhu2012image}, one-layer MGMC and MGMC on both image $\#15$ and $\#21$ in Fig. \ref{man energy}, which are corrupted by Gaussian noise with the noise level $\sigma=20$.
As shown, with the same regularization parameter, our MGMC provides lower numerical energy and faster convergence than the ALM-based algorithm in \cite{2017Augmented}.

\begin{figure}[!htbp]
\centering
\includegraphics[width=0.5\textwidth]{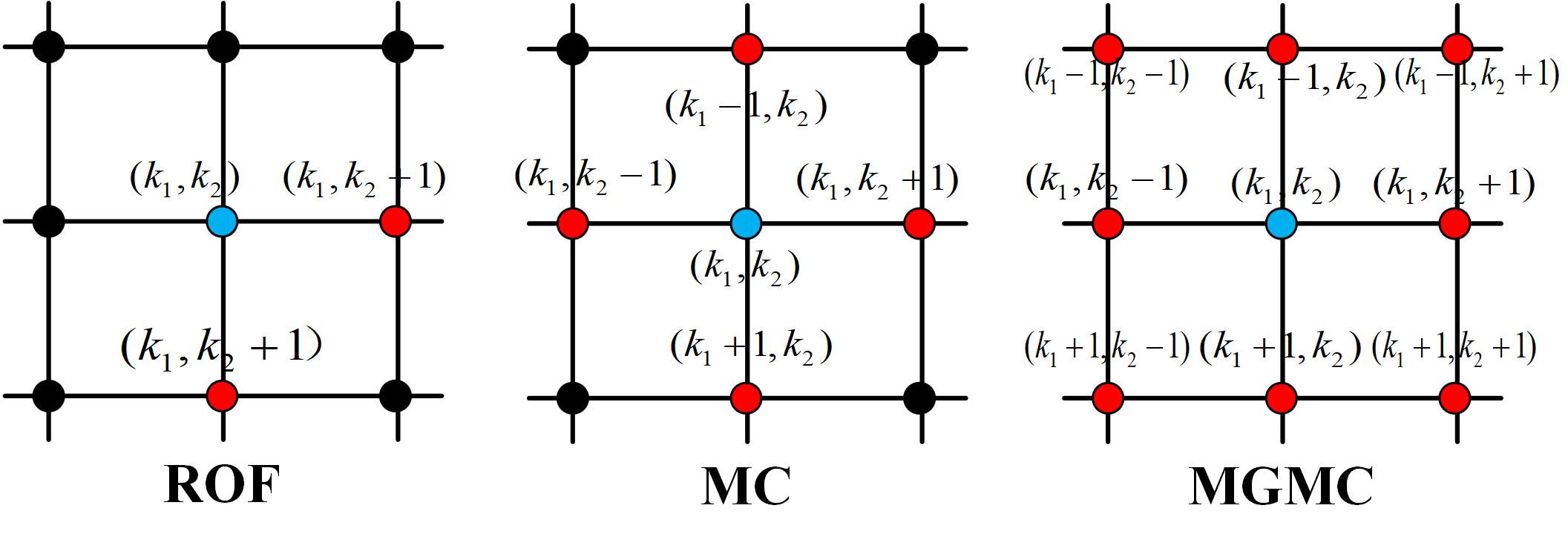}
\caption{The illustration of points used to estimate the update for the central point (blue one), where the red points are used for calculation.
}\label{ROF MC point}
\end{figure}

\begin{figure*}[!htbp]
\centering
\subfloat[\scriptsize{Energy decay of image $\#15$}]{\includegraphics[width=0.21\textwidth]{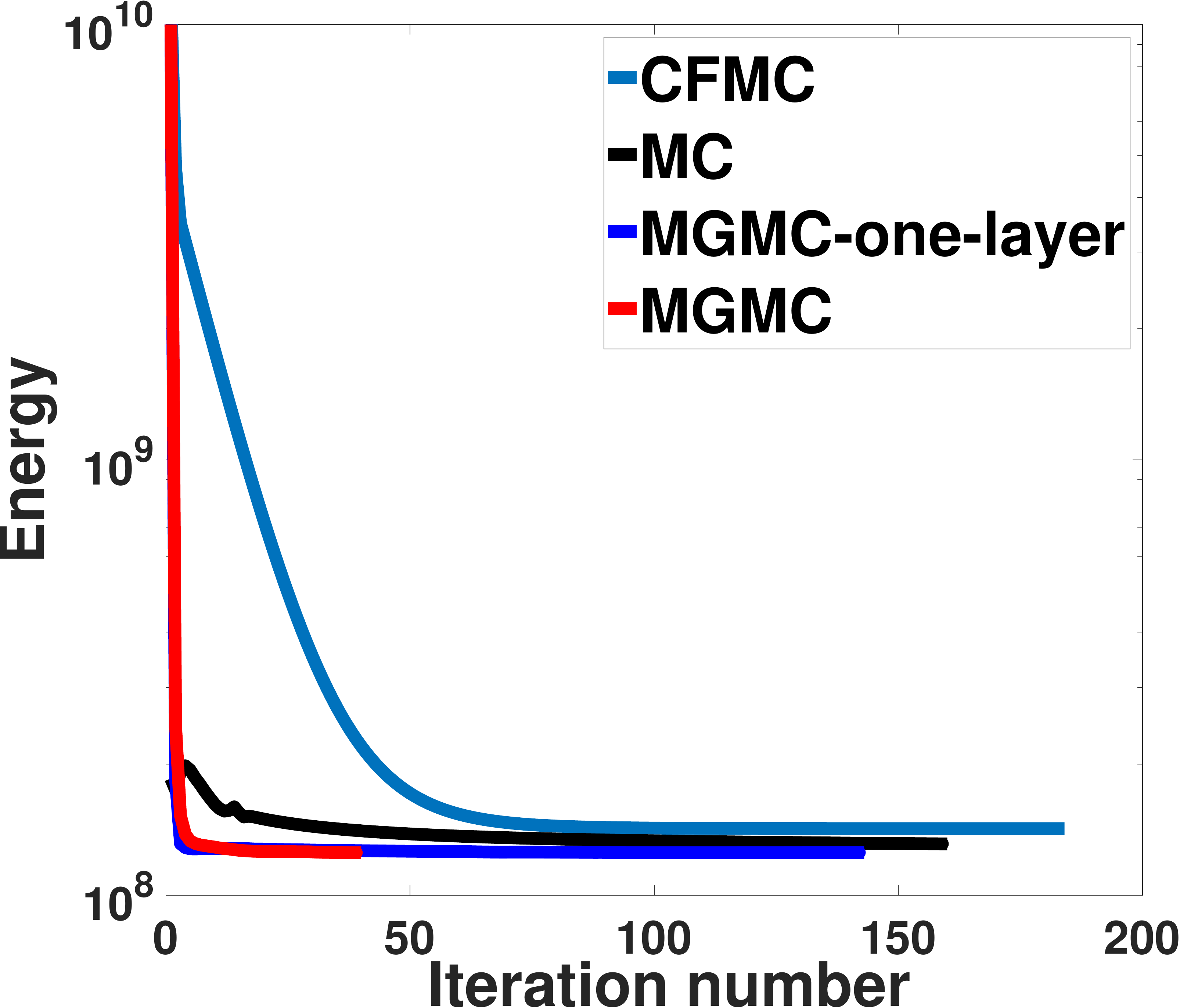}}
\subfloat[\scriptsize{Relative error of image $\#15$} ]{\includegraphics[width=0.21\textwidth]{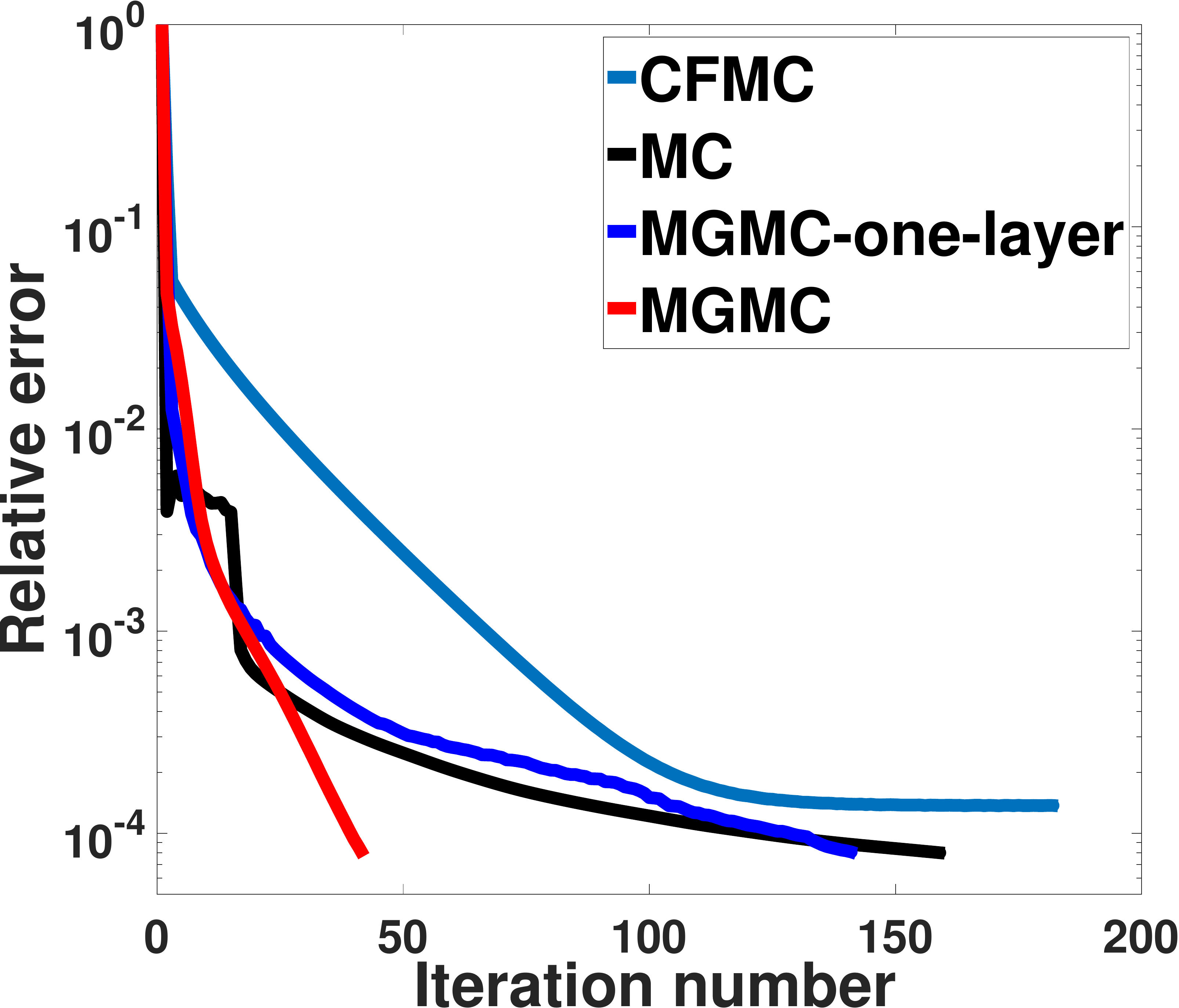}}
\subfloat[\scriptsize{Energy decay of image $\#21$}]{\includegraphics[width=0.21\textwidth]{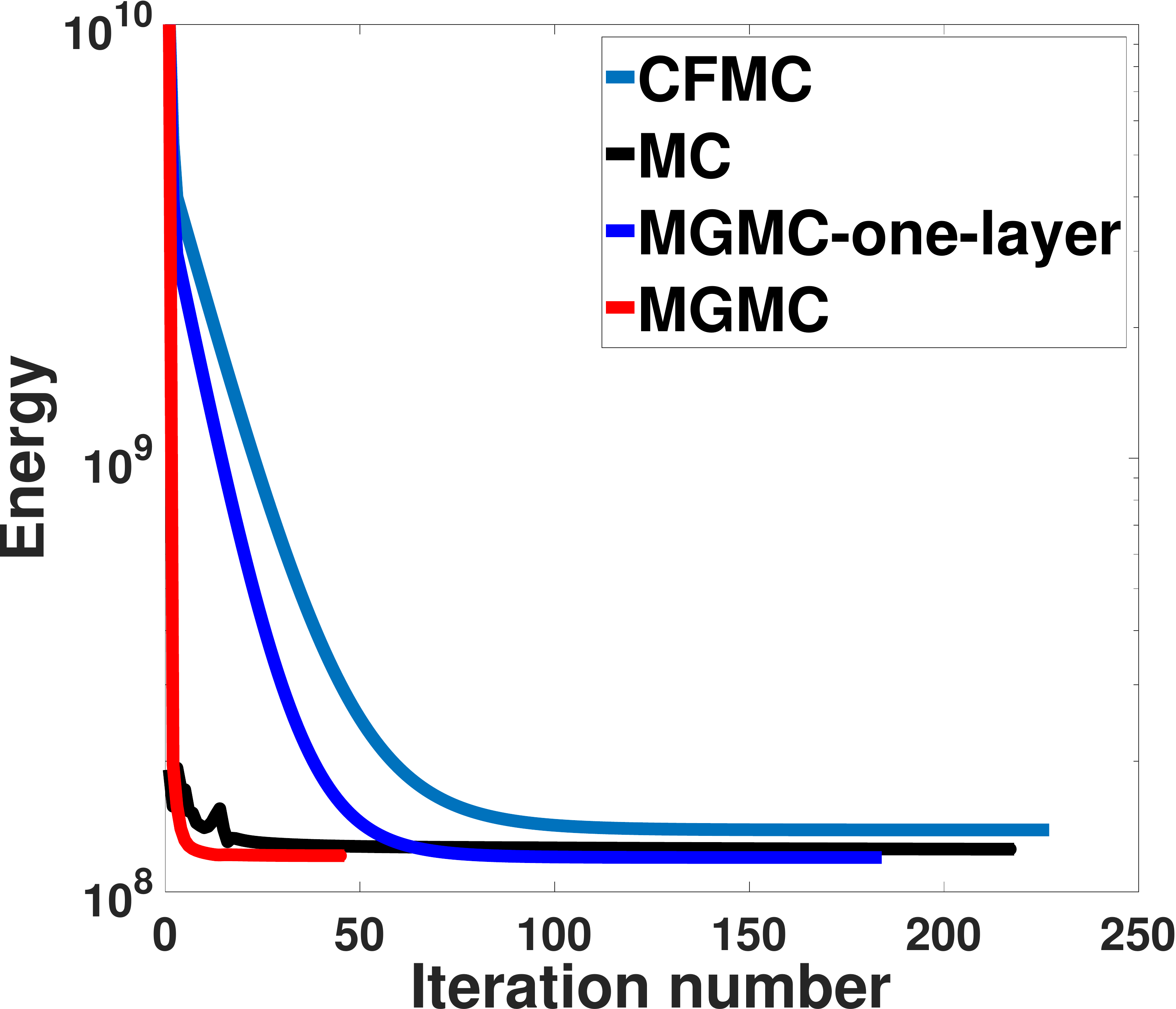}}
\subfloat[\scriptsize{Relative error of image  $\#21$}]{\includegraphics[width=0.21\textwidth]{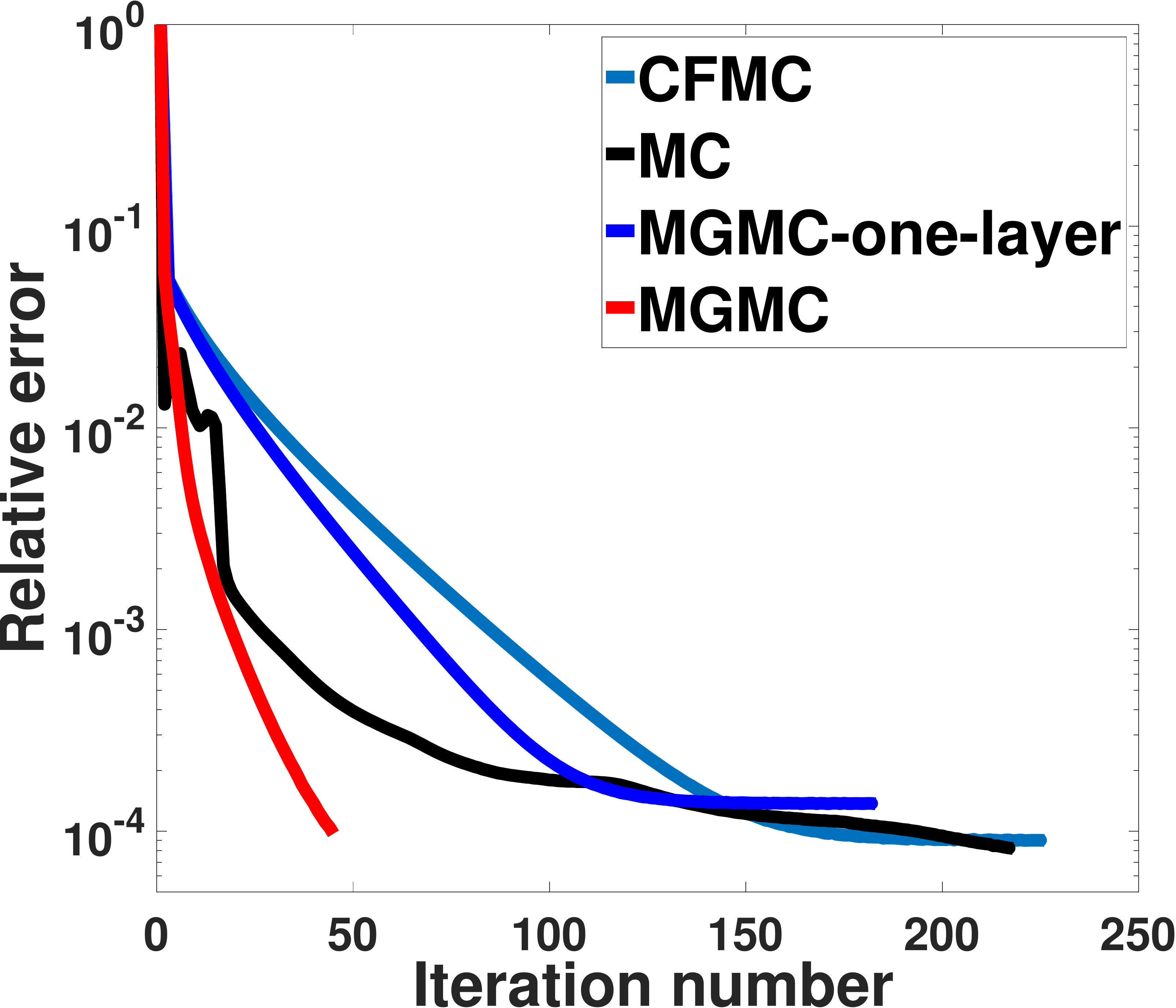}}
\caption{The decays of the numerical energy and relative error for image $\#15$ and $\#21$ corrupted by Gaussian noise of level $\sigma=20$.}\label{man energy}
\end{figure*}

\begin{figure}[!htbp]
\centering
\subfloat{\includegraphics[width=0.24\textwidth]{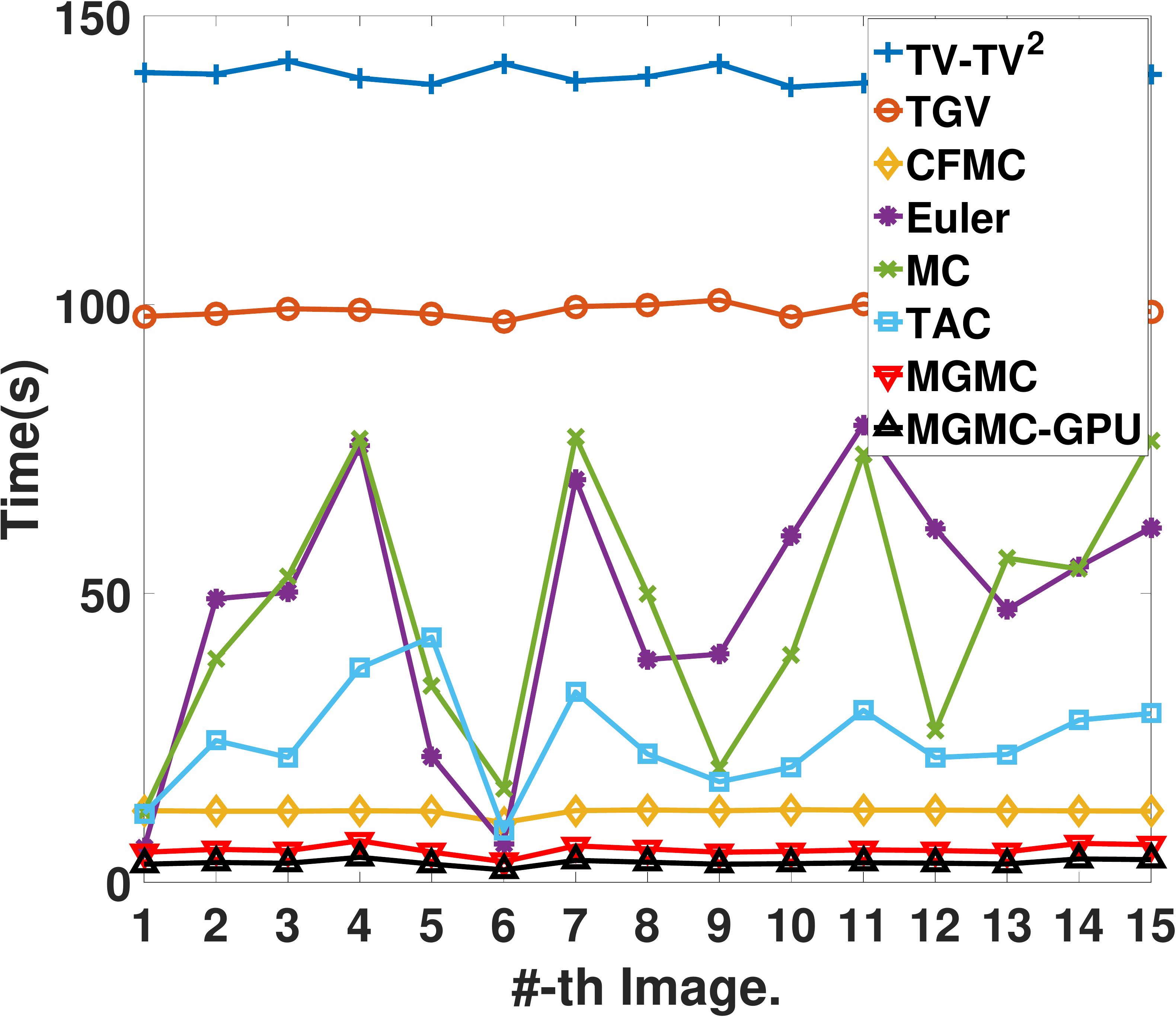}}~
\subfloat{\includegraphics[width=0.24\textwidth]{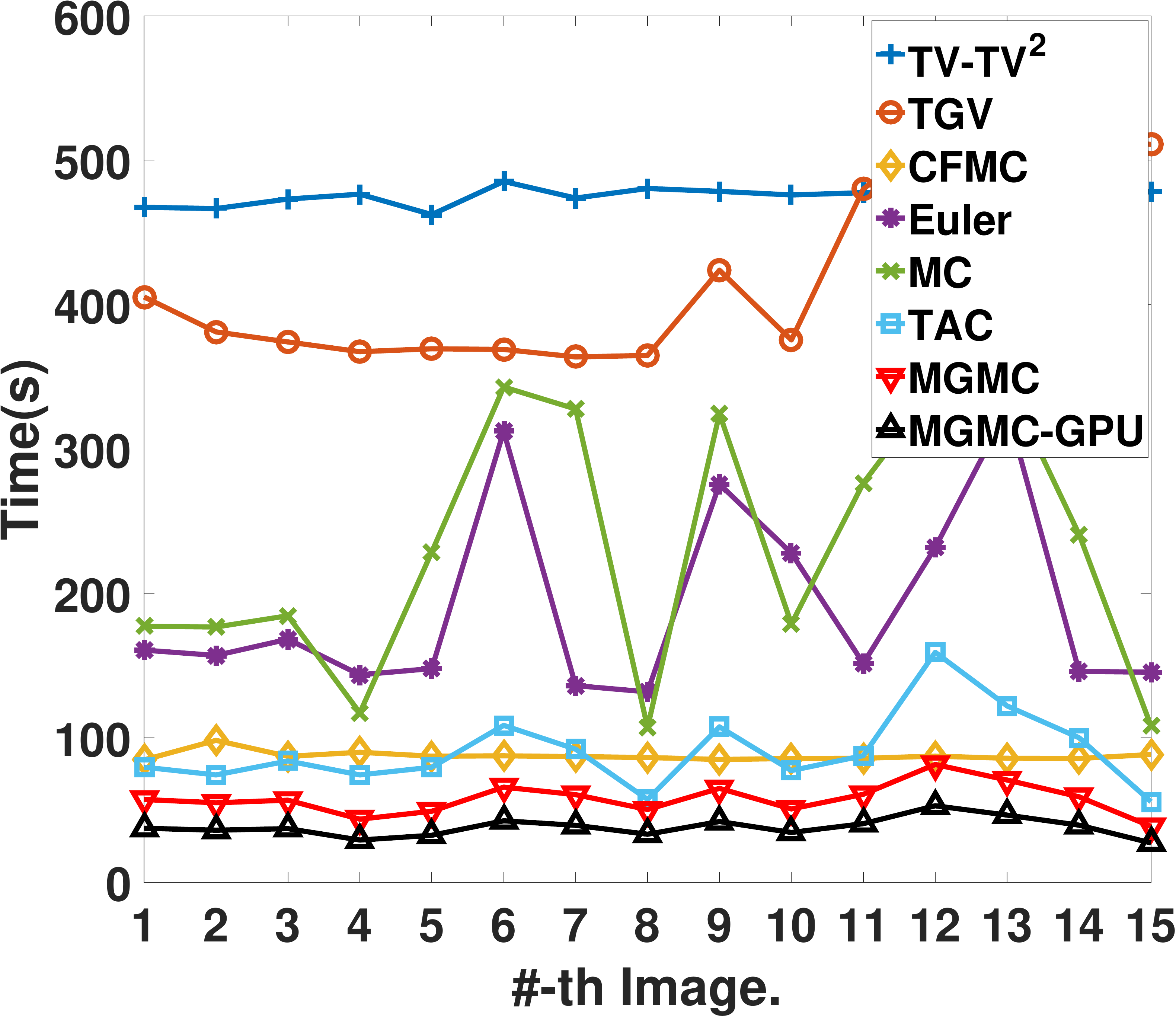}}
\caption{The computational time comparison among different methods on the 30 test images.}\label{30psnr_time}
\end{figure}

Since the sub-minimization problems of the same color in Algorithm \ref{MC_algor} are independent, we can use the GPU computation to improve its efficiency.
The GPU implementation was carried out on a computer with an Intel(R) Core i9 CPU at 3.30 GHz and Nvidia GeForce GTX 1050TI GPU card.
Fig. \ref{30psnr_time} displays the running time of different methods, where the left plot contains images numbered from $\#1$ to $\#15$ with the size of $512\times 512$ and the right plot contains images numbered from $\#16$ to $\#30$ with the size of $1024\times 1024$.
It can be observed, among the curvature minimization approaches, our MGMC method is the fastest one followed by CFMC, TAC, Euler's elastica, MC, TGV, and TVTV$^2$ model. Furthermore, both MC and Euler's elastica model spend similar computational time, more than TAC, which is in accord with our complexity analysis. Moreover, GPU implementation also accelerates efficiency. For images of size 1024$\times$1024, the computational time is improved from the 40s to 25s, which is very important for real applications. In summary, our multi-grid algorithm achieves the best performance on image restoration problems with very high efficiency, which does not trade accuracy for speed.

\section{Application to image reconstruction}\label{section4}

In this section, we extend the curvature regularization method and multi-grid algorithm to more general inverse problems.
The task is to recover $u\in \mathbb R^2$ from the observed data defined by
\begin{equation}\label{degradation}
b=Au +\nu,
\end{equation}
where $\nu$ is the random noise and $A$ is a linear and bounded operator varying with different image processing tasks. To be specific, $A$ represents the Radon transform and Fourier transform for CT and MRI reconstruction, respectively.

We use the mean curvature as the regularization term and are concern with the following image reconstruction problem
\begin{equation*}
\min_{u}~~\frac12\big\|Au-b\big\|^2_2+\alpha\sum_{ x\in\rm\Omega}\big|H(u(x))\big|,
\end{equation*}
which is solved by the aforementioned multi-grid method.
Similarly, we implement the non-overlapping domain decomposition method on each layer to make the subproblems become independent and can be solved in parallel. The sub-minimization problems belonging to the same color are gathered as follows
\begin{equation*}
\small
\min_{c_j\in \mathbb R^{N_k}}\frac12\Big\|A(u+\sum_{i\in I_k}c_j^i\phi_j^i)-b\Big\|_2^2\!+\!\alpha\sum_{i\in I_k}\sum_{x\in \tau_j^i} \big|H(u(x)+c_j^i\phi_j^i(x))\big|.
\end{equation*}
According to Proposition \ref{Bernstein}, the above subproblem can be further reformulated into the following quadratic problem
\begin{equation*}
\min_{c_j\in \mathbb R^{N_k}}\frac12\Big\|A(u+\sum_{i\in I_k}c_j^i\phi_j^i)-b\Big\|_2^2+\alpha\sum_{i\in I_k} \Big(c_j^i-d_j^i\Big)^2,
\end{equation*}
where the closed-form solution is given as
\begin{scriptsize}
\begin{equation}\label{linear equation}
\left[
\begin{array}{lll}
\!L_{1,1}\!&\!\cdots\!&\!\langle A\phi_j^1,A\phi_j^{N_k}\rangle\!\\
\!\langle A\phi_j^2,A\phi_j^1\rangle\!\!&\!\cdots\!&\!\langle A\phi_j^2,A\phi_j^{N_k}\rangle\!\\
      \vdots                              &     \vdots          &      \vdots      \\
\!\langle A\phi_j^{N_k-1},A\phi_j^1\rangle\!&\!\cdots\!&\!\langle A\phi_j^{N_{k-1}},A\phi_j^{N_k}\rangle\!\\
\!\langle A\phi_j^{N_k},A\phi_j^1\rangle\!&\!\cdots\!&\!L_{N_k,N_k}\!\\
\end{array}
\right]
\left[
\begin{array}{l}
\!c_j^1\!\\
\!c_j^2\!\\
\!\vdots\!\\
\!c_j^{N_k-1}\!\\
\!c_j^{N_k}\!\\
\end{array}
\right]
=\left[
\begin{array}{l}
\!r_j^1\!\\
\!r_j^2\!\\
\!\vdots\!\\
\!r_j^{N_k-1}\!\\
\!r_j^{N_k}\!\\
\end{array}
\right],
\end{equation}
\end{scriptsize}with
\[
L_{i,i}=\langle A\phi_j^i,A\phi_j^i\rangle +2\alpha,~~ \mbox{and}~
r_j^i=\langle b-Au,A\phi_j^i\rangle +2\alpha d_j^i,
\]
for $i=1,\cdots ,N_k$.
For such a symmetric linear system, we can implement the conjugate gradient as the numerical solver.

\subsection{CT reconstruction}
The CT reconstruction algorithms can be roughly divided into two categories \cite{2010Fundamentals}: the analytic algorithms and the iterative algorithms. The latter is known to be able to provide better reconstruction images especially when the inverse problem \eqref{degradation} becomes ill-posed \cite{2006Iterative}.
The total variation regularization \cite{Zhang2021ipi}, TV stokes model \cite{liu2013total} and total generalized variation (TGV) \cite{niu2014sparse,2015Efficient} have been studied as the regularization and was shown effective for sparse CT reconstruction problems.
Besides, the multi-grid method has also been applied for CT reconstruction to achieve better reconstruction results \cite{Marlevi2020,Zhang2021ipi}.

Now we discuss the numerical examples of the proposed multi-grid algorithm for CT reconstruction problem. Two phantom images `Shepp-Logan' and `Forbild-gen' with the size of $512\times 512$ and $1024\times 1024$, are used to evaluate the performance. We adopt the parallel-beam geometry for both images in the experiments and set the projection numbers to be $N_p=18$ and 36.

\begin{table*}[htbp]
  \centering
  \footnotesize
  \caption{The comparison in terms of PSNR, SSIM, CPU time, the number of iterations (denoted by \#) and CPU time per iteration (denoted by CPU/I) for CT reconstruction with projection numbers of $N_p=18$ and 36, where image intensity is projected to $[0,1]$.}
  \vspace{-2mm}
    \begin{tabular}{p{5pt}|p{15pt}|p{13pt}p{18pt}|p{13pt}p{18pt}|p{13pt}p{18pt}|p{13pt}p{18pt}|p{13pt}p{18pt}|p{13pt}p{18pt}|p{13pt}p{18pt}|p{13pt}p{18pt}}
\hline
\hline
          \multicolumn{2}{c|}{Sizes}    & \multicolumn{4}{c|}{512 ($\sigma=0$)}       & \multicolumn{4}{c|}{1024 ($\sigma=0$)}      & \multicolumn{4}{c|}{512 ($\sigma=0.005$)}       & \multicolumn{4}{c}{1024 ($\sigma=0.005$)} \\
          \hline
          \multicolumn{2}{c|}{Examples}   & \multicolumn{2}{c|}{\footnotesize{Shepp-Logan}} & \multicolumn{2}{c|}{\footnotesize{Fobild-gen}} & \multicolumn{2}{c|}{\footnotesize{Shepp-Logan}} & \multicolumn{2}{c|}{\footnotesize{Fobild-gen}} & \multicolumn{2}{c|}{\footnotesize{Shepp-Logan}} & \multicolumn{2}{c|}{\footnotesize{Fobild-gen}} & \multicolumn{2}{c|}{\footnotesize{Shepp-Logan}} & \multicolumn{2}{c}{\footnotesize{Fobild-gen}} \\
          \hline
        \multicolumn{2}{c|}{Methods}    &\scriptsize{MGMC} & \scriptsize{MGTV} &\scriptsize{MGMC} & \scriptsize{MGTV}& \scriptsize{MGMC} & \scriptsize{MGTV} &\scriptsize{MGMC} & \scriptsize{MGTV} & \scriptsize{MGMC} & \scriptsize{MGTV}& \scriptsize{MGMC} & \scriptsize{MGTV}& \scriptsize{MGMC} & \scriptsize{MGTV} & \scriptsize{MGMC} & \scriptsize{MGTV}  \\
          \hline
    \multirow{5}[0]{*}{18} & PSNR  & \textbf{36.85} & 34.08 & \textbf{33.01} & 31.38 & \textbf{32.42} & 30.88 & \textbf{28.52} & 27.18 & \textbf{29.37} & 28.52 & \textbf{24.67} & 24.18 & \textbf{29.37} & 28.99 & \textbf{24.48} & 23.57 \\
          & SSIM  & \textbf{0.9901} & 0.9881 & \textbf{0.9801} & 0.9746 & \textbf{0.9044} & 0.8947 & \textbf{0.8574} & 0.8361 & \textbf{0.8846} & 0.8508 & \textbf{0.8579} & 0.7126 & \textbf{0.8949} & 0.8735 & \textbf{0.8818} & 0.8715 \\
          & CPU   & 68.65 & \textbf{51.72} & 79.05 & \textbf{68.74} & \textbf{287.61} & 300.24 & \textbf{350.64} & 480.24 & \textbf{57.19} & 60.52 & \textbf{65.27} & 80.36 & \textbf{325.45} & 394.28 & \textbf{380.48} & 452.51 \\
          & \#    & 104   & \textbf{89} & 124   & \textbf{112} & 134   & \textbf{127} & 134   & \textbf{126} & 114   & \textbf{96} & 130   & \textbf{120} & 121   & \textbf{105} & 132   & \textbf{125} \\
          & CPU/I & 0.4782 & \textbf{0.3612} & 0.474 & \textbf{0.4647} & \textbf{1.9769} & 2.1793 & \textbf{2.6016} & 3.6585 & \textbf{0.5016} & 0.6304 & \textbf{0.5021} & 0.6696 & \textbf{2.6896} & 3.7551 & \textbf{2.8824} & 3.4041 \\
          \hline
    \multirow{5}[0]{*}{36} & PSNR  & \textbf{44.66} & 43.18 & \textbf{38.81} & 37.8  & \textbf{39.24} & 38.35 & \textbf{37.24} & 36.15 & \textbf{33.09} & 31.85 & \textbf{28.22} & 26.78 & \textbf{31.72} & 29.18 & \textbf{28.97} & 26.17 \\
          & SSIM  & \textbf{0.9962} & 0.9911 & \textbf{0.993} & 0.991 & \textbf{0.9914} & 0.9899 & \textbf{0.9854} & 0.9749 & \textbf{0.8796} & 0.8506 & \textbf{0.8579} & 0.8391 & \textbf{0.9245} & 0.9024 & \textbf{0.8546} & 0.8355 \\
          & CPU   & 90.65 & \textbf{73.23} & 98.85 & \textbf{77.23} & \textbf{376.43} & 454.75 & \textbf{486.42} & 654.75 & \textbf{80.54} & 90.23 & \textbf{89.85} & 110.25 & \textbf{345.52} & 425.45 & \textbf{454.46} & 500.58 \\
          & \#    & 116   & \textbf{108} & 134   & \textbf{128} & 145   & \textbf{139} & 168   & \textbf{151} & 128   & \textbf{125} & 151   & \textbf{142} & 138   & \textbf{131} & 152   & \textbf{148} \\
          & CPU/I & 0.6197 & \textbf{0.5188} & 0.594 & \textbf{0.4916} & \textbf{2.4745} & 3.1462 & \textbf{2.7804} & 4.2162 & \textbf{0.6291} & 0.7218 & \textbf{0.5956} & 0.7764 & \textbf{2.5037} & 3.2477 & \textbf{2.9767} & 3.3822 \\
          \hline
          \hline
    \end{tabular}%
   \label{CT_PSNR}%
\end{table*}%

\begin{figure}[t]
\centering
\vspace{-2.5mm}
\subfloat[MGTV($27.18$dB)]{\includegraphics[width=0.2\textwidth]{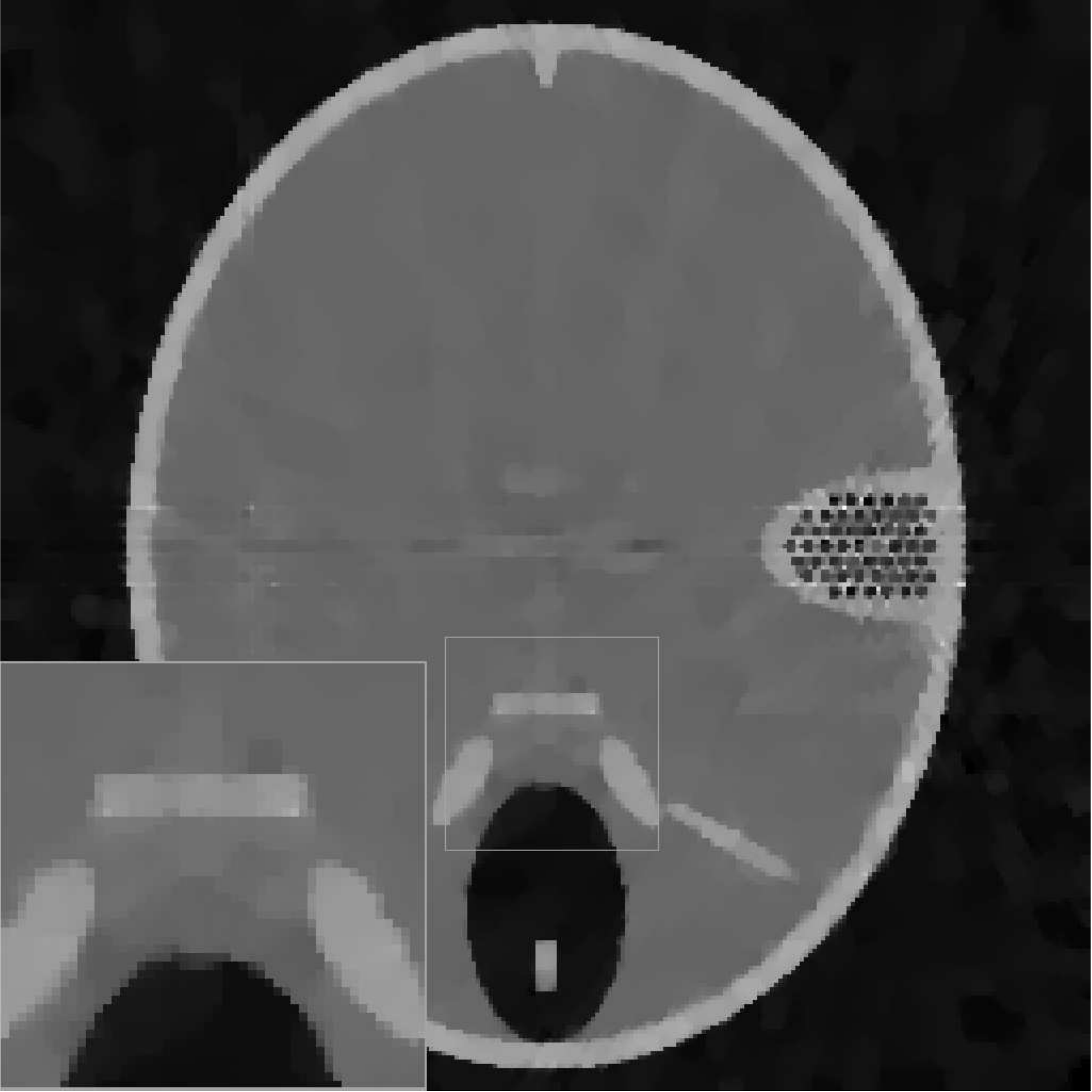}}
\hspace{-1mm}
\subfloat[MGMC($28.52$dB)]{\includegraphics[width=0.2\textwidth]{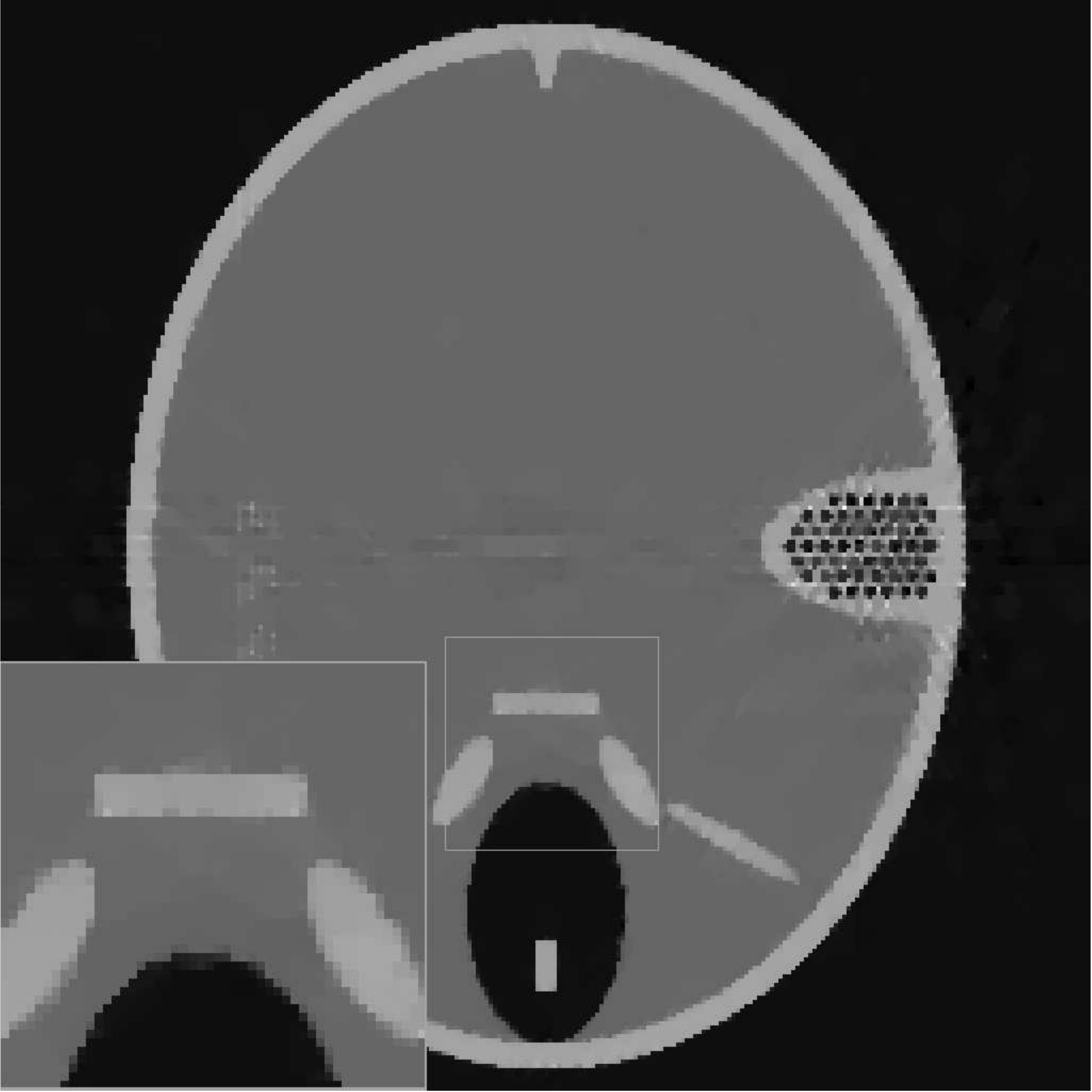}}

\vspace{-2mm}
\subfloat[MGTV($26.17$dB)]{\includegraphics[width=0.2\textwidth]{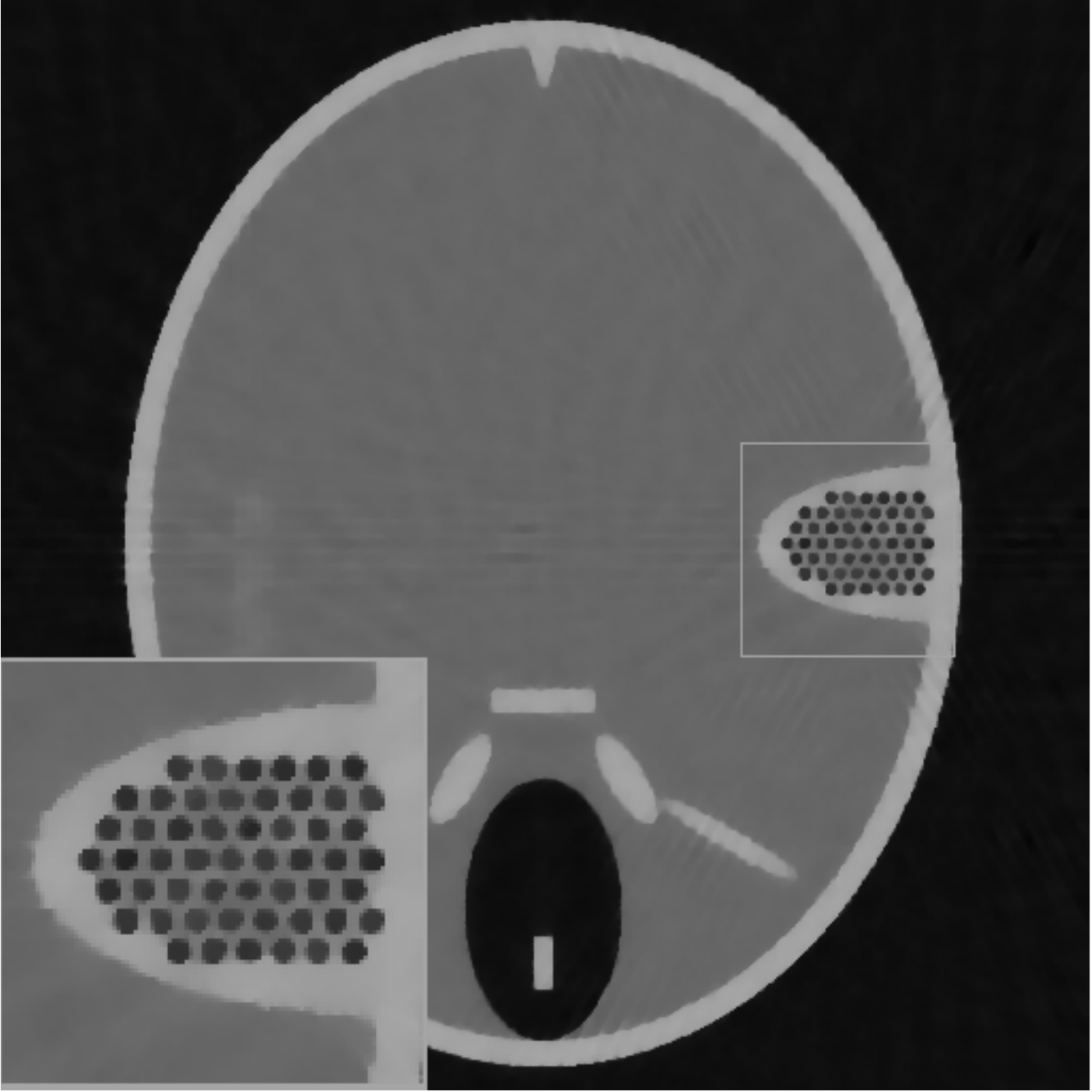}}
\hspace{-1mm}
\subfloat[MGMC($28.97$dB)]{\includegraphics[width=0.2\textwidth]{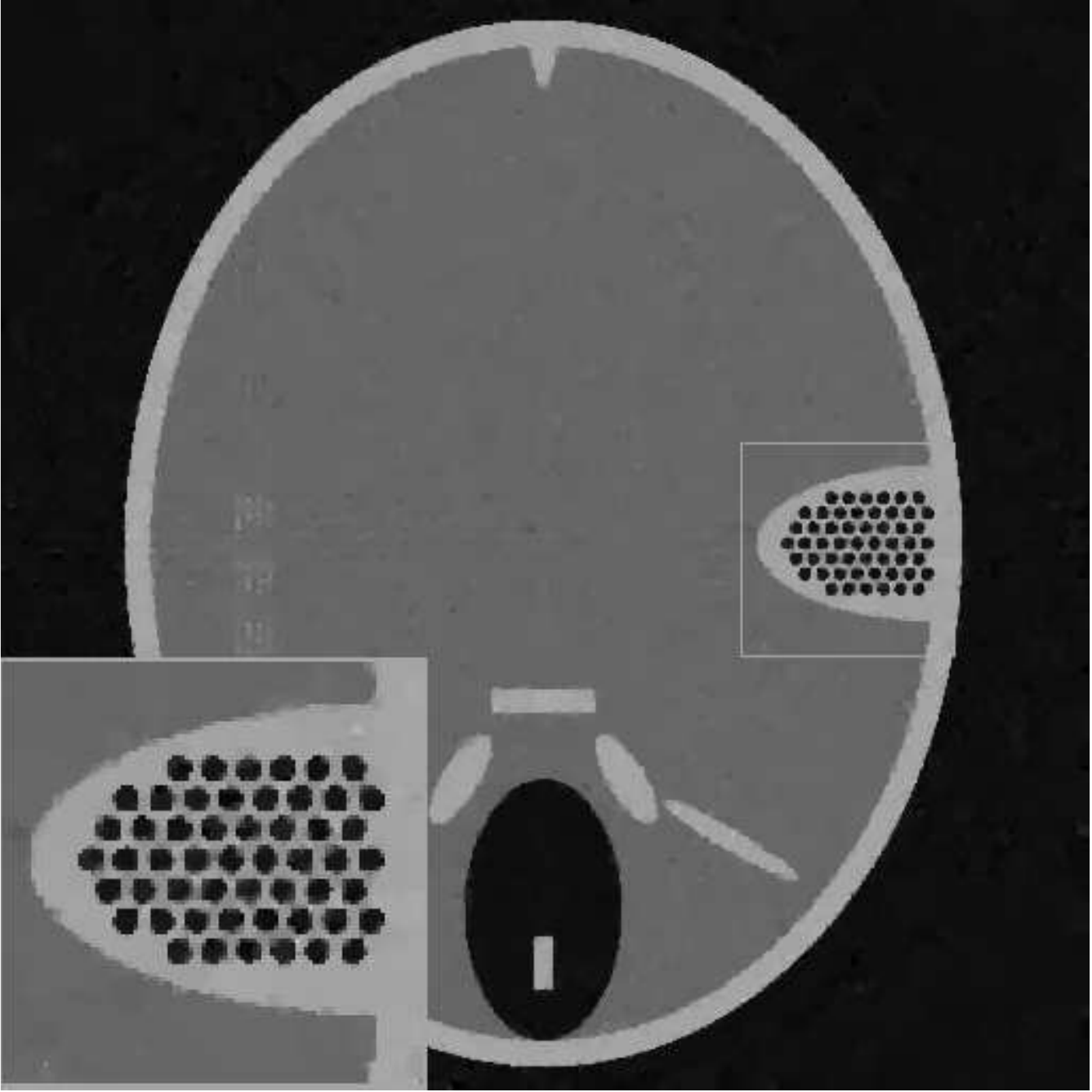}}
\caption{The comparison results on 'Forbild-gen' with the size of $1024\times1024$ and projection numbers be $N_p=18$, while the second row is the result with noise level $0.005$ and the projection numbers be $N_p=36$.}
\label{CT_figure}
\end{figure}

In what follows, we evaluate the effectiveness and efficiency of the proposed multi-grid method by comparing it with the total variation model \cite{Zhang2021ipi}, which is also implemented by the  multi-grid method.
The implementation details are described as follows
\begin{enumerate}
\item[1)] The multi-grid total variation model (denoted by MGTV) \cite{Zhang2021ipi}: The parameters are set as  $\alpha=3.5\times10^{-5}$ and $2\times10^{-5}$ for projection number $N_p=18$ and 36, respectively, and $\beta=10^{-6}$ for noiseless experiments. We set $\alpha=3\times10^{-4}$ and $5\times10^{-4}$ for $N_p=18$ and 36 on the noise level $\sigma=0.005$, where the image intensity is projected to $[0,1]$.
\item[2)] The multi-grid mean curvature model (denoted by MGMC): The parameters are set as $\alpha=4\times10^{-3}$ and $3\times10^{-3}$ for projection number $N_p=18$ and 36, respectively. We set $\alpha=1\times10^{-4}$ and $2\times10^{-4}$ for $N_p=18$ and 36 on the noise level $\sigma=0.005$.
\end{enumerate}
Both multi-grid methods are stopped using the relative error of the numerical energy \eqref{reerr} for $\epsilon=10^{-4}$ and the stopping criteria for the linear system \eqref{linear equation} is when the iteration number reaches the maximum iteration number of $10$. The number of layers is set as $J=4$ for both algorithms.

The comparison results of PSNR, SSIM,  CPU time, the number of iterations, and CPU time per iteration are all recorded in Table \ref{CT_PSNR}. For different combinations of images and projection numbers, our MGMC always gives higher PSNR and SSIM than MGTV, which benefits from the strong priors of the curvature regularization.
In terms of computational efficiency, because our local minimization problem has an analytical solution, the computational efficiency of our MGMC is quite high, which is verified by the CPU time per iteration. On the other hand, it requires solving a nonlinear PDE on the local problems for the TV regularization model, which is time consuming. The advantages of our method become stronger for the large-scale problems that our curvature regularization model performs faster than TV model.
Finally, we present the selective reconstruction results in Fig. \ref{CT_figure}. As can be seen, our MGMC outperforms the MGTV method, which produces homogeneous results with fine details and small structures.

\begin{table*}[htbp]
  \centering
  \caption{The comparison in terms of PSNR, SSIM, CPU time, the number of iterations (denoted by \#) and CPU time per iteration (denoted by CPU/I) among different methods for compressed sensing MRI reconstruction problems, where the zero-mean Gaussian noise with noise level of $\sigma =10$ is introduced into both data.}
    \begin{tabular}{c|c|ccccc|ccccc}
    \hline
    \hline
         \multirow{2}[0]{*}{Images}  & \multirow{2}[0]{*}{Methods}       & \multicolumn{5}{c|}{Cartersian sampling ($20.06\%$)}  & \multicolumn{5}{c}{Radial sampling ($12.65\%$)} \\
     \cline{3-12}
          &       & \multicolumn{1}{l}{CPU} & \multicolumn{1}{l}{$\#$} & \multicolumn{1}{l}{CPU/I} & \multicolumn{1}{l}{PSNR} & \multicolumn{1}{l|}{SSIM} & \multicolumn{1}{l}{CPU} & \multicolumn{1}{l}{$\#$} & \multicolumn{1}{l}{CPU/I} & \multicolumn{1}{l}{PSNR} & \multicolumn{1}{l}{SSIM} \\
    \hline
    \multirow{4}[0]{*}{Brain} & TV    & 19.47 & 4993  & 0.0039 & 22.91 & 0.6401 & 35.57 & 7412  & 0.0041 & 30.42 & 0.9419 \\
          & TGVST   & 24.32 & 191   & 0.1272 & 23.59 & 0.7842 & 15.11 & 151   & 0.1159 & 31.36 & 0.9615 \\
          & BM3D-MRI   & 24.75 &33   & 0.7487 & 23.72 & 0.6415 & 28.51 & 38   & 0.7500  & 31.47 & 0.9561 \\
          & MGMC  & \textbf{15.14} & 174   & 0.087 & \textbf{24.62} & \textbf{0.7855} & \textbf{12.78} & 145   & 0.079 & \textbf{31.85} & \textbf{0.9674} \\
    \hline
    \multirow{4}[0]{*}{Foot} & TV   & 25.95 & 5768  & 0.0045 & 25.96 & 0.8443 & 29.57 & 7214  & 0.0043 & 27.9  & 0.8466 \\
          & TGVST   & 23.27 & 185   & 0.1243 & 27.29 & 0.8556 & 22.34 & 160   & 0.1396 & 28.57 & 0.8546 \\
         & BM3D-MRI   & 29.41 & 39   & 0.7514  & 27.11 & 0.8547 & 28.45 & 37   & 0.7546  & 28.63 & 0.8389 \\
          & MGMC  & \textbf{13.45} & 149   &  0.085 & \textbf{27.63} & \textbf{0.8559} & \textbf{15.54} & 166   & 0.087 & \textbf{29.12} & \textbf{0.8569} \\
    \hline
    \hline
    \end{tabular}%
  \label{MRI_PSNR}%
\end{table*}%

\begin{figure*}[!htbp]
\centering
\subfloat{\includegraphics[width=0.12\textwidth]{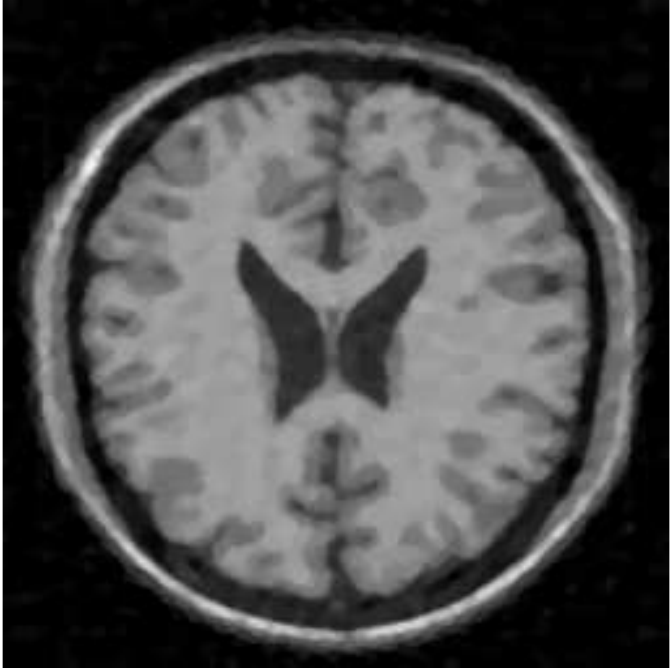}}
\hspace{-1mm}
\subfloat{\includegraphics[width=0.12\textwidth]{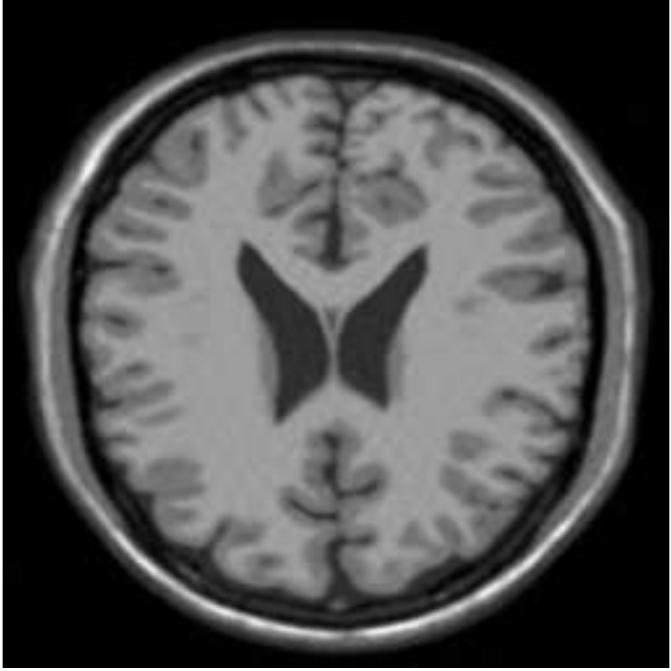}}
\hspace{-1mm}
\subfloat{\includegraphics[width=0.12\textwidth]{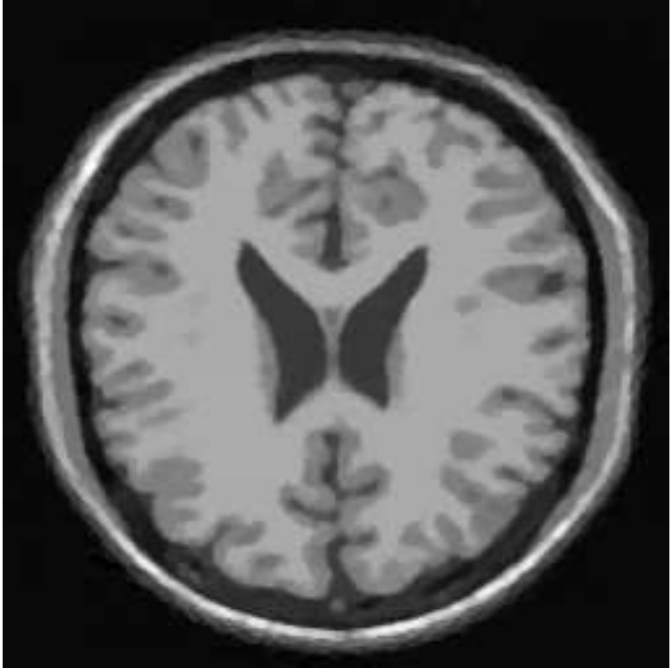}}
\hspace{-1mm}
\subfloat{\includegraphics[width=0.12\textwidth]{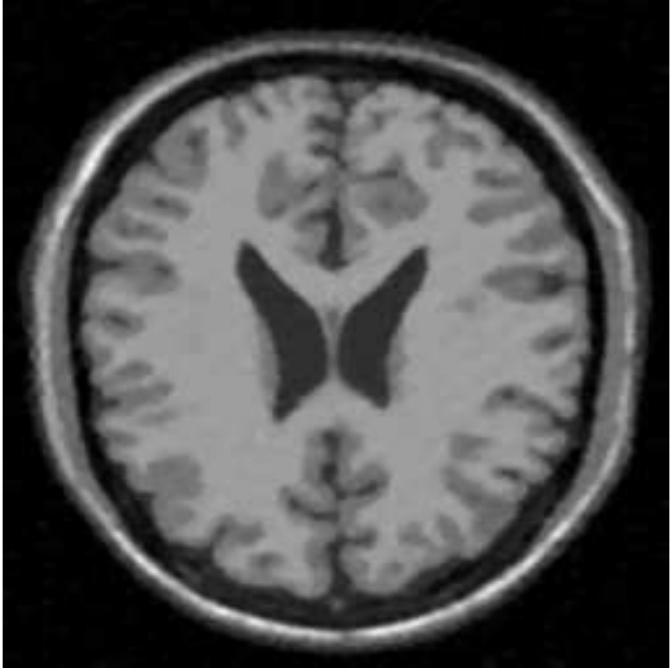}}
\hspace{-1mm}
\subfloat{\includegraphics[width=0.12\textwidth]{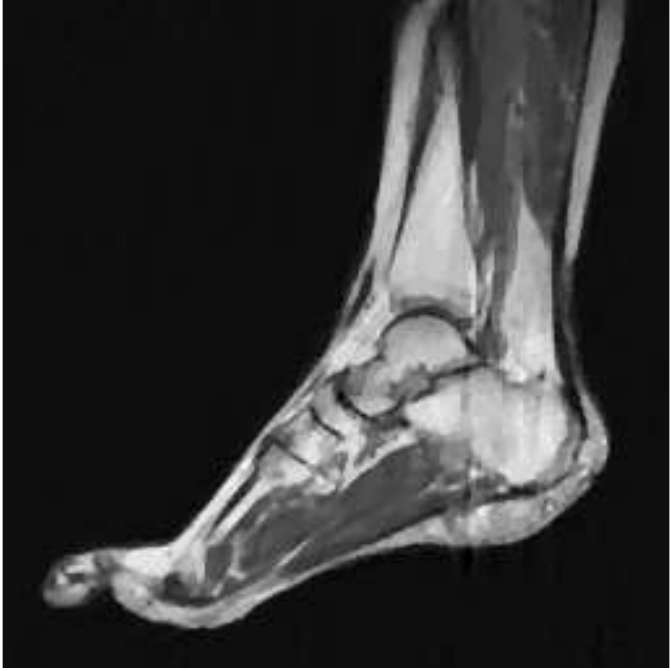}}
\hspace{-1mm}
\subfloat{\includegraphics[width=0.12\textwidth]{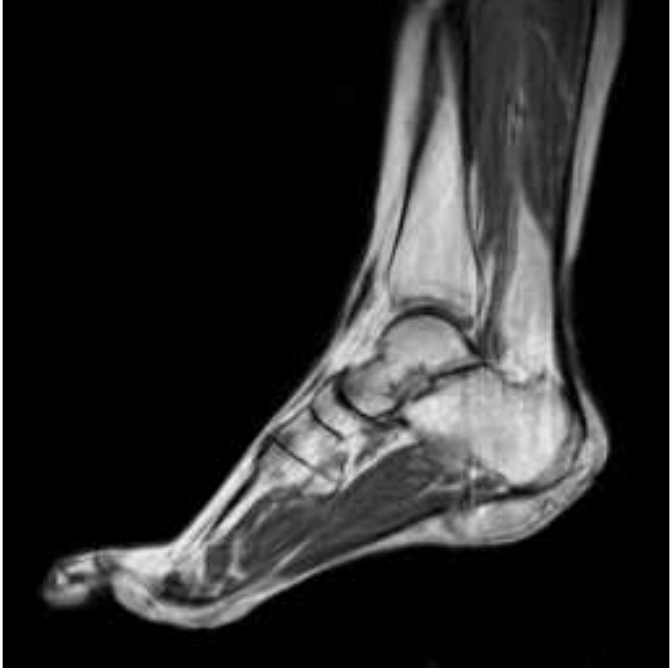}}
\hspace{-1mm}
\subfloat{\includegraphics[width=0.12\textwidth]{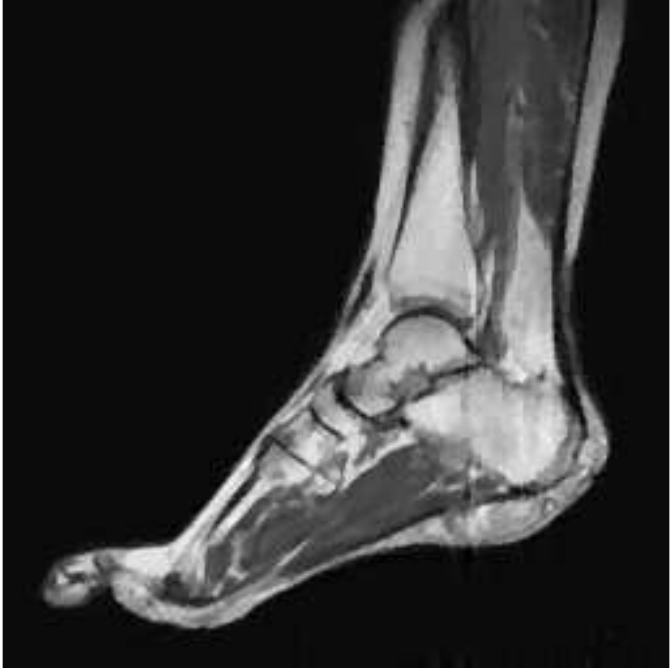}}
\hspace{-1mm}
\subfloat{\includegraphics[width=0.12\textwidth]{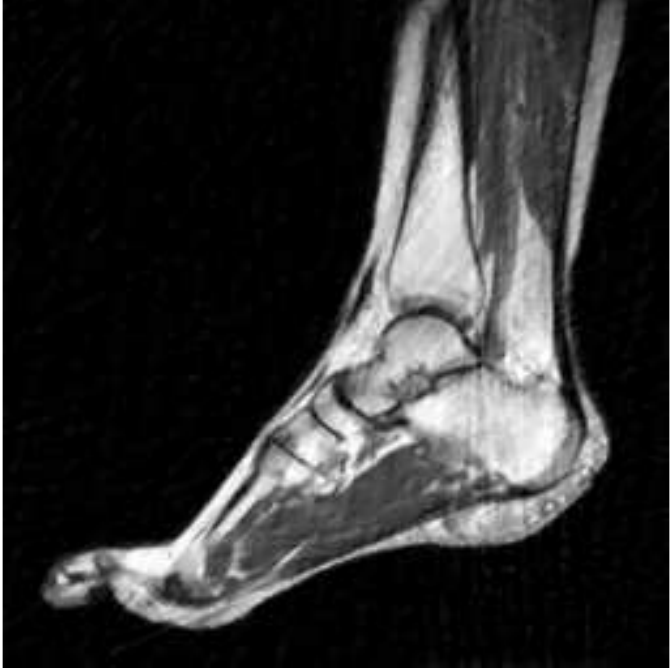}}

\vspace{-2mm}
\setcounter{subfigure}{0}
\subfloat[\footnotesize{TV}]{\includegraphics[width=0.12\textwidth]{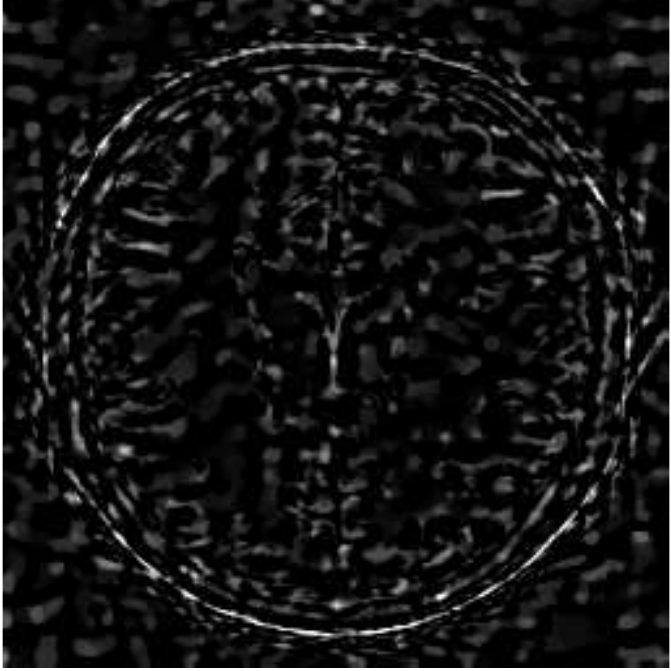}}
\hspace{-1mm}
\subfloat[\footnotesize{TGVST}]{\includegraphics[width=0.12\textwidth]{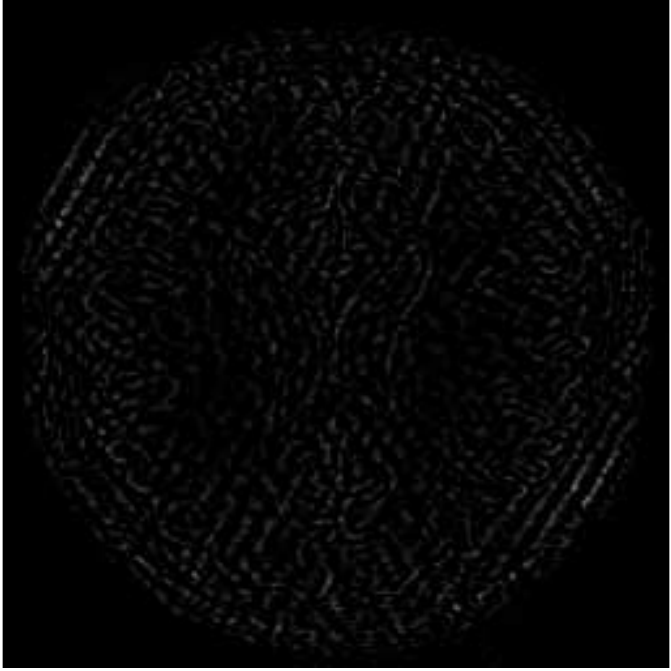}}
\hspace{-1mm}
\subfloat[\footnotesize{BM3D-MRI}]{\includegraphics[width=0.12\textwidth]{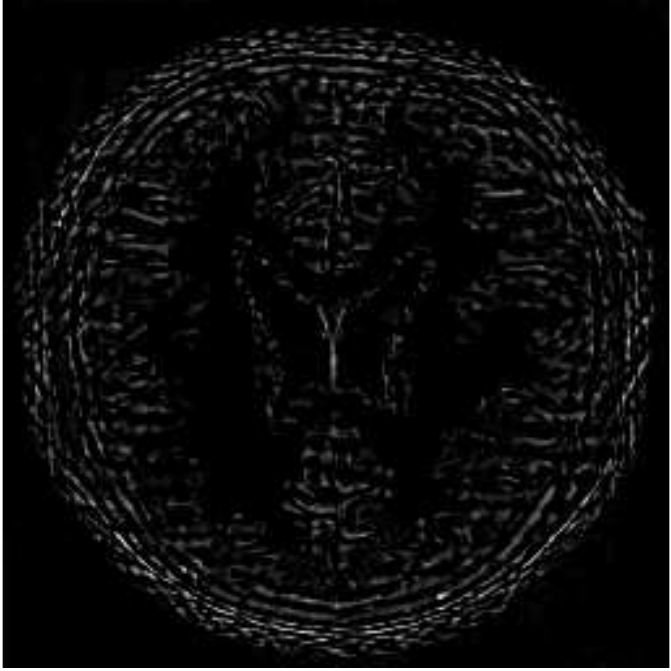}}
\hspace{-1mm}
\subfloat[\footnotesize{MGMC}]{\includegraphics[width=0.12\textwidth]{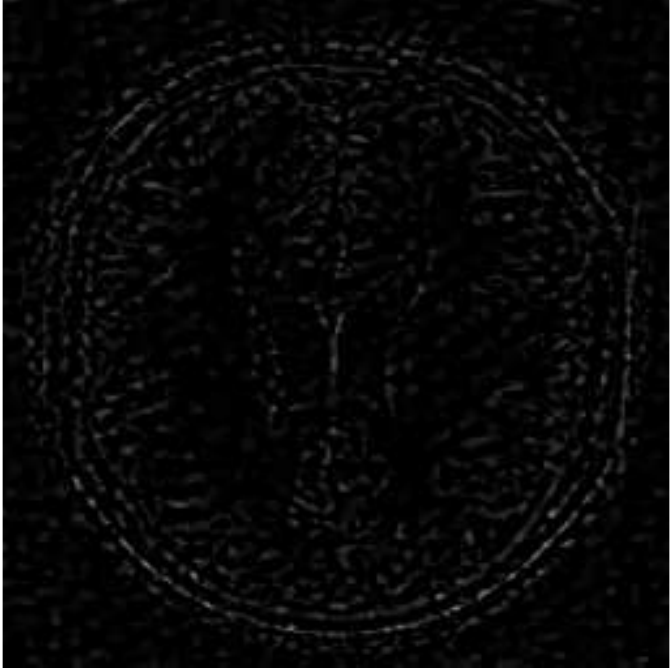}}
 \hspace{-1mm}
\subfloat[\footnotesize{TV}]{\includegraphics[width=0.12\textwidth]{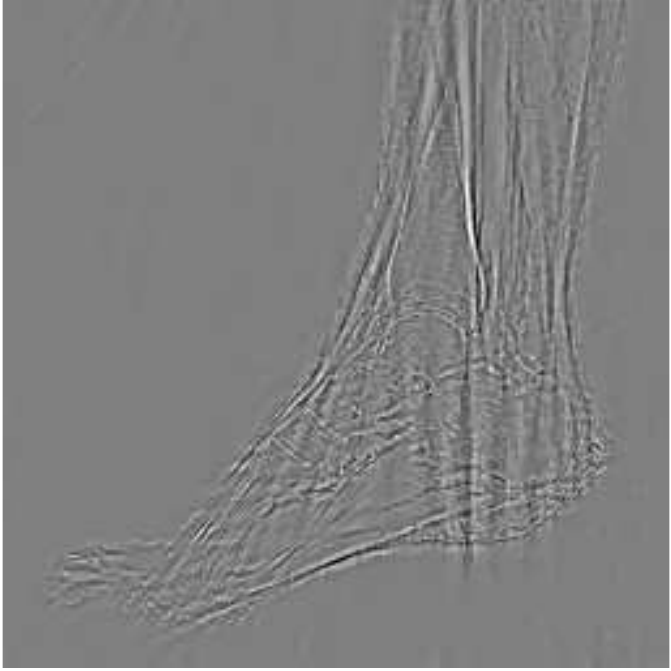}}
\hspace{-1mm}
\subfloat[\footnotesize{TGVST}]{\includegraphics[width=0.12\textwidth]{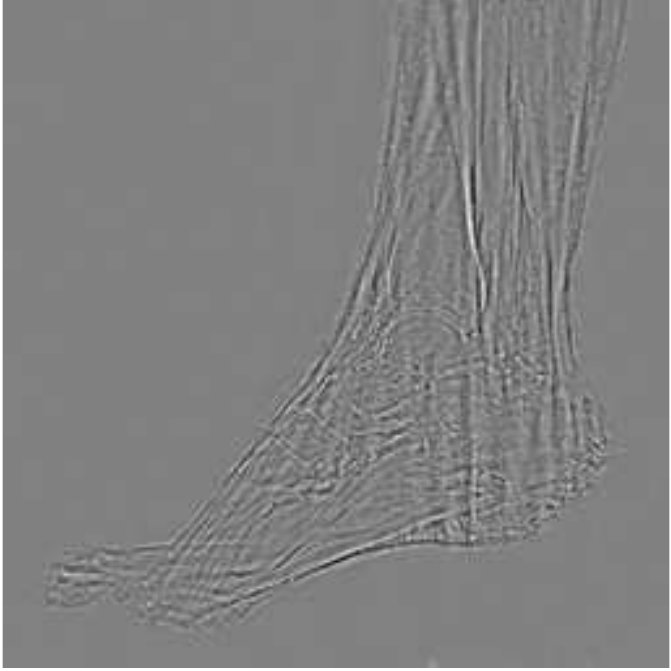}}
\hspace{-1mm}
\subfloat[\footnotesize{BM3D-MRI}]{\includegraphics[width=0.12\textwidth]{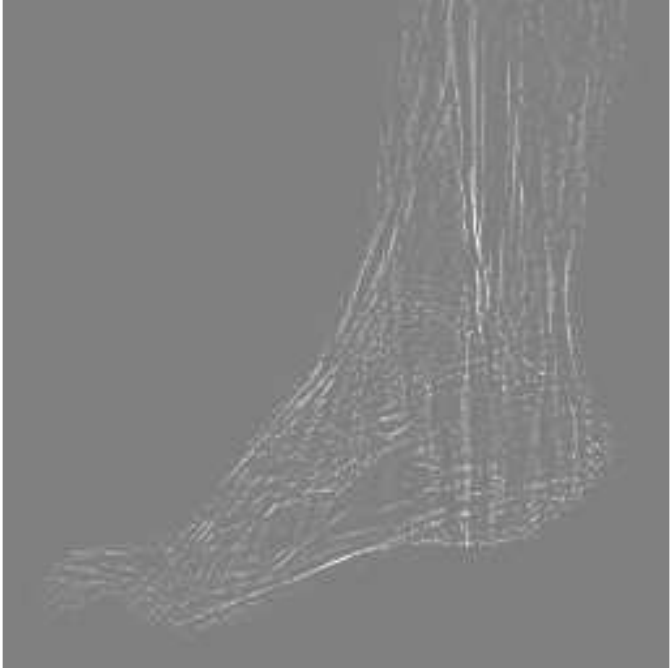}}
\hspace{-1mm}
\subfloat[\footnotesize{MGMC}]{\includegraphics[width=0.12\textwidth]{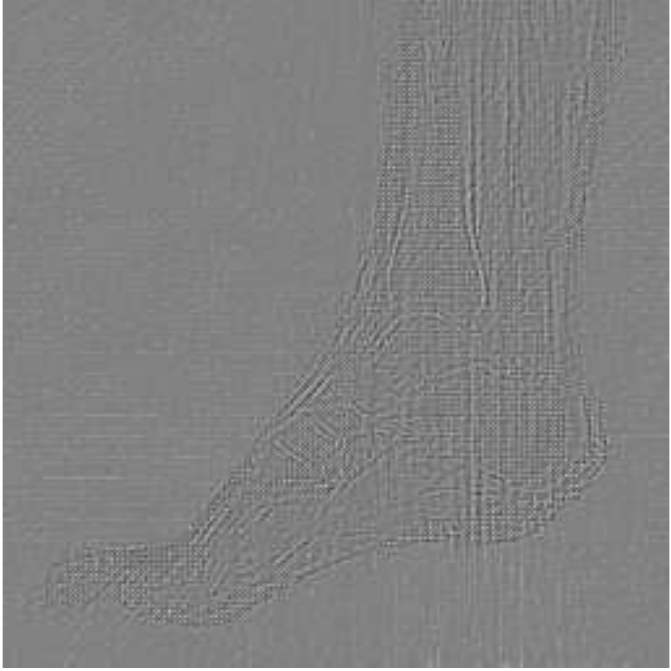}}
\caption{MRI reconstruction results and residual images of the brain image and foot image under $12.65\%$ radial sampling patterns and Gaussian noise of $\sigma=10$. Note that the residual images are displayed in $[0,0.2]$ and $[-0.4,0.4]$ for brain and foot images, respectively.}\label{MRI_figure}
\end{figure*}

\subsection{MRI reconstruction}
Similarly, different higher-order regularization terms have been used for compressed sensing MRI reconstruction problems, such as total generalized variation and shearlet transform (TGVST) \cite{guo2014new}, BM3D-MRI \cite{eksioglu2016decoupled}, Euler's elastica \cite{yashtini2020euler}, and nonlocal elastica regularization \cite{Yan2020} etc. These methods can effectively recover the missing details and preserve geometric information.
We also apply the MGMC method for compressed sensing MRI problem, where $A$ becomes a composite operator defined as $A=\mathcal{P}\mathcal{F}$ with the selection operator $\mathcal{P}$ and Fourier transform $\mathcal{F}$.

In the following part, we use two MR images as examples, one brain image and one foot image of size $256\times 256$.
Both Cartersian sampling pattern and radial sampling pattern are chosen for evaluation.
We also introduced the zero-mean complex Gaussian noise with the standard variation of $\sigma=10$ into the under-sampled data.
The performance of MGMC is compared with state-of-the-art variational methods including  TGVST \cite{guo2014new} and BM3D method \cite{eksioglu2016decoupled}, the implementation detail of which are presented as follows:
\begin{enumerate}
\item[1)] TV \cite{chambolle2011first}: The TV regularization model was solved by primal dual method. The step size is given as $\tau=1/(2L_F)$ and $\sigma=L_F/L^2$ for the primal and dual variable, respectively, where $L=\|\nabla\|$ and $L_F$ is the Lipschitz constant of $F(u)=\|\mathcal{P}\mathcal{F}u-f\|_2^2$. The regularization parameter is set as $\alpha=3\times10^{-3}$ for both under-sampled patterns.
\item[2)] TGVST \cite{guo2014new}:  We implement the TGVST algorithm
with same parameters as the ones used in the original paper such as $\beta=10^3$, $\lambda=0.01$, $\alpha_0$ is raised from $10^{-3}$ to $10^{-2}$ and $\alpha_1$ is fixed as $10^{-3}$.
\item[3)] BM3D-MRI \cite{eksioglu2016decoupled}: The range of parameter for the observation fidelity is $\lambda_{BM3D}\in [0,5]$. The total iteration number is set as 50 to balance the performance and efficiency.
Both the iteration number and relative error $\epsilon=1\times 10^{-4}$ are used as the terminating conditions.
\item[4)] MGMC: The parameters are set as $\alpha=5\times10^{-3}$ for both under-sampled patterns.
\end{enumerate}

Table \ref{MRI_PSNR} displays the PSNR, SSIM, and CPU time of the comparison methods. As illustrated, all high-order regularization methods produce higher PSNR and SSIM than TV model. And our MGMC outperforms the TGVST and BM3D-MRI in terms of PSNR, SSIM and CPU time.
We also exhibit reconstructed images and residual images in Figure \ref{MRI_figure}, which are consistent with quantitative results in Table \ref{MRI_PSNR}. It can be observed that less structural information is presented in the residual images obtained by our MGMC method.

\section{Conclusion}\label{section5}

We proposed an efficient multi-grid algorithm for solving the curvature-based minimization problems that rely on the piecewise constant basic spanned subspace correction. The original minimization was then transferred into a series of local problems from the fine layer to the coarse layer, each of which was solved by the forward-backward splitting scheme with a convergence guarantee. More importantly, there existed analytical solutions to the sub-minimization problems on local patches, which can be solved very efficiently. We also applied the non-overlapping domain decomposition method on each layer to increase the parallelism for improving computational efficiency. Furthermore, we implemented the proposed algorithms by GPU computation to deal with large-scale image processing tasks.
Comparative experiments on image denoising and reconstruction problems demonstrate the efficient performance of the proposed method by comparing it with several advanced denoising methods.

Although we proved the energy diminishing of the pixel-wise minimization problem \eqref{local_MSM}, it is difficult to show the fine layer problem converges to a minimizer of \eqref{MCM} due to the nonconvexity of the curvature energy. As far as we know, the convergence analysis of curvature regularization models has made some achievements, mainly for the Euler's elastica regularization model. In \cite{BrediesPW15,ChambolleP19}, convex relaxation of elastica energy via functional lifting was studied to establish numerical algorithms with convergence guarantee. He, Wang and Chen \cite{he2021penalty} proposed a penalty relaxation algorithm with the theoretical guarantee to find a stationary point of Euler's elastica model. However, the convergent algorithms for solving the mean curvature and Gaussian curvature energies are still very limited, which we would like to study in the future.
Other future works include developing more efficient algorithms for curvature-related minimization models and expanding the applications of curvature regularization, e.g., improving the robustness of deep neural network models \cite{Moosavi-Dezfooli_2019_CVPR,guo2020connections,Singla_2021_ICCV}.

\section*{Acknowledgment}
The authors would like to thank the associated editor and the anonymous reviewers for providing us numerous valuable suggestions to revise the paper.

\appendices

{\appendix[Proof of Lemma 1]
\begin{proof}
We prove the lemma followed the Lemma 1 in \cite{duchi2009efficient}.
The first-order optimality condition of Algorithm \ref{sub_alg} gives
\[c_{t+1}=c_t-\eta_t\partial f(c_{t+\frac 12}) -\eta_t\nabla r(c_{t+1}).\]
The convexity of $r(c)$ implies that for any $\widetilde{c}$
\begin{equation}\label{convex property}
r(\widetilde{c})\geq r(c_{t+1})+\langle \nabla r(c_{t+1}),\widetilde{c}-c_{t+1} \rangle.
\end{equation}
Since there is $\|\nabla r(c)\|^2\leq G^2$ and $\|\partial f(c)\|^2\leq D^2$, we can obtain the following inequality from the Cauchy-Shwartz inequality
\begin{equation}\label{cauchy}
\small
\begin{split}
\langle \nabla r(c_{t+1}),c_{t+1}-c_t \rangle&=\langle \nabla r(c_{t+1}),-\eta_t\partial f(c_{t+\frac 12})-\eta_t\nabla r(c_{t+1})  \rangle\\
&\leq   \eta_t (G^2+GD).
\end{split}
\end{equation}
By expanding the squared norm of the difference between $c_{t+1}$ and $\widetilde{c}$, it gives
\begin{small}
\begin{equation*}
\begin{split}
&\|c_{t+1}-\widetilde{c}\|^2= \|c_t-\eta_t\partial f(c_{t+\frac 12}) -\eta_t\nabla r(c_{t+1})-\widetilde{c}\|^2\\
&=\|c_t-\widetilde{c}\|^2-2\eta_t\langle\partial  f(c_{t+\frac 12}), c_t-\widetilde{c}\rangle +\|\eta_t\partial f(c_{t+\frac 12})+\eta_t\nabla r(c_{t+1})\|^2   \\
&-2\eta_t[\langle\nabla r(c_{t+1}),c_{t+1}-\widetilde{c} \rangle -\langle\nabla r(c_{t+1}),c_{t+1}-c_t \rangle].\\
\end{split}
\end{equation*}
\end{small}We now use \eqref{convex property} and \eqref{cauchy} to get
\begin{equation}\label{data estimation}
\small
\begin{split}
\|c_{t+1}&-\widetilde{c}\|^2\leq \|c_t-\widetilde{c}\|^2-2\eta_t r(c_{t+1})+2\eta_tr(\widetilde{c})\\
&+\eta_t^2(3G^2+D^2+4GD)+2\eta_t\langle\partial f(c_{t+\frac 12}), \widetilde{c}-c_t\rangle. \\
\end{split}
\end{equation}
By Proposition \ref{Bernstein}, we have $f(c_{t+\frac 12})\leq f(c)$.
Relying on the definition of $\partial f(c_{t+\frac 12})$, we can estimate the last term
\begin{equation}\label{last estimation}
\small
\begin{split}
&\langle\partial f(c_{t+\frac 12}), \widetilde{c}-c_t\rangle =\langle\partial f(c_{t+\frac 12}), \widetilde{c}-c_{t+\frac 12}\rangle+ \langle \partial f(c_{t+\frac 12}),c_{t+\frac 12}-c_t\rangle\\
&\leq  \Big\langle \frac{f(c_{t+\frac 12})-f(\widetilde{c}) }{c_{t+\frac 12}-\widetilde{c}}, \widetilde{c}-c_{t+\frac 12} \Big\rangle+ \Big\langle \frac{f(c_{t+\frac 12})-f(c_t) }{c_{t+\frac 12}-c_t},c_{t+\frac 12}-c_t \Big \rangle \\
&\leq f(\widetilde{c})-f(c_{t}).
\end{split}
\end{equation}
By substituting \eqref{last estimation} to  \eqref{data estimation}, we obtain
\begin{equation*}
\begin{split}
 2\eta_t \Big(f(c_{t})+r(c_{t+1})&-J(\widetilde{c})\Big)\leq \|c_{t}- \widetilde{c}\|^2 \\
&- \|c_{t+1}- \widetilde{c}\|^2+ \eta_t^2(5G^2+3D^2).
\end{split}
\end{equation*}
 which completes the proof.
\end{proof}
}



\begin{thebibliography}{10}
\providecommand{\url}[1]{#1}
\csname url@samestyle\endcsname
\providecommand{\newblock}{\relax}
\providecommand{\bibinfo}[2]{#2}
\providecommand{\BIBentrySTDinterwordspacing}{\spaceskip=0pt\relax}
\providecommand{\BIBentryALTinterwordstretchfactor}{4}
\providecommand{\BIBentryALTinterwordspacing}{\spaceskip=\fontdimen2\font plus
\BIBentryALTinterwordstretchfactor\fontdimen3\font minus
  \fontdimen4\font\relax}
\providecommand{\BIBforeignlanguage}[2]{{%
\expandafter\ifx\csname l@#1\endcsname\relax
\typeout{** WARNING: IEEEtran.bst: No hyphenation pattern has been}%
\typeout{** loaded for the language `#1'. Using the pattern for}%
\typeout{** the default language instead.}%
\else
\language=\csname l@#1\endcsname
\fi
#2}}
\providecommand{\BIBdecl}{\relax}
\BIBdecl

\bibitem{2017Augmented}
W.~Zhu, X.~C. Tai, and T.~Chan, ``Augmented lagrangian method for a mean
  curvature based image denoising model,'' \emph{Inverse Problems and Imaging},
  vol.~7, no.~4, pp. 1409--1432, 2017.

\bibitem{rudin1992nonlinear}
L.~I. Rudin, S.~Osher, and E.~Fatemi, ``Nonlinear total variation based noise
  removal algorithms,'' \emph{Physica D: Nonlinear Phenomena}, vol.~60, no.~4,
  pp. 259--268, 1992.

\bibitem{Yves2002}
Y.~Meyer, ``Oscillating patterns in image processing and nonlinear evolution
  equations,'' \emph{University Lecture Ser. 22, AMS}, 2002.

\bibitem{Lysaker2003Noise}
M.~Lysaker, A.~Lundervold, and X.-C. Tai, ``Noise removal using fourth-order
  partial differential equation with applications to medical magnetic resonance
  images in space and time,'' \emph{{IEEE} Transactions on Image Processing},
  vol.~12, no.~12, pp. 1579--1590, 2003.

\bibitem{Bredies2010total}
K.~Bredies, K.~Kunisch, and T.~Pock, ``Total generalized variation,''
  \emph{{SIAM} Journal on Imaging Sciences}, vol.~3, no.~3, pp. 492--526, 2010.

\bibitem{2011AA}
X.~C. Tai, J.~Hahn, and G.~J. Chung, ``A fast algorithm for {E}uler's
  {E}lastica model using augmented lagrangian method,'' \emph{SIAM Journal on
  Imaging Sciences}, vol.~4, no.~1, pp. 313--344, 2011.

\bibitem{Zhu2012image}
W.~Zhu and T.~Chan, ``Image denoising using mean curvature of image surface,''
  \emph{{SIAM} Journal on Imaging Sciences}, vol.~5, no.~1, pp. 1--32, 2012.

\bibitem{ChambolleP19}
A.~Chambolle and T.~Pock, ``Total roto-translational variation,''
  \emph{Numerische Mathematik}, vol. 142, no.~3, pp. 611--666, 2019.

\bibitem{belyaev2018adaptive}
A.~Belyaev and P.-A. Fayolle, ``Adaptive curvature-guided image filtering for
  structure+ texture image decomposition,'' \emph{IEEE Transactions on Image
  Processing}, vol.~27, no.~10, pp. 5192--5203, 2018.

\bibitem{pei2020curvature}
H.~Pei, B.~Wei, K.~Chang, C.~Zhang, and B.~Yang, ``Curvature regularization to
  prevent distortion in graph embedding,'' \emph{Advances in Neural Information
  Processing Systems}, vol.~33, pp. 20\,779--20\,790, 2020.

\bibitem{dong2020cure}
B.~Dong, H.~Ju, Y.~Lu, and Z.~Shi, ``Cure: Curvature regularization for missing
  data recovery,'' \emph{SIAM Journal on Imaging Sciences}, vol.~13, no.~4, pp.
  2169--2188, 2020.

\bibitem{2010Multigrid}
C.~Brito-Loeza and K.~Chen, ``Multigrid algorithm for high order denoising,''
  \emph{{SIAM} Journal on Imaging Sciences}, vol.~3, no.~3, pp. 363--389, 2010.

\bibitem{gong2019weighted}
Y.~Gong and O.~Goksel, ``Weighted mean curvature,'' \emph{Signal Processing},
  vol. 164, pp. 329--339, 2019.

\bibitem{2017Curvature}
Y.~Gong and I.~F. Sbalzarini, ``Curvature filters efficiently reduce certain
  variational energies,'' \emph{{IEEE} Transactions on Image Processing},
  vol.~26, no.~4, pp. 1786--1798, 2017.

\bibitem{Zhong2020image}
Q.~Zhong, K.~Yin, and Y.~Duan, ``Image reconstruction by minimizing curvatures
  on image surface,'' \emph{Journal of Mathematical Imaging and Vision},
  vol.~63, no.~1, pp. 30--55, 2021.

\bibitem{wang2022efficient}
C.~Wang, Z.~Zhang, Z.~Guo, T.~Zeng, and Y.~Duan, ``Efficient sav algorithms for
  curvature minimization problems,'' \emph{IEEE Transactions on Circuits and
  Systems for Video Technology}, 2022.

\bibitem{2009Analogue}
M.~Elsey and S.~Esedoglu, ``Analogue of the total variation denoising model in
  the context of geometry processing,'' \emph{{SIAM} Journal on Multiscale
  Modeling and Simulation}, vol.~7, no.~4, pp. 1549--1573, 2009.

\bibitem{2015Image}
C.~Brito-Loeza, K.~Chen, and V.~Uc-Cetina, ``Image denoising using the gaussian
  curvature of the image surface,'' \emph{Numerical Methods for Partial
  Differential Equations}, vol.~32, no.~3, pp. 518--527, 2015.

\bibitem{lu2016implementation}
W.~Lu, J.~Duan, Z.~Qiu, Z.~Pan, R.~W. Liu, and L.~Bai, ``Implementation of
  high-order variational models made easy for image processing,''
  \emph{Mathematical Methods in the Applied Sciences}, vol.~39, no.~14, pp.
  4208--4233, 2016.

\bibitem{2000Polyhedral}
L.~Alboul and R.~Damme, ``Polyhedral metrics in surface reconstruction: Tight
  triangulations,'' \emph{The Mathematics of Surfaces}, pp. 309--336, 1997.

\bibitem{gong2013local}
Y.~Gong and I.~F. Sbalzarini, ``Local weighted gaussian curvature for image
  processing,'' in \emph{2013 IEEE International Conference on Image
  Processing}.\hskip 1em plus 0.5em minus 0.4em\relax IEEE, 2013, pp. 534--538.

\bibitem{Gong2019mean}
Y.~Gong, ``Mean curvature is a good regularization for image processing,''
  \emph{{IEEE} Transactions on Circuits and Systems for Video Technology},
  vol.~29, no.~8, pp. 2205--2214, 2019.

\bibitem{Nash2000A}
Nash and Stephen, ``A multigrid approach to discretized optimization
  problems,'' \emph{Optimization Methods and Software}, vol.~14, no.~1, pp.
  99--116, 2000.

\bibitem{2003An}
R.~Kimmel and I.~Yavneh, ``An algebraic multigrid approach for image
  analysis,'' \emph{{SIAM} Journal on Scientific Computing}, vol.~24, no.~4,
  pp. 1218--1231, 2003.

\bibitem{Frohn2004Nonlinear}
C.~Frohn-Schauf, S.~Henn, and K.~Witsch, ``Nonlinear multigrid methods for
  total variation image denoising,'' \emph{Computing and Visualization in
  Science}, vol.~7, no.~4, pp. 199--206, 2004.

\bibitem{2006On}
T.~F. Chan and K.~Chen, ``On a nonlinear multigrid algorithm with primal
  relaxation for the image total variation minimisation,'' \emph{Numerical
  Algorithms}, vol.~41, no.~4, pp. 387--411, 2006.

\bibitem{Chen2007A}
K.~Chen and X.~C. Tai, ``A nonlinear multigrid method for total variation
  minimization from image restoration,'' \emph{Journal of Scientific
  Computing}, vol.~33, no.~2, pp. 115--138, 2007.

\bibitem{savage2005improved}
J.~Savage and K.~Chen, ``An improved and accelerated non-linear multigrid
  method for {T}otal-{V}ariation denoising,'' \emph{International Journal of
  Computer Mathematics}, vol.~82, no.~8, pp. 1001--1015, 2005.

\bibitem{Chan2006An}
T.~F. Chan and K.~Chen, ``An optimization based multilevel algorithm for total
  variation image denoising,'' \emph{Multiscale Modeling and Simulation},
  vol.~5, no.~2, pp. 615--645, 2006.

\bibitem{Zhang2021ipi}
Z.~Zhang, X.~Li, Y.~Duan, K.~Yin, and X.-C. Tai, ``An efficient multi-grid
  method for {TV} minimization problems,'' \emph{Inverse Problems and Imaging},
  vol.~15, no.~5, pp. 125--138, 2021.

\bibitem{Tai2021}
X.-C. Tai, L.-J. Deng, and K.~Yin, ``A multigrid algorithm for maxflow and
  min-cut problems with applications to multiphase image segmentation,''
  \emph{Journal of Scientific Computing}, vol.~87, no.~3, pp. 101--128, 2021.

\bibitem{2011A}
N.~Chumchob, K.~Chen, and C.~Brito-Loeza, ``A fourth-order variational image
  registration model and its fast multigrid algorithm,'' \emph{{SIAM} Journal
  on Multiscale Modeling and Simulation}, vol.~9, no.~1, pp. 89--128, 2011.

\bibitem{combettes2005signal}
P.~L. Combettes and V.~R. Wajs, ``Signal recovery by proximal forward-backward
  splitting,'' \emph{Multiscale Modeling and Simulation}, vol.~4, no.~4, pp.
  1168--1200, 2005.

\bibitem{raguet2013generalized}
H.~Raguet, J.~Fadili, and G.~Peyr{\'e}, ``A generalized forward-backward
  splitting,'' \emph{SIAM Journal on Imaging Sciences}, vol.~6, no.~3, pp.
  1199--1226, 2013.

\bibitem{chen2007nonlinear}
K.~Chen and X.-C. Tai, ``A nonlinear multigrid method for total variation
  minimization from image restoration,'' \emph{Journal of Scientific
  Computing}, vol.~33, no.~2, pp. 115--138, 2007.

\bibitem{chan2010multilevel}
R.~H. Chan and K.~Chen, ``A multilevel algorithm for simultaneously denoising
  and deblurring images,'' \emph{SIAM Journal on Scientific Computing},
  vol.~32, no.~2, pp. 1043--1063, 2010.

\bibitem{pustelnik2011parallel}
N.~Pustelnik, C.~Chaux, and J.-C. Pesquet, ``Parallel proximal algorithm for
  image restoration using hybrid regularization,'' \emph{IEEE Transactions on
  Image Processing}, vol.~20, no.~9, pp. 2450--2462, 2011.

\bibitem{tang2020compressive}
V.~H. Tang, A.~Bouzerdoum, and S.~L. Phung, ``Compressive radar imaging of
  stationary indoor targets with low-rank plus jointly sparse and total
  variation regularizations,'' \emph{IEEE Transactions on Image Processing},
  vol.~29, pp. 4598--4613, 2020.

\bibitem{duchi2009efficient}
J.~Duchi and Y.~Singer, ``Efficient online and batch learning using forward
  backward splitting,'' \emph{The Journal of Machine Learning Research},
  vol.~10, pp. 2899--2934, 2009.

\bibitem{papafitsoros2014combined}
K.~Papafitsoros and C.-B. Sch{\"o}nlieb, ``A combined first and second order
  variational approach for image reconstruction,'' \emph{Journal of
  mathematical imaging and vision}, vol.~48, no.~2, pp. 308--338, 2014.

\bibitem{2010Fundamentals}
G.~T. Herman, \emph{Fundamentals of computerized tomography}.\hskip 1em plus
  0.5em minus 0.4em\relax Fundamentals of Computerized Tomography, 2010.

\bibitem{2006Iterative}
J.~Qi and R.~M. Leahy, ``Iterative reconstruction techniques in emission
  computed tomography,'' \emph{Physics in Medicine and Biology}, vol.~51,
  no.~15, pp. 541--549, 2006.

\bibitem{liu2013total}
Y.~Liu, Z.~Liang, J.~Ma, H.~Lu, K.~Wang, H.~Zhang, and W.~Moore, ``Total
  variation-stokes strategy for sparse-view {X}-ray {CT} image
  reconstruction,'' \emph{{IEEE} transactions on medical imaging}, vol.~33,
  no.~3, pp. 749--763, 2013.

\bibitem{niu2014sparse}
S.~Niu, Y.~Gao, Z.~Bian, J.~Huang, W.~Chen, G.~Yu, Z.~Liang, and J.~Ma,
  ``Sparse-view {X}-ray {CT} reconstruction via total generalized variation
  regularization,'' \emph{Physics in Medicine and Biology}, vol.~59, no.~12,
  pp. 2997--3009, 2014.

\bibitem{2015Efficient}
J.~Chen, L.~Wang, B.~Yan, H.~Zhang, and G.~Cheng, ``Efficient and robust 3{D}
  {CT} image reconstruction based on total generalized variation regularization
  using the alternating direction method,'' \emph{Journal of X-ray science and
  Technology}, vol.~23, no.~6, pp. 683--698, 2015.

\bibitem{Marlevi2020}
D.~Marlevi, H.~Kohr, J.-W. Buurlage, B.~Gao, K.~J. Batenburg, and
  M.~Colarieti-Tosti, ``Multigrid reconstruction in tomographic imaging,''
  \emph{{IEEE} Transactions on Radiation and Plasma Medical Sciences}, vol.~4,
  no.~3, pp. 300--310, May 2020.

\bibitem{guo2014new}
W.~Guo, J.~Qin, and W.~Yin, ``A new detail-preserving regularization scheme,''
  \emph{SIAM journal on imaging sciences}, vol.~7, no.~2, pp. 1309--1334, 2014.

\bibitem{eksioglu2016decoupled}
E.~M. Eksioglu, ``Decoupled algorithm for mri reconstruction using nonlocal
  block matching model: Bm3d-mri,'' \emph{Journal of Mathematical Imaging and
  Vision}, vol.~56, no.~3, pp. 430--440, 2016.

\bibitem{yashtini2020euler}
M.~Yashtini, ``Euler’s elastica-based algorithm for parallel mri
  reconstruction using sensitivity encoding,'' \emph{Optimization Letters},
  vol.~14, no.~6, pp. 1435--1458, 2020.

\bibitem{Yan2020}
M.~Yan and Y.~Duan, ``Nonlocal elastica model for sparse reconstruction,''
  \emph{Journal of Mathematical Imaging and Vision}, vol.~62, no.~4, pp.
  532--548, 2020.

\bibitem{chambolle2011first}
A.~Chambolle and T.~Pock, ``A first-order primal-dual algorithm for convex
  problems with applications to imaging,'' \emph{Journal of mathematical
  imaging and vision}, vol.~40, no.~1, pp. 120--145, 2011.

\bibitem{BrediesPW15}
K.~Bredies, T.~Pock, and B.~Wirth, ``A convex, lower semicontinuous
  approximation of {E}uler's {E}lastica energy,'' \emph{{SIAM} Journal on
  Mathematical Analysis}, vol.~47, no.~1, pp. 566--613, 2015.

\bibitem{he2021penalty}
F.~He, X.~Wang, and X.~Chen, ``A penalty relaxation method for image processing
  using {E}uler's {E}lastica model,'' \emph{SIAM Journal on Imaging Sciences},
  vol.~14, no.~1, pp. 389--417, 2021.

\bibitem{Moosavi-Dezfooli_2019_CVPR}
S.-M. Moosavi-Dezfooli, A.~Fawzi, J.~Uesato, and P.~Frossard, ``Robustness via
  curvature regularization, and vice versa,'' in \emph{Proceedings of the
  IEEE/CVF Conference on Computer Vision and Pattern Recognition (CVPR)}, June
  2019, pp. 9070--9078.

\bibitem{guo2020connections}
Y.~Guo, L.~Chen, Y.~Chen, and C.~Zhang, ``On connections between
  regularizations for improving {DNN} robustness,'' \emph{IEEE Transactions On
  Pattern Analysis And Machine Intelligence}, vol.~43, no.~12, pp. 4469--4476,
  2020.

\bibitem{Singla_2021_ICCV}
V.~Singla, S.~Singla, S.~Feizi, and D.~Jacobs, ``Low curvature activations
  reduce overfitting in adversarial training,'' in \emph{Proceedings of the
  IEEE/CVF International Conference on Computer Vision (ICCV)}, October 2021,
  pp. 16\,423--16\,433.

\end{thebibliography}
\end{document}